\documentclass[11pt]{article}

\usepackage{amsmath, amsthm, amssymb, graphicx, enumerate, fullpage, url}
\usepackage{color}
\usepackage{verbatim}
\usepackage{mathrsfs}
\usepackage{mathtools}
\usepackage{array}

\theoremstyle{plain}
\newtheorem{thm}{Theorem}[section]
\newtheorem{theorem}[thm]{Theorem}
\newtheorem{definition}[thm]{Definition}

\newtheorem{lemma}[thm]{Lemma}
\newtheorem{corollary}[thm]{Corollary}
\newtheorem{claim}[thm]{Claim}

\theoremstyle{definition}

\newtheorem{remark}[thm]{Remark}
\newtheorem{example}[thm]{Example}

\newcommand{\Prn}{\mathrm{Pr}}

\usepackage{mdframed}
\newcommand{\putinbox}[1]{\begin{mdframed}#1\end{mdframed}}

\newcommand{\cC}{C}

\newcommand{\N}{\mathbb{N}}
\newcommand{\F}{\mathbb{F}}
\newcommand{\K}{\mathbb{K}}
\newcommand{\E}{\mathbb{E}}

\newcommand{\ba}{{\mathbf{a}}}
\newcommand{\bb}{{\mathbf{b}}}
\newcommand{\bx}{{\mathbf{x}}}
\newcommand{\bY}{{\mathbf{Y}}}
\newcommand{\by}{{\mathbf{y}}}
\newcommand{\bu}{{\mathbf{u}}}
\newcommand{\bi}{{{\mathbf{i}}}}
\newcommand{\bj}{{{\mathbf{j}}}}

\newcommand{\bX}{{\mathbf{X}}}
\newcommand{\vi}{{{\mathbf{i}}}}

\newcommand{\MULT}{\mathsf{MULT}}
\newcommand{\mult}{\mathsf{mult}}
\newcommand{\wt}{\mathsf{wt}}

\newcommand{\inparen}[1]{\left(#1\right)}

\newcommand{\calL}{\mathcal{L}}

\newcommand{\tilds}{{\tilde{s}}}

\newcommand{\ALGPRUNE}{\mathsf{PruneListFRS}}
\newcommand{\ALGPRUNEMULT}{\mathsf{PruneListMULT}}
\newcommand{\RecCand}{\mathsf{RecoverCandidates}}
\newcommand{\MultimultMain}{\mathsf{LocalListRecoverMULT}}
\newcommand{\FRS}{\mathrm{FRS}}
\newcommand{\poly}{\mathsf{poly}}
\newcommand{\polylog}{\mathsf{polylog}}
\newcommand{\vspan}{\mathsf{span}}

\newcommand{\ii}{\textbf{i}}
\renewcommand{\char}{\textsf{char}}

\newcommand{\Cool}{{$(X^q, d)$-closed}}
\newcommand{\qdim}{{\mathrm{qdim}}}
\newcommand{\MultEnc}{{\mathsf{MultEnc}}}
\newcommand{\FRSEnc}{{\mathsf{FRSEnc}}}
\newcommand{\Domain}{{\mathcal{D}}}
\def\multiset#1#2{\ensuremath{\left(\kern-.3em\left(\genfrac{}{}{0pt}{}{#1}{#2}\right)\kern-.3em\right)}}

\newcommand{\dist}{\ensuremath{\operatorname{dist}}}
\newcommand{\eps}{\varepsilon}
\renewcommand{\epsilon}{\varepsilon}

  \newcommand{\TODO}[1]{\textcolor{red}{\textbf{TODO: #1}}}
  \newcommand{\mkw}[1]{\textcolor{blue}{\footnotesize\textbf{[#1 --mary]}}}

\newcommand{\PRUNE}{\mathsf{PRUNE}}

\begin{document}

\title{Improved decoding of Folded Reed-Solomon and Multiplicity Codes}
\author{
Swastik Kopparty\thanks{Department of Mathematics and Department of Computer Science, Rutgers University. Research supported in part by NSF grants CCF-1253886 and CCF-1540634. \texttt{swastik.kopparty@gmail.com}.} 
\and  
Noga Ron-Zewi\thanks{Department of Computer Science, Haifa University. \texttt{noga@cs.haifa.ac.il}.} 
\and 
Shubhangi Saraf\thanks{Department of Mathematics and Department of Computer Science, Rutgers University. Research supported in part by NSF grants
CCF-1350572 and CCF-1540634. \texttt{shubhangi.saraf@gmail.com}.} 
\and 
Mary Wootters\thanks{Department of Computer Science and Department of Electrical Engineering, Stanford University.  \texttt{marykw@stanford.edu}.  Research supported in part by NSF grant CCF-1657049.} }
\maketitle

\begin{abstract}
In this work, we
show new and improved error-correcting properties of folded Reed-Solomon codes and multiplicity codes.
Both of these families of codes are based on polynomials over finite fields, and both have been the sources of recent advances
in coding theory.  Folded Reed-Solomon codes were the first explicit constructions of codes known to achieve list-decoding capacity; multivariate multiplicity codes were the first constructions of high-rate locally correctable codes; and univariate multiplicity codes are also known to achieve list-decoding capacity.

However, previous analyses of the error-correction properties of these codes did not yield optimal results.  In particular, in the list-decoding setting, the guarantees on the list-sizes were polynomial in the block length, rather than constant; and for multivariate multiplicity codes, local list-decoding algorithms could not go beyond the Johnson bound.

In this paper, we show that Folded Reed-Solomon codes and multiplicity codes are in fact better than previously known in the context of list-decoding and local list-decoding. More precisely, we first show that Folded RS codes achieve list-decoding capacity with {\em constant} list sizes, independent of the block length; and that high-rate univariate multiplicity codes can also be list-recovered with constant list sizes.
Using our result on univariate multiplicity codes, we show that multivariate multiplicity codes are high-rate, {\em locally} list-recoverable codes. 
Finally, we show how to combine the above results with standard tools to obtain capacity achieving locally list decodable codes with query complexity significantly lower than was known before.
\end{abstract}

\setcounter{page}{0}
\thispagestyle{empty}
\newpage

\section{Introduction}\label{sec:Intro}

An error correcting code $C \subset \Sigma^n$ is a collection of \em codewords \em $c$ of length $n$ over an alphabet $\Sigma$.
The goal in designing $C$ is to enable the recovery of a codeword $c \in C$ given a corrupted version $\tilde{c}$ of $c$, while at the same time making $C$ as large as possible.
In the classical unique decoding problem, the goal is to efficiently recover $c$ from any $\tilde{c} \in \Sigma^n$ so that $c$ and $\tilde{c}$ differ in at most $\alpha n$ places; this requires that the \em relative distance \em $\delta$ of the code
 (that is, the fraction of places on which any two codewords differ) to be at least $2\alpha$.  

Modern applications of error correcting codes, both in coding theory
and theoretical computer science, have highlighted the importance
of variants of the unique decoding problem, incuding \em list decoding, \em and \em local decoding. \em  
In list-decoding, the amount of error $\alpha$ is large enough that unique recovery of the codeword $c$ is impossible (that is, $\alpha > \delta/2$), and instead the goal is to return a short list $\mathcal{L} \subset C$ with the guarantee that $c \in \mathcal{L}$.
In local decoding, we still have $\alpha < \delta/2$, but the goal is to recover a single symbol $c_i$ of a codeword $c$, after querying not too many positions of the corrupted codeword $\tilde{c}$.  
In a variant known as \em local list-decoding, \em we seek local information about a symbol even when $\alpha > \delta/2$.
List-decoding, local decoding, and local list-decoding are important primitives in error correcting codes, with applications in coding theory, complexity theory, pseudorandomness and cryptography.

Algebraic codes have been at the heart of the study of list-decoding, local-decoding and local list-decoding.
One classical example of this is Reed-Solomon (RS) codes, whose codewords are comprised of evaluations of low-degree polynomials.\footnote{That is, a codeword of an RS code has the form $(f(x_0), f(x_1), \ldots, f(x_{n-1})) \in \mathbb{F}^n$ for some low-degree polynomial $f \in \mathbb{F}[X]$.}  In the late 1990's, Guruswami and Sudan~\cite{Sudan,GS-list-dec} gave an algorithm for efficiently list-decoding Reed-Solomon codes well beyond half the distance of the code, and this kicked off the field of algorithmic list-decoding. 
A second example is Reed-Muller (RM) codes, the multivariate analogue of Reed-Solomon codes.
The structure of Reed-Muller codes is very amenable to local algorithms: a codeword of a Reed-Muller code corresponds to a multivariate low-degree polynomial, and considering the restriction of that polynomial to a line yields a univariate low-degree polynomial, \em a.k.a. \em  a Reed-Solomon codeword.  This local structure is the basis for Reed-Muller codes being
locally testable~\cite{RubinfeldSudan} and locally decodable~\cite{Lip90, BFLS91}. Using this locality in concert with the Guruswami-Sudan algorithm leads to local list-decoding schemes~\cite{AS03,STV01} for these codes.

More recently, variants of Reed-Solomon and Reed-Muller codes have emerged to obtain improved list-decoding and local-decoding properties.  Two notable examples, which are the focus of this work, are \em Folded Reed-Solomon \em (FRS) and \em multiplicity codes. \em  Both of these constructions have led to recent advances in coding theory.  We introduce these codes informally here, and give formal definitions in Section~\ref{sec:Preliminaries}.

Folded Reed-Solomon codes, introduced by Guruswami and Rudra in~\cite{GR08_folded_RS}, are a simple variant of Reed-Solomon codes.  If the codeword of a Reed-Solomon code is $(c_0, c_2, \ldots, c_{n-1}) \in \Sigma^n$, then the folded version (with folding parameter $s$) is
\[ \left( \begin{bmatrix} c_0 \\ c_1 \\ \vdots \\ c_{s-1} \end{bmatrix},  \begin{bmatrix} c_s \\ c_{s+1} \\ \vdots \\ c_{2s - 1} \end{bmatrix}, \ldots, \begin{bmatrix} c_{n-s} \\ c_{n-s+1} \\ \vdots \\ c_{n-1}\end{bmatrix} \right) \in (\Sigma^s)^{n/s}. \]
The main property of these codes that makes them interesting is that they admit much better list-decoding algorithms~\cite{GR08_folded_RS} than the original Guruswami-Sudan algorithm: more precisely, it allows for the error tolerance $\alpha$ to be much larger for a code of the same rate,\footnote{The \em rate \em of a code $C \in \Sigma^n$ is defined as $R = \frac{1}{n} \log_{|\Sigma|}(|C|)$ and quantifies how much information can be sent using the code.  We always have $R \in (0,1)$, and we would like $R$ to be as close to $1$ as possible.}
asymptotically obtaining the {\em optimal} trade-off.

Multiplicity codes, introduced in the univariate setting by Rosenbloom and Tsfasman in~\cite{RT-m-metric}
and in the multivariate setting by Kopparty, Saraf and Yekhanin in~\cite{KSY14}, are variants of polynomial codes that also include evaluations of derivatives.  
That is, while a symbol of a RS codeword is of the form $f(x) \in \F$ for some low-degree polynomial $f \in \mathbb{F}[X]$ and some $x \in \F$, a symbol in a univariate multiplicity code codeword is of the form $(f(x), f^{(1)}(x), f^{(2)}(x), \ldots, f^{(s-1)}(x)) \in \F^s$, where $s$ is the \em multiplicity parameter. \em  Similarly, while a symbol of an RM codeword is of the form $f(\bx)$ for $\bx \in \F^m$ for some low-degree multivariate polynomial $f \in \mathbb{F}[X_1,\ldots,X_m]$, a symbol in a multivariate multiplicty code includes all partial derivatives of order less than $s$.
Multivariate multiplicity codes were shown in \cite{KSY14} to have strong locality properties, and were the first constructions known of high-rate locally decodable codes.
Meanwhile, univariate multiplicity codes were shown in~\cite{Kop15, GW13} to be list-decodable in the same parameter regime as folded Reed-Solomon codes\footnote{They were previously shown to be list-decodable up to the Johnson bound by Nielsen~\cite{Nielsen}.}, also achieving asymptotically optimal trade-off between rate and error-tolerance.

In this work, we show that Folded Reed-Solomon codes, univariate multiplicity codes, and multivariate multiplicity codes are even more powerful than was previously known in the context of list-decoding and local list-decoding.  Our motivations for this work are threefold:
\begin{enumerate}
	\item First, FRS codes and multiplicity codes are basic and natural algebraic codes, central to many recent results in coding theory (\cite{GR08_folded_RS, KSY14, Kop15, GW13, DL12, KMRS, GKORS17}, to name a few) and understanding their error-correcting properties is important in its own right.
	\item Second, by composing our new results with known techniques, we obtain capacity-achieving locally list-decodable codes with significantly improved query complexity than previously known.
	\item Third, while there have been improved constructions of list-decodable and locally list-decodable codes building on FRS and multiplicity codes (discussed more below), those constructions involve significant additional pseudorandom ingredients.
Our results give simpler constructions of capacity achieving list-decodable and locally list-decodable codes with the best known parameters.
In particular, we give the first constructions of linear\footnote{Many codes in
this paper have alphabet $\Sigma = \F_q^s$, where $\F_q$ is a finite field. For such ``vector alphabet" codes, we use the term ``linear" to 
mean ``$\F_q$-linear".} 
capacity-achieving list-decodable codes with constant alphabet size and constant output list size.
\end{enumerate}
We will state our results and contributions more precisely in Section~\ref{ssec:contributions} after setting up a bit more notation and surveying related work.

\subsection{Related work}

\paragraph{List-recoverable codes.}
While the discussion above focused on the more well-known problem of list-\emph{decoding}, in this work we actually focus on a generalization of list-decoding known as list-\emph{recovery}.
Given a code $C \subseteq \Sigma^n$, an $(\alpha, \ell, L)$-list-recovery algorithm for $C$ takes as input a sequence of lists $S_1,\ldots,S_n \subseteq \Sigma$, each of size at most $\ell$, and returns a list $\mathcal{L}$ of all of the codewords $c \in C$ so that $c_i \in S_i$ for all but an $\alpha$ fraction of the coordinates $i$; the combinatorial requirement is that $|\mathcal{L}| \leq L$.  List-decoding is the special case of list-recovery when $\ell=1$.

Both list-recovery and list-decoding have been important in coding theory, especially in theoretical computer science, for the past several decades (see~\cite{Sudan,Vad12} for overviews).  Initially, the generalization to list recovery was used as a building block towards constructions of list decodable and uniquely decodable codes \cite{GI02,GI03,GI04,GI05,KMRS,GKORS17,HRW17}, although it has since found additional applications in algorithm design~\cite{INR10,NPR12,GNPRS13}.

The Guruswami-Sudan algorithm, mentioned above, is in fact a list-recovery algorithm as well as a list-decoding algorithm, and can efficiently list-recover Reed-Solomon codes up to radius $\alpha = 1 - \sqrt{\ell\cdot R}$, with polynomial list sizes $L$; this trade-off is known as the \em Johnson bound.\em\ 
It is a classical result that there are codes that go beyond the Johnson bound while keeping the output list size polynomial in $n$, or even constant: for large alphabet sizes, the ``correct" limit (called the \em list-decoding \em or \em list-recovering capacity\em), is $\alpha = 1 - R$, provided $q$ is sufficiently larger than $\ell$, and this is achieved by uniformly random codes.   There is a big difference between $1 - \sqrt{\ell \cdot R}$ and $1 - R$, especially when $\ell > 1$.  In particular, the Guruswami-Sudan algorithm requires Reed-Solomon codes to have rate $R < 1/\ell$ to be $(\alpha, \ell, L)$-list-recoverable for nontrivial $\alpha$, while a completely random code can achieve rates arbitrarily close to $1$ (of course, without efficient decoding algorithms).  
For a decade it was open whether or not one could construct explicit codes which efficiently achieve list-decoding capacity.

In a breakthrough result, Guruswami and Rudra~\cite{GR08_folded_RS} (building on the work of Parvaresh and Vardy~\cite{PV05}) showed that the folding operation described above can make RS codes approach capacity with polynomial list-sizes.  For some time, this was the only known route to capacity-achieving codes, until it was shown in~\cite{GW13,Kop15} that univariate multiplicity codes also do the job (again, with polynomial list sizes).  
Since then there has been a great deal of work aimed at reducing the list size and alphabet size of these constructions, both of which were polynomial in $n$ (and both of which would ideally be independent of $n$).  To reduce the alphabet size to constant,
two high-level strategies are known to work: (1) swapping out the standard polynomial codes for Algebraic Geometry (AG) codes~\cite{GX12,GX13,GK14},
and (2) concatenation and \em distance amplification \em using expander graphs~\cite{AEL95,GI04,HW15,GKORS17,HRW17}.
To reduce the list-size to constant, the known strategies involve passing to 
 carefully constructing {\em subcodes} of Folded Reed-Solomon codes and univariate multiplicity codes, via pseudorandom objects such as \em subspace evasive sets \em or \em subspace designs\em~\cite{DL12,GW13,GX12,GX13,GK14}. 

In this work, we show that in fact both folded Reed-Solomon codes and univariate multiplicity codes are \em already \em list-recoverable with constant list-sizes, with no additional modification needed!  The resulting codes still have large alphabet sizes, but this can be ameliorated by using the same expander-based techniques described above.

We summarize the state of affairs for list-recovery in Table~\ref{tab:listrecovery}, and discuss our contributions in more detail below in Section~\ref{ssec:contributions}.
\begin{center}
\begin{table}
\footnotesize
\begin{tabular}{|m{4cm}|m{2.3cm}|m{2cm}|m{2cm}|c|m{1.5cm}|m{2.5cm}|}
\hline 
Code & Alphabet size $|\Sigma|$ & List size $L$ & Explicit? & Linear? & Decoding time & Notes \\ \hline
Completely random code & $\ell^{O(1/\eps)}$ & $O(\ell/\eps)$ & No & No & - &   \\\hline
Random linear code \cite{RW18} & $\ell^{O(1/\eps)}$ & $q^{O_\eps(\log^2(\ell))}$ & No & Yes & - & \\\hline
Folded RS codes \cite{GR08_folded_RS} & $\left( \frac{n}{\eps^2} \right)^{O(\log(\ell)/\eps^2)}$& $\left( \frac{n}{\eps^2} \right)^{O(\log(\ell)/\eps^2)}$ & Yes & Yes & $n^{O(\log(\ell)/\eps)}$ & \\ \hline 
Univariate Multiplicity \cite{Kop15} & $\left( \frac{n}{\eps^2} \right)^{O(\log(\ell)/\eps^2)}$& $\left( \frac{n}{\eps^2} \right)^{O(\log(\ell)/\eps^2)}$ & Yes & Yes & $n^{O(\log(\ell)/\eps)}$ & \\ \hline 
Folded RS/Univariate Multiplicity \cite{GW13} & $\left( \frac{n \ell}{\eps^2} \right)^{O(\ell/\eps^2)}$ & $\left( \frac{n \ell}{\eps} \right)^{O(\ell/\eps)}$ & Yes & Yes & $O( \ell \cdot n^2 / \eps )$& Output is a small subspace containing all nearby codewords. \\ \hline
Folded RS codes (This work, Theorem~\ref{thm:frs-list-rec-const-list}) & $\left( \frac{ n \ell}{\eps^2} \right)^{O(\ell/\eps^2)}$& $\left( \frac{ \ell}{\eps} \right)^{O(\frac{1}{\eps} \log(\ell/\eps))}$ & Yes & Yes & $\poly(n,L)$ &  \\ \hline 
Univariate Multiplicity codes (This work, Theorem~\ref{thm:unimult-list-rec-const-list-smalld}) & $\left( \frac{ n \ell }{\eps^2} \right)^{O(\ell/\eps^2)}$& $\left( \frac{ \ell}{\eps} \right)^{O(\frac{1}{\eps} \log(\ell/\eps))}$ & Yes & Yes & $\poly(n,L)$& For $d < q$ only.\\ \hline \hline
Folded RS subcodes (via subspace
evasive)~\cite{DL12} & $\left( \frac{ n \ell }{\eps^2 } \right)^{O(\ell/\eps^2)}$ & $O\left(\frac{\ell}{\eps}\right)^{O(\ell/\eps)}$ & Yes & No & $O_{\ell, \eps}(n^2)$ & \\ \hline
Folded AG (via subspace evasive)~\cite{GX12} & $\exp\left( \frac{ \ell \log(\ell/\eps) }{\eps^2} \right)$ & $O\left( \frac{\ell}{\eps} \right)$ & No & No & $\poly_{\ell,\eps}(n)$ &\\ \hline
Folded AG (via subspace designs)~\cite{GX13, GK14} & $\exp\left( \frac{ \ell \log(\ell/\eps) }{\eps^2} \right)$ & $2^{2^{2^{O_{\eps,\ell}(\log^*(n))}}}$ & Yes & Yes & $O_{\ell,\eps}(n^{O(1)})$ &\\ \hline
Tensor products of AG subcodes, plus expander techniques~\cite{HRW17} & $\exp(\ell/\eps^2)$ & $2^{2^{2^{O_{\eps,\ell}(\log^*(n))}}}$ & Yes & Yes & $O_{\ell,\eps}(n^{1.01})$ &  \\ \hline
Folded RS codes, plus expander techniques (This work, Corollary~\ref{cor:FRS_AEL}) & $(1 + \ell)^{O(1/\eps^5)}$& $O_{\eps, \ell}(1)$ & Yes & Yes & $\poly_{\ell,\eps}(n)$ & \\ \hline
\end{tabular}
\caption{Constructions of $(\alpha, \ell, L)$-list-recoverable codes of rate $R^* - \eps$, where
$R^* = 1 - \alpha$ is list-recovering capacity (when $|\Sigma| \geq (1+\ell)^{\Omega(1/\epsilon)}$).
 The top part of the table focuses on ``simple" algebraic constructions; the bottom part has constructions which are involved.  We assume that $R^* \in (0,1)$ is constant (independent of $n, \eps, \ell$).}\label{tab:listrecovery}
\end{table}

\begin{table}
\footnotesize
\begin{tabular}{|m{3cm}|c|c|c|c|}
\hline 
Code & Alphabet size $|\Sigma|$ & List size $L$ & Locality $t$ & Explicit?  \\ \hline
Tensor products of AG Subcodes, plus expander techniques~\cite{HRW17} & $\exp(\ell/\eps^2)$ & $2^{2^{2^{O_{\eps,\ell}(\log^*(n))}}}$ & $O_{\ell, \eps}(n^{0.01})$ & Yes   \\ \hline
Multivariate Multiplicity codes, plus expander techniques (This work, Theorem~\ref{thm:multimult-smallalpha2}) & $(1 + \ell)^{O(\ell/\eps^{11})}$ & $O_{\ell,\eps}(1) $ & $O_{\ell,\eps}(n^{0.01})$ & Yes  \\ \hline
Multivariate Multiplicity codes, plus expander techniques (This work, Theorem~\ref{thm:multimult-smallalpha}) & $(1 + \ell)^{O(1/\eps^{10})}$ & $\exp( \sqrt{ \log(n) \log\log(n) })$ & $\exp( \log^{3/4}(n) ( \log\log(n) )^{1/4} )$ & Yes  \\ \hline
\end{tabular}
\caption{Constructions of $(t, \alpha, \ell, L)$-locally-list-recoverable codes of rate $R^* - \eps$, where
$R^* = 1 - \alpha$ is list-recovering capacity (when $|\Sigma| \geq (1+\ell)^{\Omega(1/\epsilon)}$).
 We assume that $R^* \in (0,1)$ is constant (independent of $n, \eps, \ell$).}\label{tab:locallistrecovery}
\end{table}

\end{center}

\paragraph{Locally list-recoverable codes.}
As mentioned above, local decoding has been an important theme in coding theory for the past several decades.
Locality makes sense in the context of list-recovery as well.
The definition of local list-recovery (given formally below as Definition~\ref{defn:local-list-recover} below) is a bit involved, but intuitively the idea is as follows.  
As with list-recovery, we have input lists $S = (S_1,\ldots,S_n)$, so that each $S_i$ is of size at most $\ell$.  
The goal is to obtain information about a single symbol $c_i$ of a codeword $i$, given query access to $S$.
More precisely,
we will require that the decoder output a short list of randomized algorithms $A_1,\ldots,A_L$, each of which corresponds to a codeword $c$ with $|\{i \,:\, c_i \not\in S_i\}| \leq \alpha n$.  The requirement is that if $A_r$ corresponds to a codeword $c$, then on input $i$, $A_r(i)$ outputs $c_i$ with high probability, and using no more than $t$ queries to $S$.  If such a decoder exists, we say that the code is $(t, \alpha, \ell, L)$-locally-list-recoverable.  Local list-decoding is the case special case where $\ell = 1$.

This definition may seem a bit convoluted, but it turns out to be the ``right" definition for a number of settings.
For example, local list-decoding algorithms are at the heart of algorithms in cryptography~\cite{GL89}, learning theory~\cite{KM93}, and hardness amplification and derandomization~\cite{STV01}.
Locally list-\emph{recoverable} codes have been desirable as a step towards obtaining efficient capacity-achieving local list-decoding algorithms.  In particular, high-rate locally list-recoverable codes, combined with standard techniques, yield {\em capacity-achieving} locally list-decodable and locally list-recoverable codes.

However, until recently, we did not know of \em any \em high-rate locally list-recoverable codes.
The first such construction was given recently in \cite{HRW17}.   The approach of \cite{HRW17} is as follows: it takes a folded AG subcode from~\cite{GX13,GK14} (which uses subspace designs to find the subcode); applies tensor products many times; and concatenates the result with a locally correctable code.  Finally, to obtain capacity-achieving locally list-decodable/recoverable, codes, that work applies an expander-based technique of~\cite{AEL95} to pseudorandomly scramble up the symbols of the codewords to amplify the amount of error tolerated.  

The reason that so much machinery was used in \cite{HRW17} is that despite a great deal of effort, the ``natural" algebraic approaches did not seem to work.  Perhaps the most natural algebraic approach is via Reed-Muller codes, which have a natural local structure.  As discussed above, a Reed-Muller codeword corresponds to a low-degree multivariate polynomial, and restricting such a polynomial to a line yields a low-degree univariate polynomial, which corresponds to a Reed-Solomon codeword.    Using this connection, along with the Guruswami-Sudan algorithm for Reed-Solomon codes, Arora and Sudan~\cite{AS03} and Sudan, Trevisan and Vadhan~\cite{STV01} gave algorithms for locally list-decoding Reed-Muller codes up the the Johnson bound\footnote{Technically these algorithms only came within a factor $\sqrt{2}$ of the Johnson bound. To go all the way to the Johnson bound, one needs some additional ideas~\cite{BK-planes}; see~\cite{GK-lifted, Kop15} for further variations on this.}.  This algorithm also extends naturally to local list-recovery up to the Johnson bound~\cite{GKORS17}, but this means that for large values of $\ell$ one cannot obtain high-rate codes.  

One might hope to use a similar approach for multivariate multiplicity codes; after all, the univariate versions are list-recoverable to capacity.  However, the fact that the list sizes were large was an obstacle to this approach, and again previous work on the local list-decodability of multivariate multiplicity codes also only worked up to the Johnson bound~\cite{Kop15}.

In this work, we return to this approach, and---using our results on univariate multiplicity codes---show that in fact high-rate multivariate multiplicity codes are locally list-recoverable.  
Using our construction, combined with some expander-based techniques, we obtain capacity-achieving locally list-recoverable codes which improve on the state-of-the-art.  The quantitative results are stated in Table~\ref{tab:locallistrecovery}, and we discuss them in more detail in the next section.

\subsection{Our contributions}\label{ssec:contributions}
The main contribution of this work improved results on the (local)-list-recoverability of FRS codes and multiplicity codes.   We discuss a few of the concrete outcomes below.
\begin{itemize}
	\item \textbf{Constant list sizes for folded Reed-Solomon codes.}
Theorem~\ref{thm:frs-list-rec-const-list} says that a folded RS code of rate $R$ and alphabet size $q^{O(\ell/\eps^2)}$ is $(1 - R - \eps, \ell, L)$-list-recoverable with $L = \left( \ell/\eps \right)^{O\left( \frac{1}{\eps} \log(\ell/\eps) \right)}$.  This improves over the previous best-known list size for this setting, which was $\left( n/\eps \right)^{O\left( \frac{1}{\eps^2} \log(\ell) \right) }$.  In particular, when $\eps, \ell$ are constant, the list size $L$ improves from polynomial in $n$ to a constant.

\item \textbf{Constant list sizes for univariate multiplicity codes.} 
Theorem~\ref{thm:unimult-list-rec-const-list-smalld} recovers the same quantitative results as Theorem~\ref{thm:frs-list-rec-const-list} for univariate multiplicity codes with degree $d$ smaller than the characteristic of the underlying field.

When the degree $d$ is larger than the characteristic, which is what is relevant for the application to multivariate multiplicity codes, we obtain a weaker result.  We no longer have capacity-achieving codes, but we obtain high-rate list-recoverable codes with constant list sizes.  More precisely, Theorem~\ref{thm:full-unimult-list-rec} implies that rate $R$ univariate multiplicity codes are efficiently $(\alpha, \ell, L)$-list-recoverable for $L = \ell^{O(\ell \log(\ell))}$ and $\alpha = O( (1 -R)^2/\ell)$.  In particular, Theorem~\ref{thm:full-unimult-list-rec} is nontrivial even for high-rate codes, while the Johnson bound only gives results for $R < 1/\ell$.

	\item \textbf{High-rate multivariate multiplicity codes are locally list-recoverable.}
One reason to study the list-recoverability of univariate multiplicity codes is because list-recovery algorithms for univariate multiplicity codes can be used in local list-recovery algorithms for multivariate multiplicity codes.  Theorems~\ref{thm:main-multimult} and \ref{thm:main-multimult2} show that high-rate multivariate multiplicity codes are locally list-recoverable.   More precisely, in Theorem~\ref{thm:main-multimult}, we show that for constant $\ell, \epsilon$, a multivariate multiplicity code of length $n$ with rate $1 - \epsilon$: is efficiently $(t, \alpha, \ell, L)$-locally-list-recoverable for $\alpha = 1/\polylog(n)$, with list size $L$ and query complexity $t$ that are sub-polynomial in the block length $n$.   In Theorem~\ref{thm:main-multimult2}, we instantiate the same argument with slightly different parameters to show a similar result where $\alpha$ and $L$ are constant, but the query complexity $t$ is of the form $t = O(n^{0.01})$.

	\item \textbf{Capacity-achieving locally list-recoverable codes over constant-sized alphabets.}
Theorems~\ref{thm:main-multimult} and \ref{thm:main-multimult2} give high-rate locally-list-recoverable codes; however, these codes do not achieve capacity, and the alphabet sizes are quite large.  Fortunately, following previous work, we can apply a series of by-now-standard expander-based techniques to obtain capacity-achieving locally list-recoverable codes over constant-sized alphabets.
We do this in Theorems~\ref{thm:multimult-smallalpha} and \ref{thm:multimult-smallalpha2}, respectively. 

The only previous construction of capacity-achieving locally list-recoverable codes (or even high-rate locally list-recoverable codes) is due to \cite{HRW17}, which achieved arbitrary polynomially small query complexity (and even subpolynomial query complexity $n^{O(1/\log\log n)}$) with slightly superconstant list size.

Our codes in Theorem~\ref{thm:multimult-smallalpha} achieve subpolynomial
query complexity $\widetilde{\exp}(\log^{3/4} n)$ and subpolynomial list size.
This brings the query complexity for capacity achieving local list-decodability
close to the best known query complexity for locally decodable codes~\cite{KMRS}, which is $\widetilde{\exp}(\log^{1/2} n)$ (for the same codes).

Our codes in Theorem~\ref{thm:multimult-smallalpha2}
have arbitrary polynomially small query complexity, and constant list-size.
This improves upon the codes of~\cite{HRW17}.

The quantitative details are shown in Table~\ref{tab:locallistrecovery}.

	\item \textbf{Deterministic constructions of capacity-achieving list-recoverable codes with constant alphabet size and list size.}
	Our result in Theorem~\ref{thm:frs-list-rec-const-list} for Folded Reed-Solomon codes give capacity-achieving list-recoverable codes with constant list size, but with polynomial alphabet size.  By running these through some standard techniques, we obtain in Corollary~\ref{cor:FRS_AEL} efficient deterministic constructions of $\F_q$-linear, capacity-achieving, list-recoverable codes with constant alphabet size and list size, with a decoding algorithm that runs in time $n^{O(1)} \cdot \log(n)^{O_{\ell, \epsilon}(1)}$.
	
Codes with these properties do not seem to have been written down anywhere in the literature. Prior to our work, the same standard techniques could have also been applied to the codes of \cite{DL12} (which are nonlinear subcodes of Folded Reed-Solomon codes) to construct nonlinear codes with the same behavior.

\end{itemize}

\subsection{Overview of techniques}

In this subsection, we give an overview of the proofs of our main results.

\subsubsection{List recovery of folded Reed-Solomon and univariate multiplicity codes with constant output list size}

Let $C \subseteq \Sigma^n$ be either a folded Reed-Solomon code or a univariate multiplicity code
with constant distance $\delta > 0$.
Suppose that $s$ is the ``folding parameter" or ``multiplicity parameter," respectively, so that $\Sigma = \F_q^s$.
We begin with a warm-up by describing an algorithm for {\em zero-error} list-recovery; that is, when $\alpha = 0$.
Here we are given ``received lists" $S \in {\Sigma \choose \ell}^n$, and we want to find
the list $\calL$ of all codewords $c \in C$ such that $c_i \in S_i$ for each $i$.
The groundbreaking work of~\cite{GR08_folded_RS} showed that for constant $\ell$
and large but constant $s$, $\calL$ has size at most $q^{O_{\ell}(1)}$, and
can be found in time $q^{O_{\ell}(1)}$. We now show that $\calL$ is in fact of
size at most $L = O_{\ell, \delta}(1)$, and can be found in time $\poly(q, L)$.

The starting point for our improved list-recovery algorithms for 
folded Reed-Solomon and univariate multiplicity codes is the 
{\em linear-algebraic approach} to list-recovering these codes that was taken in~\cite{GW13}.
The main punchline of this approach is that the list $\calL$ is contained
in  an $\F_q$ affine-subspace $v_0 + V$ of dimension at most $O_\eps(\ell)$, and further that this subspace can be found in time $\poly(q)$ (this immediately
leads to the previously known bound on $\calL$).
Armed with this insight, we now bring the received lists $S$ back into play.
How many elements $c$ of the affine space $v_0 + V \subseteq C$ can 
have $c_i \in S_i$ for all $i \in [n]$?
We show that there cannot be too many 
such $c$.

The proof is algorithmic: we will give a randomized algorithm $\PRUNE$,
which when given the low dimensional affine space $v_0 + V$,
outputs a list of $K = O(1)$ elements of $C$, such that
such that for any $c \in \calL$, $c$ is included in the output of $A$
with high probability. This implies that $|\calL| \leq O(K) = O(1)$.

The algorithm $\PRUNE$ works as follows. For some parameter $\tau = O(1)$,
we pick coordinates $i_1, i_2, \ldots, i_\tau \in [n]$ uniformly at random.
Then the algorithm iterates over all the $\ell^\tau$ choices of $(y_1. \ldots, y_{\tau}) \in \prod_{j=1}^\tau S_{i_j}$.
For each such $(y_1,\ldots, y_{\tau})$, $\PRUNE$
checks if there is a unique element $w$ of $v_0 + V$ such that $w_{i_j} = y_j$ for all $j \in [\tau]$.
If so, we output that unique element $w$; otherwise (i.e., either there are either zero or greater than one 
such $w$'s) we do nothing.
Thus the algorithm $\PRUNE$ outputs at most $\ell^{\tau} = O(1)$ elements of $C$.

It remains to show that for any $c \in \calL$, the algorithm outputs $c$ with high probability. Fix such a $c$.
By assumption, for every $i \in [n]$, $c_i \in S_i$. Thus there will be
an iteration where the algorithm $\PRUNE$ takes $(y_1, \ldots, y_{\tau}) = (c_{i_1}, \ldots, c_{i_\tau})$.
In this iteration, there will be at least one $w$ (namely $c$) which has the
desired property. Could there be more? If there was another $c' \in v_0 + V$ with this property,
then the nonzero vector $c-c' \in V$ would have the property that $c-c'$ vanishes 
on all coordinates $i_1, \ldots, i_{\tau}$.
It turns out that this can only happen with very low probability. Lemma 2 from~\cite{Saraf-Yekhanin} shows that that
for any linear space $V$ with dimension $k$ and distance at least $\delta$, for $\tau$ a large enough constant ($\tau = \Omega(k/\delta)$),
it is very unlikely that there exists a nonzero element of $V$ that vanishes at
$\tau$ random coordinates $i_1, \ldots, i_{\tau}$.
Thus with high probability,
$c$ is the unique $w$ found in that iteration, and is thus included in the output
of $\PRUNE$. This completes the description and analysis of the algorithm $\PRUNE$,
and thus of our zero-error list-recovery algorithm.

One way to prove (a version of) Lemma 2 from~\cite{Saraf-Yekhanin} is as follows.
First we note the following simple but important lemma:
\begin{lemma}
\label{lem:general-subspace-design}
Let $\Sigma = \F_q^s$. Let $W \subseteq (\Sigma)^n$ be an $\F_q$-subspace
with $\dim(W) = t \geq 1$. Suppose $W$ has minimum distance at least $\delta$.
Then:
$$ \mathbb E_{i \in [n]} [ \dim(W \cap H_i) ] \leq t - \delta,$$
where $H_i = \{v \in \Sigma^n \mid v_i = 0 \}$.
\end{lemma}
Lemma~\ref{lem:general-subspace-design} says that for any subspace $W \subseteq \Sigma^n$ of good distance, fixing a coordinate
to $0$ reduces the dimension a little in expectation.
Iterating this, we see that fixing many coordinates is very likely to reduce the dimension down to zero, and this proves the result that we needed above.

With our warm-up complete, we turn to
our main theorem on the list-recoverability of Folded Reed-Solomon codes (Theorem~\ref{thm:frs-list-rec-const-list}), which shows that the output list size is small even in the presence of an $\alpha = \delta - \epsilon$ fraction of errors (for small $\epsilon > 0$).  Our approach generalizes the $\alpha = 0$ case described above.
Let $\calL$ be the list of $(\delta-\epsilon)$-close codewords.
Again, the linear-algebraic list decoder of~\cite{GW13} can produce a low
dimensional affine subspace $v_0 + V$ such that $\calL \subseteq v_0 + V$.
Next, we show that the very same algorithm $\PRUNE$ described above (with a different
setting of the parameter $\tau$) does the desired list-recovery with at least some small constant probability $p_0$.
This will imply that $|\calL| \leq \frac{\ell^{\tau}}{p_0}$.

To see why this works, 
fix a codeword $c \in \mathcal{L}$.
First observe that if we pick $i_1, \ldots, i_{\tau}$ uniformly at random,
the probability that $c_{i_j} \in S_{i_j}$ for all $j = 1,\ldots, \tau$
is at least $p' = (1- \delta + \epsilon)^{\tau}$.  This is small, but not too small; thus, there is some chance that at least one $w$ (the correct one) is found by $\PRUNE$.

Following the previous analysis, we now have to bound the probability that
for random $i_1, \ldots, i_{\tau} \in [n]$, the space of codewords from $V$ that vanish on all of $i_1,
\ldots, i_{tau}$ has dimension at least one.  This is the probability that strictly greater than one $w$ is found by $\PRUNE$.
This time we will need a stronger (and much more specialized) version of Lemma~\ref{lem:general-subspace-design},
which shows that for subspaces $W$ of the Folded Reed-Solomon code, 
fixing a random coordinate to $0$ reduces the dimension by a lot: much more than the $\delta$ that we got from Lemma~\ref{lem:general-subspace-design}. 
Such a lemma was proved in~\cite{GK14}, although in a different language, and for a very different purpose.
This lemma roughly shows that the expected dimension of $W \cap H_i$, for a random $i \in [n]$,
is at most $(1 - \delta)\dim(W)$. Setting $\tau = O(\log(\dim(V))/\delta)$,
with $\tau$ applications of this lemma, we get that the probability
that the space of codewords from $V$ that vanish on all of $i_1, \ldots, i_{\tau}$ has dimension at least one is at most $p'' = (1-\delta)^{\tau}\dim(V)$.
Note that this probability is tiny compared to $p'$, and thus
the probability that the algorithm $\PRUNE$ succeeds in finding $c$
is at least $p' - p'' \approx p'$, as desired.

The description above was for folded RS codes, but 
 same method works for univariate multiplicity codes whose
degree $d$ is smaller than the characteristic of the field $\F_q$.
We state this in Theorem~\ref{thm:unimult-list-rec-const-list-smalld}.
The proof follows the same outline, using a different but analogous lemma from~\cite{GK14}.

For application to local list-recovery of multivariate multiplicity codes, however, we need to deal with univariate multiplicity codes where the degree $d$ is larger than $q$.
In Theorem~\ref{thm:full-unimult-list-rec}, we show how to accomplish this when the fraction of errors $\alpha$ is very small.
The algorithm and the outline of the analysis described above can again do the job for this setting, although the analysis is much more involved. The proof, which we give in Section~\ref{sec:Unimult}, gives better quantitative bounds than the previous approach, and requires us to open up the relevant lemma from \cite{GK14}.  At the end of the day, we are able to prove a reasonable version of this lemma for the case when $d > q$, and this allows the analysis to go through.

\subsubsection{Local list-recovery of multivariate multiplicity codes}

We now describe the high-level view of our local list-recovery algorithms. 
Our algorithm for local list-recovery of multivariate multiplicity codes
follows the general paradigm for local list-decoding of Reed-Muller codes
by Arora and Sudan~\cite{AS03} and Sudan, Trevisan and Vadhan~\cite{STV01}.
In addition to generalizing various aspects of the paradigm, we need to introduce some further ideas to account for the fact that we are in the high rate setting\footnote{These ideas can also be used to improve the analysis of the 
\cite{AS03} and \cite{STV01} local list-decoders for Reed-Muller codes. In particular, they can remove the restriction that the degree $d$ needs to be
at most $1/2$ the size of the field $\F_q$ for the local list-decoder to work.}.

Local list-decoding of Reed-Muller codes is the following problem: we are given
a function $r: \F_q^m \to \F_q$ which is promised to be close to the evaluation
table of some low degree polynomial $Q(X_1, \ldots, X_m)$.
At the high level, the local list-decoding algorithm of~\cite{STV01} for Reed-Muller
codes has two phases: generating advice, and decoding with advice. To generate the advice,
we pick a uniformly random $\ba \in \F_q^m$ and ``guess" a value $z \in \F_q$ (this guessing
can be done by going over all $z \in \F_q$). Our hope for this guess is that $z$ equals $Q(\ba)$.

Once we have this advice, we see how to decode. We define an oracle machine $M^r[\ba, z]$,
which takes as advice $[\ba, z]$, has query access to $r$, and given an input $\bx \in \F_q^m$, tries to compute $Q(\bx)$.
The algorithm first considers the line $\lambda$ passing through $\bx$ and the advice point $\ba$,
and list-decode the restriction of $r$ to this line to obtain a list $\calL_\lambda$ of univariate polynomials.
These univariate polynomials are candidates for $Q|_\lambda$. Which of these univariate polynomials
is $Q|_{\lambda}$? We use our guess $z$ (which is suppose to be $Q(\ba)$): if there is a unique
univariate polynomial in the list with value $z$ at $\ba$, then we deem that to be our candidate for
$Q|_\lambda$, and output its value at the point $\bx$ as our guess for $Q(\bx)$. This algorithm 
will be correct on the point $\bx$ if (1) there are not too many errors on the line through $\bx$ and $\ba$,
and (2) no other polynomnial in $\calL_\lambda$ takes the same value at $\ba$ as $Q|_\lambda$ does.
The first event is low probability by standard sampling bounds, and the second is low probability
using the random choice of $\ba$ and the fact that $\calL_\lambda$ is small.
This algorithm does not succeed on all $\bx$, but one can show that for random $\ba$ and $z = Q(\ba)$,
this algorithm does succeed on most $\bx$. Then we can run a standard local correction algorithm for
Reed-Muller codes to then convert it to an algorithm that succeeds on all $\bx$ with high probability.

We are trying to locally list-recover a multivariate multiplicity code; the codewords are of the form $(Q^{(<s)}(\by))_{\by \in \F_q^m}$, where $Q^{(<s)}(\by) \in \F_q^{{m + s - 1 \choose m}} =: \Sigma_{m,s}$ is a tuple that consists of all partial derivatives of $Q$ of order less than $s$, evaluated at $\by$.
We are given query access to a function $S:\F_q^m \to { \Sigma_{m,s} \choose \ell }$, 
where $S(\by) \subset \Sigma_{m,s}$ is the received list for the coordinate indexed by $\by$. 
Suppose for the following discussion that $Q(\bX) \in \F_q[X_1,\ldots, X_m]$ is a low-degree multivariate polynomial so that $|\{\by \,:\, Q^{(<s)}(\by) \not\in S(\by) \}| \leq \alpha q^m$.   We want to describe an algorithm that, with high probability will output a randomized algorithm $A_j:\F_q^m \to \Sigma_{m,s}$ that will approximate $Q^{(<s)}$.

There are two main components to the algorithm again: generating the advice, and decoding with advice.
The advice is again a uniformly random point $\ba \in \F_q^m$, and a guess $z$
which is supposed to equal $Q^{(<s^*)}(\ba)$, {\em a very high order evaluation of $Q$ at $\ba$}, for some $s^* \gg s$.
We discuss how to generate $z$ later, let us first see how to use this advice to decode.

To decode using the advice $[\ba, z]$, we give an oracle machine $M^S[\ba, z]$ which takes advice $[\ba, z]$ and has query access to $S$. 
If $z = Q^{(<s^*)}(\ba)$, then $M^S[\ba,z](\bx)$ will be equal to $Q^{(<s)}(\bx)$ with high probability over $\bx$ and $\ba$.  This algorithm
is discussed in Section~\ref{ssec:oracle}.   Briefly, the idea is to consider the line $\lambda$ through $\bx$ and $\ba$ and again run the
univariate list-recovery algorithm on the restrictions of $S$ to this line to obtain a list $\mathcal{L}_\lambda$. 
We hope that $Q|_\lambda$ is in this list, and that $Q|_\lambda$ does not have the same {\em order $s^*$} evaluation\footnote{This is
why we take $s^*$ large: it is much more unlikely that there will be a collision of higher order evaluations at the random point $\ba$.}
on $\ba$ as any other element of $\mathcal L_\lambda$ -- this will allow us to identify it with the help of the advice $z = Q^{(<s^*)}(\ba)$.
Once we identify $Q|_\lambda$, we output its value at $\bx$ as our guess for $Q(\bx)$.

To generate the advice $z$, we give an algorithm $\RecCand$, which takes as input a point $\ba \in \F_q^m$,
has query access to $S$, and returns a short list $Z \subset \Sigma_{m, s^*}$ of guesses for $Q^{(<s^*)}(\ba)$.
Recall that we have $s^*$ quite a bit larger than $s$.  This algorithm is discussed in Section~\ref{ssec:reccand}.  
Briefly, $\RecCand$ works by choosing random lines through $\ba$ and running the (global) list-recovery algorithm
for univariate multiplicity codes on the restriction of the lists $S$ to these lines.
Then it aggregates the results to obtain $Z$. This aggregation turns out to be a list-recovery problem
for Reed-Muller codes evaluated on product sets. We describe this algorithm for list-recovery in Appendix~\ref{app:rm}.

Summarizing, our local list-recovery algorithm works as follows.  
First, we run $\RecCand$ on a random point $\ba \in \F_q^m$ to generate a short list $Z \subseteq \Sigma_{m,s^*}$  of possibilities for $Q^{(<s^*)}(\ba)$.  Then, for each $z \in Z$, we will form the oracle machine $M^S[\ba,z]$.  We are not quite done even if the advice $z$ is good, since $M^S[\ba,z](\bx)$ may not be equal to $Q^{(<s)}(\bx)$; we know this probably happens for most $\bx$'s, but not necessarily for the one that we care about.  Fortunately, $M^S[\ba,z]$ will agree with $Q^{(<s)}$ for many inputs $\by$, and so we can use the fact that multivariate multiplicity codes are locally correctable to finish the job~\cite{KSY14}.  When we iterate over the advice $z \in Z$, this will give the list of randomized algorithms $A_1,\ldots, A_L$ that the local list-recovery algorithm returns.

\subsubsection{Organization}
We begin in Section~\ref{sec:Preliminaries} with notation and preliminary definitions.  Once these are in place, we will prove Theorem~\ref{thm:frs-list-rec-const-list} about Folded RS codes in Section~\ref{sec:FRS}.  In Section~\ref{sec:Unimult}, we extend our analysis of Folded RS codes to univariate multiplicity codes, and prove Theorems~\ref{thm:unimult-list-rec-const-list-smalld} and \ref{thm:full-unimult-list-rec} for small and large degrees $d$ respectively.  In Section~\ref{sec:Multimult}, we present our local list-recovery algorithm for multivariate multiplicity codes, and state Theorems~\ref{thm:main-multimult} and \ref{thm:main-multimult2} about high-rate local list-recovery of multivariate multiplicity codes.  Finally in Section~\ref{sec:Smallalpha} we run our results through the expander-based machinery of~\cite{AEL95}, to obtain Theorems~\ref{thm:multimult-smallalpha} and \ref{thm:multimult-smallalpha2} which give capacity-achieving locally list-recoverable codes over constant-sized alphabets.

\section{Notation and Preliminaries}\label{sec:Preliminaries}
We begin by formally defining the coding-theoretic notions we will need, and by setting notation.
We denote by $\F_{q}$ the finite field of $q$ elements. 
For any pair of strings $x,y\in\Sigma^{n}$, the \textsf{relative
distance} between $x$ and $y$ is the fraction of coordinates on which $x$ and $y$
differ, and is denoted by $\dist(x,y):= \left|\left\{ i\in\left[n\right]:x_{i}\ne y_{i}\right\} \right|/n$. For a positive integer $\ell$ we denote 
by ${\Sigma \choose \ell}$ the set containing all subsets of $\Sigma$ of size $\ell$, and for any pair of strings $x \in \Sigma^n$ 
and $S \in  {\Sigma \choose \ell}^n$ we denote by $\dist(x,S)$ the fraction of coordinates $i \in [n]$ for which $x_i \notin S_i$, that is, 
$\dist(x,S):= \left|\left\{ i\in\left[n\right]:x_{i}\notin S_i \right\} \right|/n$. Throughout the paper, we use $\exp(n)$ to denote $2^{\Theta(n)}$. Whenever we use $\log$, it is to the base $2$.
The notation $O_a(n)$ and $\poly_a(n)$ means that we treat $a$ as a constant; that is, $\poly_a(n) = n^{O_a(1)}$.

\subsection{Error-correcting codes}

Let $\Sigma$ be an alphabet and let $n$ be a positive integer (the
\textsf{block length}). A code is simply a subset $C\subseteq\Sigma^{n}$. The elements of a code $C$ are called \textsf{codewords}. 
If $\F$ is a finite field and $\Sigma$ is a vector space over $\F$,
we say that a code $C\subseteq\Sigma^{n}$ is \textsf{$\F$-linear}
if it is an $\F$-linear subspace of the $\F$-vector space $\Sigma^{n}$.
In this work most of our codes will have alphabets $\Sigma = \F^s$, and we will use \textsf{linear} to mean $\F$-linear.
The \textsf{rate}
of a code is the ratio $\frac{\log|C|}{\log(|\Sigma|^{n})}$, which
for $\F$-linear codes equals $\frac{\dim_{\F}(C)}{n\cdot\dim_{\F}(\Sigma)}$.
The \textsf{relative distance} $\dist(C)$ of $C$ is the minimum $\delta >0$ such that for
every pair of distinct codewords $c_{1},c_{2}\in C$ it holds that
$\dist(c_{1},c_{2})\ge\delta$.

Given a code $C \subseteq \Sigma^n$, we will occasionally abuse notation and think of $c \in C$ as a map $c: \mathcal{D} \to \Sigma$, where $\mathcal{D}$ is some domain of size $n$.  With this notation, the map $c:\mathcal{D} \to \Sigma$ corresponds to the vector $(c(x))_{x \in \mathcal{D}} \in \Sigma^n$.

For a code $C \subseteq \Sigma^n$ of relative
distance $\delta$, a given parameter $\alpha<\delta/2$, and a string
$w\in\Sigma^{n}$, the \textsf{problem of decoding from $\alpha$~fraction
of errors} is the task of finding the unique $c\in C$ (if any) which
satisfies $\dist(c,w)\leq\alpha$. 

\subsection{List-decodable and list-recoverable codes}

List decoding is a paradigm that allows one to correct more than a $\delta/2$ fraction of errors by returning a small list of close-by codewords. 
More formally, for $\alpha \in [0,1]$ and an integer $L$ we say that a code $C \subseteq \Sigma^n$ is \textsf{$(\alpha,L)$-list-decodable} if for any $w \in \Sigma^n$ there are at most $L$ different codewords $c\in C$ which satisfy that $\dist(c,w)\leq \alpha$. 

\textsf{List recovery} is a more general notion
where one is given as input a small list of candidate symbols for each of the coordinates and is required to output a list of codewords that are consistent with many of the input lists. 
Formally we say that a code $C\subseteq \Sigma^n$ is \textsf{$(\alpha,\ell,L)$-list-recoverable} if for any $S \in {\Sigma \choose \ell}^{n}$ there are at most $L$ different codewords $c\in C$ which satisfy that $\dist(c,S)\leq \alpha$. Note that list decoding corresponds to the special case of $\ell=1$. 

\subsection{Locally correctable and locally list-recoverable codes}

\paragraph{Locally correctable codes.}

Intuitively, a code is said to be \textsf{locally correctable}~\cite{BFLS91,STV01,KT00}
if, given a codeword $c\in C$ that has been corrupted by some errors,
it is possible to decode any coordinate of $c$ by reading only a
small part of the corrupted version of $c$. Formally, it is defined
as follows. 
\begin{definition}[Locally correctable code (LCC)]
We say that a code $C\subseteq\Sigma^{n}$ is $(t,\alpha)$-\textsf{locally correctable}
if there
exists a randomized algorithm $A$ that satisfies the following requirements: 
\begin{itemize}
\item \textbf{Input:} $A$ takes as input a coordinate $i\in\left[n\right]$
and also gets oracle access to a string $ w\in\Sigma^{n}$ that is
$\alpha$-close to a codeword $c\in C$. 
\item \textbf{Query complexity:} $A$ makes at most $t$ queries to the
oracle $w$. 
\item \textbf{Output:} $A$ outputs $c_{i}$ with probability at least $\frac{2}{3}$. 
\end{itemize}
\end{definition}
\begin{remark}
By definition it holds that $\alpha<\dist(C)/2$. The above success
probability of~$\frac{2}{3}$ can be amplified using sequential repetition,
at the cost of increasing the query complexity. Specifically, amplifying
the success probability to $1-e^{-t}$ requires increasing the query
complexity by a multiplicative factor of $O(t)$. 
\end{remark}

\paragraph{Locally list-recoverable codes.}

The following definition generalizes the notion of locally correctable codes to the setting of list decoding / recovery.
In this setting the algorithm $A$ is required to find all the nearby codewords in an implicit sense.

\begin{definition}[Locally list-recoverable code]\label{defn:local-list-recover}
We say that a code $C\subseteq\Sigma^{n}$ is $(t,\alpha,\ell,L)$-\textsf{locally list-recoverable}
if there exists a randomized algorithm $A$ that satisfies the following requirements: 
\begin{itemize}
\item \textbf{Input:} $A$ gets oracle access to a string $S \in {\Sigma \choose \ell}^{n}$.
\item \textbf{Query complexity:} $A$ makes at most $t$ queries to the oracle $S$.
\item \textbf{Output:} $A$ outputs $L$ randomized algorithms $A_1,\ldots,A_L$, where each $A_j$ takes as input a coordinate $i\in\left[n\right]$, makes at most $t$ queries to the oracle $S$, and outputs a symbol in $\Sigma$. 
\item \textbf{Correctness:} For every codeword $c\in C$ for which $\dist(c,S)\leq \alpha$, with probability at least $\frac{2}{3}$ over the randomness of $A$ the following event happens:
there exists some $j\in [L]$ such that for all $i \in [n]$,  $$ \Pr[A_j(i) = c_i] \geq \frac{2}{3},$$ where the probability is over the internal randomness of $A_j$.
\end{itemize}
\end{definition}
We say that $A$ has running time $T$  if $A$ outputs the description of the algorithms $A_1,\ldots,A_L$ in time at most $T$ and each $A_j$ has running time at most $T$. 
We say that a code is \textsf{$(t,\alpha,L)$-locally list-decodable} if it is $(t,\alpha,1,L)$-locally list-recoverable.

\subsection{Polynomials and derivatives}

Let $\F_q[X]$ be the space of univariate polynomials over $\F_q$.
We will often be working with linear and affine subspaces
of  $\F_q[X]$. We will denote linear subspaces of $\F_q[X]$ by the letters
$U, V, W$, and  affine subspaces of $\F_q[X]$ as 
$v_0 + V$, where $v_0 \in \F_q[X]$ and $V$ is a linear subspace.

For polynomials $P_1,\ldots,P_s \in \F_q[X]$, we define their \textsf{Wronskian}, $W(P_1,\ldots,P_s)$, by
\[ W(P_1,\ldots,P_s)(X) = \begin{pmatrix} P_1(X) & \cdots & P_s(X) \\
P_1^{(1)}(X) & \cdots & P_s^{(1)}(X) \\
\vdots & & \vdots \\
P_1^{(s-1)}(X) & \cdots & P_s^{(s-1)}(X) 
\end{pmatrix}. \]

For $i \in \mathbb{N}$, we define the \textsf{$i$'th (Hasse) derivative} $P^{(i)}(X)$ as the coefficient of $Z^i$ in the expansion
\[ P(X + Z) = \sum_i P^{(i)}(X)Z^i. \]

For multivariate polynomials $P \in \F_q[X_1,\ldots,X_m]$, we use the notation $\mathbf{X} = (X_1, \ldots, X_m)$ and $\mathbf{X}^{\mathbf{i}} = \prod_{j} X_j^{i_j}$ where $\mathbf{i} = (i_1,\ldots,i_m) \in \mathbb{Z}^m$.  
For $\mathbf{i} \in \mathbb{Z}^m$, we define the \textsf{$\mathbf{i}$'th (Hasse) derivative} $P^{(\mathbf{i})}(\mathbf{X})$ by
\[ P(\mathbf{X} + \mathbf{Z}) = \sum_{\mathbf{i}} P^{(\mathbf{i})}(\mathbf{X}) \mathbf{Z}^{\mathbf{i}}. \]

\subsection{Some families of polynomial codes}
In this section, we formally define the families of codes we will study: folded Reed-Solomon codes~\cite{GR08_folded_RS}, univariate multiplicity codes~\cite{RT-m-metric,KSY14,GW13}, and multivariate multiplicity codes~\cite{KSY14}.

 \paragraph{Folded Reed-Solomon codes.}
Let $q$ be a prime power, and let $s,d,n$ be nonnegative integers such that $n \leq (q-1)/s$. Let $\gamma \in \F_q$ be a primitive element of $\F_q$, and let $a_1,a_2,\ldots,a_n$ be distinct elements in $\{\gamma^{si} \mid 0 \leq i \leq (q-1)/s -1 \}$. Let $\Domain = \{a_1, \ldots, a_n\}$. 

For a polynomial $P(X) \in \F_q[X]$ and $a \in \F_q$, let $P^{[s]}(a) \in \F_q^{s}$ 
denote the vector:
$$ P^{[s]}(a) = \left[  \begin{matrix} P(a) \\  P(\gamma a) \\ \vdots \\ P(\gamma^{s-1} a)  \end{matrix}  \right].$$

 The \textsf{folded Reed-Solomon code $\FRS_{q,s}(n,d)$} is a code over alphabet $\F_q^s$.
To every polynomial $P(X) \in \F_q[X]$ of degree at most $d$, there corresponds a codeword
$c$:
$$ c: \Domain \to \F_q^s,$$
where for each $a \in \Domain$:
$$ c(a) = P^{[s]}(a).$$
Explicitly,
\begin{align*}
P(x) &\mapsto \left(  P^{[s]}(a_1), P^{[s]}(a_2), \ldots, P^{[s]}(a_n) \right) \\
&=\left(
 \left[ \begin{matrix} P(a_1) \\  P(\gamma a_1) \\ \vdots \\ P(\gamma^{s-1} a_1)  \end{matrix} \right],
 \left[ \begin{matrix} P(a_2) \\  P(\gamma a_2) \\ \vdots \\ P(\gamma^{s-1} a_2)  \end{matrix} \right], \ldots,
 \left[ \begin{matrix} P(a_n) \\  P(\gamma a_n) \\ \vdots \\ P(\gamma^{s-1} a_n)  \end{matrix} \right]
 \right).
\end{align*}
We denote the codeword of $\FRS_{q,s}(n,d)$ corresponding to the polynomial $P(X)$
by $\FRSEnc_s(P)$ (when the parameters $q,n$ are clear from the context).

Note that \textsf{Reed-Solomon codes} correspond to the special case of $s=1$. 
The following claim summarizes the basic properties of folded Reed-Solomon codes. 
 \begin{claim}[\cite{GR08_folded_RS}]
 The folded Reed-Solomon code $\FRS_{q,s}(n,d)$ is an $\F_q$-linear code over alphabet $\F_q^s$ of block length $n$, rate $(d+1)/(sn)$, and relative distance at least $1-d/(sn)$. 
 \end{claim}

\paragraph{Univariate multiplicity codes.}
Let $q$ be a prime power, and let $s,d,n$ be nonnegative integers such that $n\leq q$. 
Let $a_1,a_2,\ldots,a_n$ be distinct elements in $\F_q$. Let $\Domain = \{a_1, \ldots, a_n\}$.

For a polynomial $P(X) \in \F_q[X]$, let $P^{(<s)}(x) \in \F_q^{s}$ 
denote the vector:
$$ P^{(<s)}(a) = \left[  \begin{matrix} P(a) \\  P^{(1)}(a) \\ \vdots \\ P^{(s-1)}(a)  \end{matrix}  \right].$$

The \textsf{univariate multiplicity code $\MULT_{q,s}^{(1)}(n,d)$} is a code over alphabet $\F_q^s$.
To every polynomial $P(X) \in \F_q[X]$ of degree at most $d$, there corresponds a codeword
$c$:
$$ c: \Domain \to \F_q^s,$$
where for each $a \in \Domain$:
$$ c(a) = P^{(<s)}(a).$$
Explicitly,  
\begin{align*}
P(x) &\mapsto \left(P^{(<s)}(a_1), P^{(<s)}(a_2), \ldots, P^{(<s)}(a_n)\right) \\
&= \left(
 \left[ \begin{matrix} P(a_1) \\  P^{(1)}(a_1) \\ \vdots \\ P^{(s-1)}(a_1)  \end{matrix} \right],
 \left[ \begin{matrix} P(a_2) \\  P^{(1)}(a_2) \\ \vdots \\ P^{(s-1)}(a_2)  \end{matrix} \right], \ldots,
 \left[ \begin{matrix} P(a_n) \\  P^{(1)}(a_n) \\ \vdots \\ P^{(s-1)}(a_n) \end{matrix} \right]
 \right).
\end{align*}
We denote the codeword of $\MULT^{(1)}_{q,s}(n,d)$ corresponding to the polynomial $P(X)$
by $\MultEnc_s(P)$ (when the parameters $q,n$ are clear from the context).

Once again, Reed-Solomon codes correspond to the special case of $s=1$. 

 \begin{claim}[\cite{KSY14}, Lemma 9]
 The univariate multiplicity code $\MULT_{q,s}^{(1)}(n,d)$ is an $\F_q$-linear code over alphabet $\F_q^s$ of block length $n$, rate $(d+1)/(sn)$, and relative distance at least $1-d/(sn)$. 
 \end{claim}

Of particular importance is the setting where $q = n$ and $\Domain$ equals
the whole field $\F_q$. We refer to this code as the {\em whole-field univariate multiplcity code}, and denote it by $\MULT_{q,s}^{(1)}(d)$. 
This will be relevant to multivariate multiplicity codes, which we define next.

\paragraph{Multivariate multiplicity codes.}
Multivariate multiplicity codes are a generalization of whole-field univariate multiplicity codes to the multivariate setting.

Let $q$ be a prime power, and let $s, d, m$ be nonnegative integers.
Let $U_{m,s}$ denote the set $\{ \vi \in \N^m \mid \wt(\vi) < s \}$.
Note that $|U_{m,s}| = {s + m - 1 \choose m}$.
Let $\Sigma_{m,s} = \F_q^{U_{m,s}}$.

For a polynomial $P(X_1, \ldots, X_m) \in \F_q[X_1, \ldots, X_m]$,
and a point $\ba \in \F_q^m$, define $P^{(<s)}(\ba) \in \Sigma_{m,s}$
by:
$$ P^{(<s)}(\ba) = ( P^{(\ii)}(\ba) )_{\ii \in U_{m,s}}.$$

The \textsf{multiplicity code $\MULT_{q,s}^{(m)}(d)$} is a code over alphabet $ \Sigma_{m,s}$.
To every polynomial $P(X_1,\ldots,x_m) \in \F_q[X_1,\ldots,X_m]$ of (total) degree at most $d$,
there corresponds a codeword  as  
$$ c: \F_q^m \to \Sigma_{m,s},$$
where for each $\ba \in \F_q^m$,
$$c(\ba) = P^{(<s)}(\ba).$$

Note that \textsf{Reed-Muller codes} correspond to the special case of $s=1$.
 \begin{claim}[\cite{KSY14}, Lemma 9]\label{claim:multparams}
 The multivariate multiplicity code $\MULT_{q,s}^{(m)}(d)$ is an $\F_q$-linear code over alphabet $\F_q^{m+s -1 \choose m}$ of block length $q^m$, rate at least $(1- {m^2}/{s})(d/(sq))^m$, and relative distance at least $1-d/(sq)$. 
 \end{claim}

\section{List recovering folded Reed-Solomon codes with constant output list size}\label{sec:FRS}

Our first main result shows that folded Reed-Solomon codes are list-recoverable (and in particular, list-decodable) up to capacity with constant output list size, independent of $n$.

\begin{theorem}[List recovering FRS with constant output list size]\label{thm:frs-list-rec-const-list}
Let $q$ be a prime power, and let $s,d,n$ be nonnegative integers such that $n \leq (q-1)/s$. 
Let $\epsilon >0$ and $\ell \in \N$ be such that $16 \ell /\epsilon^2 \leq s$. Then the folded Reed-Solomon code $\FRS_{q,s}(n,d)$ is $(\alpha,\ell,L)$-list-recoverable for $\alpha = 1- d/(sn) -\epsilon$ and $L =  \left( \frac {\ell} {\epsilon} \right)^{O\left(\frac{1}{\epsilon} \log \frac{\ell}{\epsilon}\right) }$.

Moreover, there is a randomized algorithm that list recovers $\FRS_{q,s}(n,d)$ with the above parameters in time $\poly(\log q, s,d,n,(\ell/\epsilon)^{ \log(\ell/\epsilon)/\epsilon})$. 
\end{theorem}

In particular, the $\ell=1$ case yields the following statement about list-decoding. 

\begin{corollary}[List decoding FRS with constant output list size]\label{cor:frs-list-dec-const-list}
Let $q$ be a prime power, and let $s,d,n$ be nonnegative integers such that $n \leq (q-1)/s$. 
Let $\epsilon >0$ be such that $16/\epsilon^2 \leq s$. Then the folded Reed-Solomon code $\FRS_{q,s}(n,d)$ is $(\alpha,L)$-list decodable
for  $\alpha = 1- d/(sn) -\epsilon$ and $L =  \left(\frac 1 \epsilon\right)^{O\left(\frac{1}{\epsilon} \log \frac{1}{\epsilon}\right) }$.

Moreover, there is a randomized algorithm that list decodes $\FRS_{q,s}(n,d)$ with the above parameters in time $\poly(\log q, s,d,n,(1/\epsilon)^{ \log(1/\epsilon)/\epsilon})$. 
\end{corollary}

The proof of Theorem \ref{thm:frs-list-rec-const-list} consists of two main steps. The first step, from \cite{GW13}, shows that the output list is contained in a low dimensional subspace.
The second step, which relies on results from \cite{GK14}, shows that the output list cannot contain too many codewords from a low dimensional subspace, and therefore is small. The two steps are presented in Sections \ref{subsec:FRS-1} and \ref{subsec:FRS-2}, respectively, followed by the proof of Theorem \ref{thm:frs-list-rec-const-list} in Section \ref{subsec:FRS-main}.

\subsection{Output list is contained in a low dimensional subspace}\label{subsec:FRS-1}

The following theorem from~\cite{GW13}  shows that
the output list is contained in a low dimensional subspace, which can also be found efficiently.

\begin{theorem}[\cite{GW13}, Theorem 7]\label{thm:output-list-low-dim}
Let $q$ be a prime power, and let $s,d,n, \ell, r$ be nonnegative integers such that $n \leq (q-1)/s$ and $r \leq s$. 
Let $S : \Domain \to {\F_q^s \choose \ell}$ be an instance of the list-recovery problem for $\FRS_{q,s}(n,d)$.
Suppose the decoding radius $\alpha$ satisfies:
\begin{equation}\label{eqdecoderadius}
  \alpha \leq 1- \frac{\ell}{r+1} - \frac{r}{r+1} \cdot \frac{s}{s-r+1} \cdot \frac{d} {sn}. 
 \end{equation}
Let
$$\calL = \{ P(X) \in \F_q[X] \mid \deg(P) \leq d \mbox{ and } \dist(\FRSEnc_s(P), S) \leq \alpha \}.$$

There is a (deterministic) algorithm that given $S$, runs in time $\poly(\log q, s,d,n,\ell)$, and computes an affine subspace $v_0 + V \subseteq \F_q[X]$ such that:
\begin{enumerate}
\item $\calL \subseteq V$,
\item $\dim(V) \leq r-1$.
\end{enumerate}
\end{theorem}

\begin{remark}
Theorem 7 of \cite{GW13}  only deals with the case where $a_i= \gamma^{s(i-1)}$ for all $i=1,\ldots ,n$, and $\ell=1$. However, it can be verified that the proof goes through for any choice of distinct $a_1,a_2,\ldots,a_n$ in $\{\gamma^{si} \mid 0 \leq i \leq (q-1)/s -1 \}$, and $\ell \in \N$ (for the latter see discussion at end of Section 2.4 of \cite{GW13}).
\end{remark}

\subsection{Output list cannot contain many codewords from a low dimensional subspace}\label{subsec:FRS-2}

To show that the output list $\calL$ cannot contain too many elements from a low dimensional subspace (and to find $\calL$ in the process), we 
 first give a preliminary randomized algorithm $\ALGPRUNE$ that outputs a constant size list $\calL'$  such that any codeword of $\calL$ appears in 
$\calL'$ with a constant probability $p_0$.
This implies that $|\calL| \leq  |\calL'|/p_0$, proving the first part of Theorem \ref{thm:frs-list-rec-const-list}.
Now that we know that $|\calL|$ is small, our final algorithm simply runs $\ALGPRUNE$ $O(\frac{1}{p_0} \log|\calL|)$ times 
and returns the union of the output lists.
By a union bound, all elements of $\calL$ will appear in the union of the output lists with high probability.
This will complete the proof of the second part of Theorem \ref{thm:frs-list-rec-const-list}. 

We start by describing the algorithm $\ALGPRUNE$ and analyzing it. 
The algorithm  is given as input $S : \Domain \to {\F_q^s \choose \ell}$, an $\F_q$-affine subspace $v_0 + V \subseteq \F_q[X]$ consisting of polynomials of degree at most $d$ and of dimension at most $r$, and a parameter $\tau \in \N$.

\medskip
\putinbox{
\noindent {\bf Algorithm $\ALGPRUNE(S,v_0 + V,\tau)$}
\begin{enumerate}
\item Initialize $\calL' = \emptyset$.
\item Pick $b_1, b_2, \ldots, b_\tau \in \Domain$ independently and uniformly at random.
\item For each choice of $y_1 \in S(b_1), y_2 \in S(b_2), \ldots, y_\tau \in S(b_\tau)$:
\label{lookforP}
\begin{itemize}
\item  If there is exactly one codeword $P(X) \in v_0 + V$ such that $P^{[s]}(b_j) = y_{j}$ for all $j \in [\tau] $, then:
$$ \calL' \leftarrow \calL' \cup \{ P(X) \}.$$
\end{itemize}
\item Output $\calL'$.
\end{enumerate}
}
\begin{lemma}\label{lem:algprune}

The algorithm $\ALGPRUNE$ runs in time $\poly(\log q, s, n, \ell^{\tau})$, and outputs a list $\calL'$ containing at most $\ell^{\tau}$ polynomials, such that any polynomial  $P(X) \in v_0 + V$ 
with $\dist(\FRSEnc_s(P), S)\leq \alpha$
appears in $\calL'$ with probability at least 
$$(1-\alpha)^\tau - r  \left(\frac{d} {(s-r)n}\right)^\tau.$$ 

\end{lemma}

\medskip
\noindent
\begin{proof}
We clearly have that $|\calL'| \leq \ell^{\tau}$, and that the algorithm has the claimed running time. 
Fix a polynomial $\hat P \in v_0 + V$ such that $\dist(\FRSEnc_s(\hat P) ,S)\leq \alpha$, we shall show below that $\hat P$ belongs to $\calL'$ with probability at least 
$$(1-\alpha)^\tau - r  \left(\frac{d} {(s-r)n}\right)^\tau.$$

Let $E_1$ denote the event that $\hat P^{[s]}(b_j) \in S(b_j)$ for all $j \in [\tau]$. 
Let $E_2$ denote the event that for all nonzero polynomials $Q \in V $
there exists some $j \in [\tau]$ such that $Q^{[s]}(b_j) \neq 0$.
By the assumption that $\dist(\FRSEnc_s(\hat P) ,S)\leq \alpha$, we readily have that 
$$\Pr[E_1] \geq (1-\alpha)^{\tau}.$$
Claim \ref{clm:2} below also shows that 
$$ \Pr[E_2] \geq 1 - r \left( \frac{d} {(s-r)n}\right)^{\tau}.$$
So both $E_1$ and $E_2$ occur with probability at least
$$(1-\alpha)^\tau - r \left(\frac{d} {(s-r)n}\right)^\tau.$$

If $E_2$ occurs, then for every choice of $y_1 \in S(b_1), y_2 \in S(b_2), \ldots, y_\tau \in S(b_2)$,
there can be at most one polynomial $P(X) \in v_0 + V$ such that $P^{[s]}(b_j) = y_j$ for all $j \in [\tau]$
(otherwise, the difference $Q = P_1 - P_2 \in V$ of two such distinct polynomials would
have $Q^{[s]}(b_j) = 0$ for all $j\in [\tau]$, contradicting $E_2$).
If $E_1$ also occurs, then in the iteration of Step~\ref{lookforP} where $y_{j} = \hat P^{[s]}(b_j)$ for each $j \in [\tau]$, the algorithm will
take  $P = \hat P$, and thus $\hat P$ will be included in $\calL'$. This completes the proof of the lemma.
\end{proof}

It remains to prove the following claim.
\begin{claim}\label{clm:2}
$$ \Pr[E_2] \geq 1 - r  \left(\frac{d} {(s-r)n}\right)^{\tau}.$$
\end{claim}

The proof of the claim relies on the following theorem from~\cite{GK14}. 
\begin{theorem}[\cite{GK14}, Theorem 14]\label{thm:subspace-design}
Let $W \subseteq \F_q[X]$ be a linear subspace of polynomials
of degree at most $d$. Suppose $\dim(W) = t \leq s$. 
Let $a_1,a_2, \ldots,a_n$  be distinct elements in $\{\gamma^{si} \mid 0 \leq i \leq (q-1)/s -1 \}$, and for $i \in [n]$ let
$$H_i=\{P(X) \in \F_q[X] \mid P(\gamma^j a_i)=0 \;\; \forall j \in \{0,1,\ldots,s-1\}\}.$$
 Then
$$ \sum_{i=1}^n  \dim(W\cap H_i) \leq \frac{d}{s-t+1}\cdot t.$$
\end{theorem}

\begin{proof}[Proof of Claim \ref{clm:2}]
For $0 \leq j \leq \tau$, let $$ V_j :=  V \cap H_{i_1}\cap H_{i_2} \cap \ldots \cap H_{i_j},$$ and $t_j: =  \dim( V_j ).$
Observe that $r = t_0 \geq t_1 \geq \ldots \geq t_\tau$, and that event $E_2$ holds if and only if $t_{\tau} =0$. 

By Theorem \ref{thm:subspace-design}, 
$$
\E[t_{j+1} \mid t_j = t] 
= \E_{i \in [n]}[\dim(V_j \cap H_i) \mid \dim(V_j)=t]
\leq \frac{t}{s-t+1} \cdot \frac{d}{n}
\leq t \cdot \frac{d}{(s-r) n}. 
$$

Thus
$$ \E[t_{j+1}] \leq \E[t_j] \cdot \frac{d} {(s-r)n},$$
and
$$\E[t_\tau] \leq \E[t_0] \cdot \left(\frac{d} {(s-r)n}\right)^\tau = r \left( \frac{d} {(s-r)n}\right)^\tau.$$

Finally, by Markov's inequality this implies in turn that 
$$\Pr[E_2]=\Pr[t_{\tau}=0] = 1- \Pr[t_{\tau}\geq 1] \geq 1- r \left( \frac{d} {(s-r)n}\right)^\tau.$$
\end{proof}

\subsection{Proof of Theorem \ref{thm:frs-list-rec-const-list}}\label{subsec:FRS-main}

We now prove Theorem \ref{thm:frs-list-rec-const-list} based on Theorem \ref{thm:output-list-low-dim} and Lemma \ref{lem:algprune}.

\begin{proof}[Proof of Theorem \ref{thm:frs-list-rec-const-list}]
Let $S : \Domain \to {\F_q^s \choose \ell}$ be the received sequence of input lists. We would like to find a list $\calL$ of size $ \left(\frac {\ell} {\epsilon}\right)^{O\left(\frac 1 \epsilon \log(\ell/\epsilon)\right)}$ that contains all polynomials $P(X)$ of degree at most $d$ with $\dist(\FRSEnc_s(P), S) \leq \alpha$.

Let $v_0 + V$ be the subspace found by the algorithm of
 Theorem \ref{thm:output-list-low-dim}   for $S$ and $r = \frac{4\ell}{\epsilon}$ (so $r \leq \frac{1}{4} \epsilon s$ by assumption that $s \geq 16\ell /\epsilon^2$). Note that for this choice of $r$ the RHS of (\ref{eqdecoderadius}) is at least
$$1- \frac{\epsilon}{4} -  \frac{1}{1 - \epsilon/4} \cdot \frac{d} {sn} \geq 1-  \frac{d} {sn} - \epsilon = \alpha,$$ and so all polynomial $P(X)$ of degree
at most $d$ with $\dist(\FRSEnc_s(P),S)\leq \alpha$ are included in $V$.

Next we invoke Lemma \ref{lem:algprune} with $S$, $v_0 + V$ and $\tau = O(\frac 1 \epsilon \log (\ell/\epsilon))$. Then the algorithm $\ALGPRUNE$ 
returns a list $\calL'$ of size at most
$\ell^\tau$ such that
 each polynomial $P(X)$ of degree at most $d$ with $\dist(\FRSEnc_s(P), S) \leq \alpha$ is included in $\calL'$ with probability $p_0$, which is at least
\begin{eqnarray*}
(1-\alpha)^\tau - r  \left( \frac{d} {(s-r)n}\right)^\tau 
&\geq&  (1-\alpha)^\tau - r \left( \frac{1}{1-\epsilon/4} \cdot \frac{d} {sn}\right)^\tau \\
&\geq& \left(1-\alpha\right)^\tau - \frac 1 2  \left(  \frac{1+\epsilon/4}{1-\epsilon/4} \cdot (1-\alpha-\epsilon)\right)^\tau \\
&\geq& \frac{1}{2}(1-\alpha)^{\tau},
\end{eqnarray*}
where the first inequality follows since $r \leq \frac{1}{4} \epsilon s$,  and the second inequality holds since 
$r  = \frac{4 \ell} {\epsilon} \leq \frac 1 2 \cdot (1+\frac{\epsilon}{4})^\tau$ and $\alpha = 1-  \frac{d} {sn} - \epsilon$.

The above implies in turn that 
$$|\calL| \leq \frac {|\calL'|}  {p_0} \leq 2  \left(\frac {\ell} {1-\alpha}\right)^{\tau} \leq \left(\frac {\ell} {\epsilon}\right)^{O\left(\frac 1 \epsilon \log(\ell/\epsilon)\right)}.$$   

Moreover, by running the algorithm $\ALGPRUNE$  $O(\frac 1 {p_0}  \log|\calL|) $ times and returning the union of all output lists, by a union bound, all elements of $\calL$ will appear in the union of the output lists with high probability (say, at least $0.99$). This gives a randomized list recovery algorithm with output list size $\left(\frac {\ell} {\epsilon}\right)^{O\left(\frac 1 \epsilon \log(\ell/\epsilon)\right)}$  and running  time $\poly(\log q, s,d,n,(\ell/\epsilon)^{ \log(\ell/\epsilon)/\epsilon})$. 

\end{proof}

\section{List recovering high-rate univariate multiplicity codes with constant output list size}\label{sec:Unimult}

In this section, we show that univariate multiplicity codes
of high rate can be list recovered from constant-sized input lists
with constant-sized output lists.

\subsection{Small $d$}
If the degree $d$ of the univariate multiplicity code is less than $\char(\F_q)$, the characteristic of the field
$\F_q$, then the proof from the previous section works verbatim. The only changes needed are as follows.
\begin{itemize}
\item First, use Theorem 17 from~\cite{GW13} instead of Theorem 7 from that paper, to show that the list is contained in a low-dimensional subspace.
\item Second, use Theorem 17 from~\cite{GK14} instead of Theorem 14  from that paper, to show that for a low dimensional subspace $W$, at a typical $a \in \Domain$ we have $\dim(\{P(X) \in W \mid P^{(<s)}(a) = 0\})$ is small. 
\end{itemize}
The condition $d < \char(\F_q)$ is used in both steps.
These changes lead to the following theorem.
\begin{theorem}[List recovering univariate multiplicity codes over prime
fields with $d < \char(\F_q)$]\label{thm:unimult-list-rec-const-list-smalld}
Let $q$ be a prime power, and let $s,d,n$ be nonnegative integers such that $d < \char(\F_q)$ 
and $n < \char(\F_q)/s$.

Let $\epsilon >0$ and $\ell \in \N$ be such that $16 \ell /\epsilon^2 \leq s$. Then the univariate muliplicity code $\MULT^{(1)}_{q,s}(n,d)$ is $(\alpha,\ell,L)$-list recoverable for $\alpha = 1- d/(sn) -\epsilon$ and $L =  \left( \frac {\ell} {\epsilon} \right)^{O\left(\frac{1}{\epsilon} \log \frac{\ell}{\epsilon}\right) }$.

Moreover, there is a randomized algorithm that list recovers $\MULT^{(1)}_{q,s}(n,d)$ with the above parameters in time $\poly(\log q, s,d,n,(\ell/\epsilon)^{ \log(\ell/\epsilon)/\epsilon})$. 
\end{theorem}

\begin{remark}[Fields of characteristic zero.]
The exact same techniques also work over fields of characteristic $0$.
We state the the analogous combinatorial statement over $\mathbb C$, which may be of independent interest.
\begin{theorem}
Let $\epsilon > 0$.
Let $d, n > 0 $ be an integer.
Let $\Domain \subseteq \mathbb C$ with $|\Domain| = n$.
Let $s > \frac{16\ell}{\epsilon^2}$ be an integer.
Let $\alpha \leq 1 - \frac{d}{sn} - \epsilon$.

Let $S: \Domain \to { \mathbb C^s \choose \ell}$ be arbitrary.
Then:
$$ \left|\left\{ P(X) \in \mathbb C[X] \mid \deg(P) \leq d  \mbox{ s.t. } | \{ a \in \Domain \mid P^{(<s)}(a) \in S(a) \}| \geq (1-\alpha) n \right\}\right| \leq \left(\frac{\ell}{\epsilon}\right)^{O\left(\frac{1}{\epsilon} \log \frac{\ell}{\epsilon}\right)}.$$
\end{theorem}
In particular, when $\delta > 0$ is a constant and we take $d = (1 - \delta)n$, $\eps = \delta$, $s = O(\ell/\delta^2)$ and $\alpha=0$, then the bound on $L$ is independent of $n$.

\end{remark}

\subsection{Large $d$}

Theorem~\ref{thm:unimult-list-rec-const-list-smalld} works when $d < q$, but 
for application to multivariate multiplicity codes, however, it is important
that we can list-recover  univariate multiplicity codes when the evaluation
set $\Domain$ equals all of $\F_q$ (i.e., whole-field univariate multiplicity codes). 

For the rest of this section we assume
that $\Domain = \F_q$, and hence that $n = q$. 
In this setting, for the rate to be high,
we would also like the degree $d$ to be close to $sq$, and thus $ \gg q \geq \char(\F_q)$.
This precludes use of the Theorem~\ref{thm:unimult-list-rec-const-list-smalld}.

Instead, we will dig deeper into the proof to see what can be salvaged
when $d$ is larger than $q$.
For the first step, it turns out that if $d$ is only moderately larger than $\char(\F_q)$,
then the list can be captured inside a moderately small dimensional subspace.
Thus if we make $q$ prime, so that $q = \char(\F_q)$, then this step can still work
for $d < sq$ and $s$ not too large.
The second step uses the $d < \char(\F_q)$ condition more essentially. By a 
reworking of several algebraic tools used in the proof of
Theorem 17 from~\cite{GK14},
we prove a generalization of it to handle polynomials of degree $> q$.
This generalization will only apply to subspaces $W$ of a special kind
(``\Cool{}" subspaces). The list-recoverability we show here
is quantitatively weaker (in terms of the fraction of errors that can
be tolerated) than the  results we proved in Theorem~\ref{thm:frs-list-rec-const-list} and Theorem~\ref{thm:unimult-list-rec-const-list-smalld}.
Nevertheless, this form of the result still suffices to
needed to show high-rate {\em local} list-recoverability of multivariate
multiplicity codes in the following section.

\begin{theorem}
\label{thm:full-unimult-list-rec}
Let $\delta > 0$.
Let $q$ be a prime, and let $s,d, \ell$ be nonnegative integers such that $s<q$ and $d < (1-\delta)sq$, and $1 \leq \ell < \frac{\delta^2 s}{16}$.
Suppose $\alpha < \frac{1}{2s}$.

Then the {\bf whole-field} univariate muliplicity code $\MULT^{(1)}_{q,s}(d)$ is $(\alpha,\ell,L)$-list recoverable for $L =  \ell^{O\left(s \log s\right) }\cdot s^{O(1)}$.

Moreover, there is a randomized algorithm that list recovers $\MULT^{(1)}_{q,s}(d)$ with the above parameters in time $\poly(q, \ell^{O(s \log s)})$. 
\end{theorem}

The proof of this theorem is again in two steps.

First we use
the linear-algebraic approach to list-recovering univariate multiplicity codes~\cite{GW13}
to show that the list is contained in a subspace.
Technically, we need to redo this proof using some additional algebraic ideas, because
it will be important for us to deal with polynomials of degree $ \gg q$, 
and~\cite{GW13} only worked for polynomials of degree at most $\char(\F_q)$.
As a consequence, the low-dimensional subspace will be of noticeably higher
dimension than in~\cite{GW13}, but will be of a special form.

Next we show that the output list $\calL$ cannot contain too many 
elements from a low-dimensional subspace of this special form.
As before, we will do this via a randomized algorithm.

\subsection{Output list is contained in a special subspace}

Suppose we are given a received word $S: \F_q \to {\F_q^{s} \choose \ell}$
for the univariate multiplicity code with degree $d$ and 
multiplicity parameter $s < q$.

The following theorem is essentially Lemma 14 from~\cite{GW13}. (This lemma
is part of the proof of Theorem 17 of~\cite{GW13},
which is used in the proof of our Theorem~\ref{thm:unimult-list-rec-const-list-smalld} above). 
The theorem gives a special affine subspace $v_0 + V$ of $\F_q[X]$
which contains all $f(X)$ whose codeword $\MultEnc_s(f)$
is close to $S$.

The main differences between the following theorem and Lemma 14 of ~\cite{GW13} are:
(1) we need to talk about list-recovery, not just list-decoding,
and (2) we work with Hasse derivatives, while~\cite{GW13} works with standard derivatives.
Both differences are minor; for completeness we include a proof in the appendix.

\begin{theorem}
\label{thm:linlistdec-mult}
Let $S : \F_q \to {\F_q^s \choose \ell}$.
Let 
\begin{equation}\label{mult-eqdecoderadius}
\alpha < 1- \frac{\ell}{r+1} - \frac{r}{r+1} \cdot \frac{s}{s-r+1} \cdot \frac{d} {sq}. 
\end{equation}
Let
$$ \mathcal L = \{ g(X) \in \F_q[X] \mid \deg(g) \leq d\mbox{ and } \dist(\MultEnc_s(g), S) \leq \alpha \}.$$

There is an algorithm $A$, which when given as input $r$,
finds polynomials $A(X), B_0(X), \ldots, B_{r-1}(X)$ such that
the affine space:
$$ v_0 + V = \{ f(X) \in \F_q[X] \mid \deg(f) \leq d \mbox{ and } A(X) + \sum_{i=0}^{r-1} B_i(X) f^{(i)}(X) = 0 \} $$
satisifies:
$$\mathcal L \subseteq v_0 + V.$$
\end{theorem}

\subsection{Special subspaces}

In this subsection, we study certain special linear spaces of polynomials.
In the next subsection we show how this is relevant to the kinds
of spaces $v_0 + V$ returned by the algorithm of Theorem~\ref{thm:linlistdec-mult}.

\begin{definition}[\Cool{}]
A subspace $W \subseteq \F_q[X]$ consisting of polynomials
of degree at most $d$ is called
\Cool{} if
for every $f(X) \in W$ with $\deg(f) \leq d-q$,
we have
$$ f(X) \cdot X^q \in W.$$

For a \Cool{} subspace $W$, we define the $q$-dimension by
$$ \qdim(W) = \left|\{ \deg(f) \mod q \mid f \in W \setminus \{0\} \}\right|.$$
\end{definition}

Observe that if $W$ is \Cool{} with $d\leq sq$,
 then $$\dim(W) \leq s \cdot \qdim(W).$$

The next lemma gives a nice basis for every \Cool{} subspace.

\begin{lemma}
\label{lem:coolstructure}
Suppose $W\subseteq \F_q[X]$ is  \Cool.
Then there exist $f_1, \ldots, f_{t'} \in W$ such that:
\begin{enumerate}
\item $\deg(f_i) \not\equiv \deg(f_j) \mod q$ for all $i \neq j$.
\item Every $f \in W$ can be uniquely written as:
$$ f(X) = \sum_{i=1}^{t'} C_i(X^q) f_i(X),$$
where for all $i$, $C_i(Y) \in \F_q[Y]$ with $q \cdot \deg(C_i) + \deg(f_i) \leq d$.
\end{enumerate}
Furthermore, we have $\qdim(W) = t'$.
\end{lemma}
\begin{proof}
Let $f_1$ be the lowest degree nonzero element of $W$,
and let $M_1 = \{ C_1(X^q) f_1(X) \mid C_1(Y) \in \F_q[Y] \mbox{ and } q \deg(C_1) + \deg(f_1) \leq d \}$.
Note that $M_1 \subseteq W$. If $M_1 = W$ we are done with $r' = 1$.

Otherwise, we proceed. 
Let $f_2$ be the lowest degree nonzero element of $W \setminus M_1$.
Observe that $\deg(f_2) \not\equiv \deg(f_1) \mod q$: otherwise
for some $a \in \F_q$ and $b \geq 0$, we would have that
$$f_2'(X) = f_2(X) - a X^{bq} f_1(X) \in f_2 + M_1$$
is an element of $W \setminus M_1$ with even lower degree than $f_2(X)$,
contradicting the choice of $f_2$.
Define $M_2 = \{ C_1(X^q) f_1(X) + C_2(X^q) f_2(X) \mid \mbox{For $i \leq 2$ we have } C_i(Y) \in \F_q[Y] \mbox{ and } q \deg(C_i) + \deg(f_i) \leq d \}$.
Note that $M_2 \subseteq W$. If $M_2 = W$, we are done with $r' = 2$.

Repeating this argument, we get polynomials
$f_1, f_2, \ldots, f_{t'} \in \bar{V}$ with degrees $d_1, d_2, \ldots, d_{t'}$ satisfying:
\[ d_i \not\equiv d_j \mod q \qquad \forall i \neq j \]
This implies that the polynomials $X^{cq} f_i(X)$ all have distinct degrees, and are thus linearly independent.
Furthermore, we  have that $W$ is equal to
$$M_{t'} = \left\{ \sum_{i=1}^{t'} C_i(X^q) f_i(X) \mid \mbox{For $i \leq t'$ we have } C_i(Y) \in \F_q[Y] \mbox{ and } q \deg(C_i) + \deg(f_i) \leq d \right\}.$$

In particular, if we define $j_i = \left\lfloor\frac{d - \deg(f_i)}{q}\right\rfloor$, then
$$ \{ X^{cq} f_i(X) \mid q \leq i \leq t', 0 \leq c \leq j_i \}$$
is a basis for $W$.
\end{proof}

The next theorem is a variant of Theorem 14 from~\cite{GK14} for \Cool{} subspaces of polynomials of degree $d > q$. 
\begin{theorem}
\label{thm:cool-subspace-design}
Let $W \subseteq \F_q[X]$ be a \Cool{} linear subspace of polynomials of degree
at most $d$.  Suppose $\dim(W) = t$ and $\qdim(W) = t' \leq s$.

Suppose $d \leq (s-t')q$.
Then:
$$ \mathbb E_{b \in \F_q} [\dim(W\cap H_b)] \leq \left( 1-1/s\right) \cdot t.$$
\end{theorem}
\begin{proof}
Let $f_1, \ldots, f_{t'}$ be those given by Lemma~\ref{lem:coolstructure}.
Let $d_i = \deg(f_i)$, and assume that $d_1 \leq d_2 \leq \ldots \leq d_{t'}$.
Let $U = \vspan\{f_1, \ldots, f_{t'} \}$.

For $b \in \F_q$, consider the map $\Phi_b: W \to \F_q^s$ given by
$$\Phi(f) = f^{(<s)}(b).$$
Then $\ker(\Phi_b) = W \cap H_b$.
Thus $\dim(W \cap H_b) = \dim(W) - \dim(\Phi(W))$.
Similarly, 
$\dim(U \cap H_b) = \dim(U) - \dim(\Phi(U))$.
Thus $$\dim(W \cap H_b) \leq \dim(W) - \dim(U) + \dim(U \cap H_b).$$
Below we will show that 
\begin{equation}
\label{eqdimu}
\dim(U) - \mathbb E_b[\dim(U \cap H_b)] \geq  \frac{t}{s -t' + 1}.
\end{equation}
Assuming this, we get that
$$\mathbb E[\dim(W \cap H_b) ] \leq t - \frac{t}{s -t'+1} \leq
\left( 1 - \frac{1}{s-t' + 1} \right) t \leq \left( 1- \frac{1}{s}\right) t,$$
which is what we wanted to prove.

We now prove Equation~\eqref{eqdimu}.
Let $D =\sum_{i=1}^{t'} d_i$.
By Item 2 of Lemma~\ref{lem:coolstructure}, we have that:
\begin{align*}
t &= \sum_{i=1}^{t'} \left(1+ \left\lfloor \frac{d-d_i}{q} \right\rfloor \right) \\
&\leq t'\left( 1+ \frac{d}{q} \right) - \frac{D}{q}.
\end{align*}

Let $Q(X)$ be the determinant of the Wronskian matrix of $f_1, \ldots, f_{t'}$.
By Lemma~\ref{lem:wronskmodq}, $Q(X)$ is a nonzero polynomial.
We have:
$$ \deg(Q) \leq D.$$
By Claim 19 from~\cite{GK14}, we have that
$$ \mult(Q, b) \geq (s- t' + 1) \cdot \dim(U \cap H_b).$$
Now using the fact that $\mathbb E_{b\in\F_q}[ \mult(Q,b)] \leq \frac{\deg(Q)}{q} \leq \frac{D}{q}$, we get:
$$\mathbb E_{b \in \F_q} [ \dim(U \cap H_b) ] \leq \frac{1}{s- t'+1} \mathbb E_{b \in \F_q} [ \mult(Q,b) ] \leq \frac{1}{s-t'+1} \cdot \frac{D}{q}.$$
Thus
$$\dim(U) - \mathbb E_{b \in \F_q} [ \dim(U \cap H_b) ] \geq t' -  \frac{1}{s-t'+1} \frac{D}{q}.$$

Thus 
\begin{align*}
(s-t'+1) \cdot \left( \dim(U) - \mathbb E_{b \in \F_q} [ \dim(U \cap H_b) ]\right) &\geq (s-t'+1)t' - \frac{D}{q}\\
&\geq t'\left( 1+ \frac{d}{q} \right) - \frac{D}{q} \qquad\mbox{By the assumption $d \leq (s-t')q$}\\
&\geq t.
\end{align*} 
This completes the proof of the lemma.
\end{proof}

\subsection{Properties of the space of solutions of a linear differential equation}

Let $v_0 + V$ be the affine space of low degree solutions
to a linear differential equation:
$$ v_0 + V = \left\{ f(X) \in \F_q[X] \mid \deg(f) \leq d \mbox{ and }A(X) + \sum_{i=0}^{r-1} B_i(X) f^{(i)}(X) = 0\right\}.$$ 
We now prove some properties of $V$ that will help us
in the second step of our list-decoding algorithm.

First we note that $V$ is \Cool{}.
\begin{lemma}
\label{lem:Vcool}
$V$ is \Cool{}.
\end{lemma}
\begin{proof}
Observe that:
$$V = \left\{f(X) \in \F_q[X] \mid \deg(f) \leq d \mbox{ and } \sum_{i=0}^r B_i(X) f^{(i)}(X) = 0\right\}.$$
Now take any $f(X) \in V$. We want to show that $g(X) = X^q f(X) \in V$.
For any $i < q$, we have:
\begin{align*}
g^{(i)}(X) &= \sum_{i' = 0}^i (X^{q})^{(i')} \cdot f^{(i-i')}(X)\\
&= \sum_{i' = 0}^i {q \choose i'}X^{q-i'} \cdot f^{(i-i')}(X)\\
&= X^{q} f^{(i)}(X).
\end{align*}
Thus (using the fact that $r < q$):
$$ \sum_{i=0}^r B_i(X) g^{(i)}(X) = X^q(\sum_{i=0}^r B_i(X) f^{(i)}(X)) = 0.$$
We conclude that $g(X) \in V$, as desired.
\end{proof}

\begin{lemma}
\label{lem:Vdim}
$$\qdim(V) \leq r-1.$$
\end{lemma}
\begin{proof}
Apply Lemma~\ref{lem:coolstructure} to $V$ to obtain $f_1, \ldots, f_{r'} \in V$, where $r' = \qdim(V)$.

Suppose $r' \geq r$.
Since $f_1, \ldots, f_{r} \in V$
we have that for each $j \in [r]$,
$$ \sum_{i=0}^{r-1} B_i(X) f_j^{(i)}(X) = 0 .$$
This means that the Wronskian matrix of
$(f_1, \ldots, f_{r})$
is singular.
However, Lemma~\ref{lem:wronskmodq} shows that the Wronskian matrix is nonsingular, a contradiction.
This completes the proof.
\end{proof}

The next two lemmas and the following corollary are trivial and we omit the proofs.

For $b \in \F_q$, let $$H_b = \{f(X) \in \F_q[X]\mid \deg(f) \leq d \mbox{ and } f^{(<s)}(b) = 0\}.$$
\begin{lemma}
\label{lem:Hcool}
For each $b \in \F_q$,  $H_b$ is \Cool{}.
\end{lemma}

\begin{lemma}
\label{lem:intcool}
The intersection of \Cool{} subspaces is \Cool{}.
\end{lemma}

\begin{corollary}
\label{cor:Vjcool}
Given elements $b_1, \ldots, b_j \in \F_q$, the space
$V \cap H_{b_1} \cap H_{b_2} \ldots H_{b_j}$ is \Cool{}.
\end{corollary}

\subsection{Pruning the list}
\label{sec:unimultprune}

Having developed the relevant machinery about \Cool{} subspaces, we
can now describe and analyze the second part of the list-decoding algorithm for univariate multiplicity codes. 
Below we give the algorithm $\ALGPRUNEMULT$, and after that we analyze it. 
The algorithm  is given as input $S : \F_q \to {\F_q^s \choose \ell}$, an $\F_q$-affine subspace $v_0 + V \subseteq \F_q[X]$, and a parameter $\tau \in \N$.

\medskip
\putinbox{
\noindent {\bf Algorithm $\ALGPRUNEMULT(S,v_0 + V,\tau)$}
\begin{enumerate}
\item Initialize $\calL' = \emptyset$.
\item Pick $b_1, b_2, \ldots, b_\tau \in \F_q$ independently and uniformly at random.
\item For each choice of $y_1 \in S(b_1), y_2 \in S(b_2), \ldots, y_\tau \in S(b_\tau)$:
\label{lookforP-mult}
\begin{itemize}
\item  If there is exactly one codeword $P(X) \in v_0 + V$ such that $P^{(<s)}(b_j) = y_{j}$ for all $j \in [\tau] $, then:
$$ \calL' \leftarrow \calL' \cup \{ P(X) \}.$$
\end{itemize}
\item Output $\calL'$.
\end{enumerate}
}

\begin{lemma}\label{lem:algprune-mult}
Suppose $v_0 + V$ is an affine space of polynomials of degree at most $d$.
Suppose $V$ is \Cool{} and $\qdim(V) = r'$. Suppose $d \leq (s - r') q$.

Then the algorithm $\ALGPRUNEMULT$ runs in time $\poly(q, s, \ell^{\tau})$, and outputs a list $\calL'$ containing at most $\ell^{\tau}$ polynomials, such that any polynomial  $P(X) \in v_0+ V$ 
with $\dist(\MultEnc_s(P), S)\leq \alpha$
appears in $\calL'$ with probability at least 
$$(1-\alpha)^\tau - r' s (1 - 1/s)^{\tau}.$$ 
\end{lemma}
\begin{proof}
We clearly have that $|\calL'| \leq \ell^{\tau}$, and that the algorithm has the claimed running time. 
Fix a polynomial $\hat P \in v_0 + V$ such that $\dist(\MultEnc_s(\hat P) ,S)\leq \alpha$. Below we will show that $\hat P$ belongs to $\calL'$ with probability at least 
$$(1-\alpha)^\tau - r's (1 - 1/s)^{\tau}.$$

Let $E_1$ denote the event that $\hat P^{(<s)}(b_j) \in S(b_j)$ for all $j \in [\tau]$. 
Let $E_2$ denote the event that for all nonzero polynomials $Q \in V$
there exists some $j \in [\tau]$ such that $Q^{(<s)}(b_j) \neq 0$.
By assumption that $\dist(\MultEnc_s(\hat P) ,S)\leq \alpha$, we readily have that 
$$\Pr[E_1] \geq (1-\alpha)^{\tau}.$$
Claim \ref{clm:mult2} below also shows that 
$$ \Pr[E_2] \geq 1 - r's \cdot (1- 1/s)^{\tau}.$$
So both $E_1$ and $E_2$ occur with probability at least
$$(1-\alpha)^\tau - r's \cdot (1- 1/s)^{\tau}.$$

If $E_2$ occurs, then for every choice of $y_1 \in S(b_1), y_2 \in S(b_2), \ldots, y_\tau \in S(b_2)$,
there can be at most one polynomial $P(X) \in v_0 + V$ such that $P^{(<s)}(b_j) = y_j$ for all $j \in [\tau]$
(otherwise, the difference $Q = P_1 - P_2 \in V$ of two such distinct polynomials would
have $Q^{(<s)}(b_j) = 0$ for all $j\in [\tau]$, contradicting $E_2$).
If $E_1$ also occurs, then in the iteration of Step~\ref{lookforP} where $y_{j} = \hat P^{(<s)}(b_j)$ for each $j \in [\tau]$, the algorithm will
take  $P = \hat P$, and thus $\hat P$ will be included in $\calL'$. This completes the proof of the lemma.
\end{proof}

It remains to prove the following claim.
\begin{claim}\label{clm:mult2}
$$ \Pr[E_2] \geq 1 - r's \cdot (1- 1/s)^{\tau}.$$
\end{claim}
\begin{proof}[Proof of Claim \ref{clm:mult2}]

We will use Theorem~\ref{thm:cool-subspace-design}. 

Since $\qdim(V) \leq r'$, we have $\dim(V) \leq r's$.
For $0 \leq j \leq \tau$, let $$ V_j :=  V \cap H_{b_1}\cap H_{b_2} \cap \ldots \cap H_{b_j},$$ and $t_j: =  \dim( V_j ).$
Observe that $ r's \geq t_0 \geq t_1 \geq \ldots \geq t_\tau$, and that event $E_2$ holds if and only if $t_{\tau} =0$.

By Corollary~\ref{cor:Vjcool},
all the $V_j$ are \Cool{}. Since $V_j \subseteq V$, we have
$\qdim(V_j) \leq \qdim(V) = r'$.

We now study the distribution of $t_0, \ldots, t_{\tau}$.
Since $\qdim(V_j) \leq r'$ and $d \leq (s-r')q$, we 
may apply Theorem~\ref{thm:cool-subspace-design}.
We get:
$$
\E[t_{j+1} \mid t_j = t] 
= \E_{b_{j+1} \in \F_q}[\dim(V_j \cap H_{b_j}|) \mid \dim(V_j)=t]
\leq (1 - 1/s) \cdot t.
$$

Thus
$$ \E[t_{j+1}] \leq (1 - 1/s) \cdot \E[t_j],$$
and
$$\E[t_\tau] \leq \left(1-1/s\right)^\tau \cdot \E[t_0] \leq \left(1-1/s\right)^\tau \cdot r's.$$

Finally, by Markov's inequality this implies in turn that 
$$\Pr[E_2]=\Pr[t_{\tau}=0] = 1- \Pr[t_{\tau}\geq 1] \geq 1- r's \left(1-1/s\right)^\tau .$$
\end{proof}

\subsection{Proof of Theorem \ref{thm:full-unimult-list-rec}}\label{subsec:mult-main}

We now prove Theorem \ref{thm:full-unimult-list-rec} based on Theorem \ref{thm:linlistdec-mult} and Lemma \ref{lem:algprune-mult}.

\begin{proof}[Proof of Theorem \ref{thm:full-unimult-list-rec}]
Let $S : \Domain \to {\F_q^s \choose \ell}$ be the received sequence of input lists. We would like to find a list $\calL$ of size
$ \ell^{O\left(s \log s\right)} \cdot s^{O(1)}$ that contains all polynomials $P(X)$ of degree at most $d$ with $\dist(\MultEnc_s(P), S) \leq \alpha$.

Let $v_0 + V$ be the affine subspace found by the algorithm of
 Theorem \ref{thm:linlistdec-mult}   for $S$ and $r = 4\ell/\delta$ (and so $s > \frac{4r}{\delta}$ by our assumption that $s > 16\ell/\delta^2$). Note that for this choice of $r$ the RHS of (\ref{mult-eqdecoderadius}) is at least
$$1 - \frac{\delta}{4} - \frac{1}{1- \delta/4} (1-\delta) > \frac{\delta}{3} > \alpha,$$
and so all polynomial $P(X)$ of degree at most $d$ with $\dist(\MultEnc_s(P),S)\leq \alpha$ are included in $v_0 + V$. 

By Lemma~\ref{lem:Vdim}, $\qdim(V) \leq r$.
Since $V$ is \Cool{}, $\qdim(V) \leq r$, and
$$d \leq (1-\delta)sq \leq (s- 4r) q \leq (s-r)q,$$
we may invoke Lemma \ref{lem:algprune-mult}.
It tells us that 
algorithm $\ALGPRUNEMULT$ with inputs $S$, $v_0 + V$ and $\tau = O(s \log(2rs))$ returns a list
$\calL'$ of size at most $\ell^\tau$ such that
 each polynomial $P(X)$ of degree at most $d$ with $\dist(\MultEnc_s(P), S) \leq \alpha$ is included in $\calL'$ with probability $p_0$, where:
\begin{align*}
p_0 &\geq (1-\alpha)^\tau - rs  \left( 1 - \frac{1}{s}\right)^\tau \\
&\geq  (1-\frac{1}{2s})^\tau - rs \left( 1 - \frac{1}{s} \right)^\tau \\
&\geq \frac{1}{2}(1-\frac{1}{2s})^{\tau},
\end{align*}
where the first inequality holds since $\alpha \leq \frac{1}{2s}$, and the second inequality holds since 
$$(1-\frac{1}{2s}) > (1+\frac{1}{2s}) (1-\frac{1}{s}) > e^{1/6s} (1-\frac{1}{s}) > e^{\ln(2rs)/\tau} (1-\frac{1}{s}).$$

The above implies in turn that 
$$|\calL| \leq \frac {|\calL'|}  {p_0} \leq 2  \left(\frac {\ell} {1-\frac{1}{2s}}\right)^{\tau} \leq \ell^{O\left(s \log s\right)} s^{O(1)}.$$   

Moreover, by running the algorithm $\ALGPRUNEMULT$  $O(\frac 1 {p_0}  \log|\calL|) $ times and returning the union of all output lists, by a union bound, all elements of $\calL$ will appear in the union of the output lists with high probability (say, at least $0.99$). This gives a randomized list recovery algorithm with output list size $\ell^{O\left(s \log s\right)}$  and running  time $\poly(q, s,\ell^{ O(s \log s )})$. 
\end{proof}

\section{Local list-recovery of multivariate multiplicity codes}\label{sec:Multimult}

In this section we show that multivariate multiplicity codes can be locally list recovered from constant sized input lists with small (or even constant) sized output lists. 

Let $\delta> 0$ be a parameter, and $m,s$ be integers. Let $d = (1-\delta)sq$. Let $\MULT^{(m)}_{q,s}(d)$
be the corresponding multiplicity code. 
Given $\ell, \alpha$, we will demonstrate a local list recovery algorithm for multivariate multiplicity codes of degree $d$ and multiplicity parameter $s$, with input list size $\ell$ and error tolerance $\alpha$.  
Our main technical statement is Lemma~\ref{lem:multimult-params}, which we will state and prove later in this section.
However, we first state two instantiations of Lemma~\ref{lem:multimult-params}, which show that (a) multivariate multiplicity codes are locally list-recoverable with sub-polynomial query complexity and sub-polynomial list size, and (b) multivariate multiplicity codes are locally list-recoverable with polynomial query complexity and constant list size.

Our first statement establishes sub-polynomial query complexity and list size.
\begin{theorem}\label{thm:main-multimult}
Let $\eps > 0$ be sufficiently small, and let $m, \ell > 0$ be any integers.

Then there is a multivariate multiplicity code $C \subseteq \Sigma^N$ with distance $\delta = \eps/2m$ and rate at least $1 - \eps$ so that 
 $C$ is $(t,\alpha, \ell, L)$-locally-list-recoverable for 
\[ \alpha \leq \frac{\delta^2}{160 \ell} ,\]
\[ L = \ell^{O( \ell m^2 \log(\ell m/\eps)/\eps^2 ) } \cdot \left( \frac{m}{\eps} \right)^{O(1)}\] 
and
\[ t = \left( \frac{ L m \ell }{\eps} \right)^{O(m)}. \]
Moreover, we have
\[ N = \left( \frac{ L m \ell }{\eps} \right)^{O(m^2)}. \]
and 
\[ |\Sigma| = \left( \frac{ L m \ell }{\eps } \right)^{ O( (\ell m^2 / \eps^2 )^m ) }. \]

The code is explicit, and can be locally-list-recovered in time $\poly(t)$.

In particular, if $\eps, \ell$ are constant, and $\ell > 1$, we have
\[ L = m^{O(m^2)}, \qquad t = m^{O(m^3)}, \qquad |\Sigma| = m^{O(m^{2m + 2})}, \qquad  N = m^{O(m^4)}, \]
and so
\[ L = \exp\left( \sqrt{ \log(N)\cdot \log\log(N) }\right), \qquad t = \exp\left( \log^{3/4}(N) \cdot (\log\log(N))^{1/4} \right), \]
\[ |\Sigma| = \exp\left(\exp\left( \log^{1/4}(N)\cdot \log\log^{3/4}(N) \right)\right). \]

\end{theorem}

Our second theorem establishes constant output list size with polynomial query complexity.
\begin{theorem}\label{thm:main-multimult2}
Let $\delta, \eps > 0$ be sufficiently small, and choose $\tau \in (0,1)$ and $\ell \in \mathbb{N}$.

Then for infinitely many $N \in \mathbb{N}$, there is a multivariate multiplicity code $C \subseteq \Sigma^N$ with distance $\delta$ and rate at least $(1 - \delta)^{O(1/\tau)}(1 - \eps)$ so that 
 $C$ is $(t,\alpha, \ell, L)$-locally-list-recoverable for some
\[ \alpha = \Omega\left( \min\left\{ \frac{\delta^2}{\ell} , \tau^2 \eps \right\}\right),\]
with
\[ t = O_{\delta, \eps, \ell, \tau}(N^\tau), \]
\[ L = O_{\delta, \eps, \ell, \tau}(1), \]
and
\[ |\Sigma| = \poly_{\delta, \eps, \ell, \tau}(N). \]
\end{theorem}
Notice that the alphabet size in Theorem~\ref{thm:main-multimult2} is large (polynomial in $N$) and in Theorem~\ref{thm:main-multimult} is very large (super-polynomial in $N$). However, we will deal with this in Section~\ref{sec:Smallalpha} to obtain capacity-achieving codes with the roughly the same parameters $t$ and $L$, and with constant alphabet size.  

The rest of this section is devoted to the proof of Theorems~\ref{thm:main-multimult} and \ref{thm:main-multimult2}, both of which follow from the more general Lemma~\ref{lem:multimult-params} below.  We give a short overview of the approach in Section~\ref{ssec:multimult-overview}, and then flesh out the details in the subsequent three subsections.

\subsection{Overview and some more notation}\label{ssec:multimult-overview}
We first introduce some terminology that will be useful for this section.
Let $U_{m,s} = \{ \vi \in \N^m \mid \wt(\vi) < s \}$.
Let $\Sigma_{m,s} = \F_q^{U_{m,s}}$ be the alphabet of the multiplicity code.
Let $w = |U_{m,s}| = {m + s-1 \choose m}$.
For an element $z \in \Sigma_{m,s}$, and a direction
$\bb \in \F_q^m$, we define the restriction of $z$ to direction
$\bb$ (denoted $z|_{\bb}$) to equal $h \in \Sigma_{1,s}$, where given by:
$$h^{(j)} = \sum_{\wt(\bj) = j} z^{(\bj)} \bb^{\bj}$$
for each $j$ such that $0 \leq j < s$.

The local list recovering algorithm has three main subroutines that we will describe and analyze in the next three subsections.  Briefly, the three components are the following: 
\begin{enumerate}
	\item A subroutine $\RecCand$, given in Section~\ref{ssec:reccand}.  $\RecCand$ takes as input a point $\ba \in \F_q^m$, has query access to $S$, and returns a short list $Z \subset \Sigma_{m, s^*}$ of guesses for $Q^{(<s^*)}(\ba)$, where we will take $s^*$ to be some parameter larger than $s$. 
	\item An oracle machine $M^S[\ba, z]$, given in Section~\ref{ssec:oracle}.  The oracle machine $M^S[\ba,z]$ is defined using an advice string $(\ba, z)$ and has query access to $S$.  If $z = Q^{(<s^*)}(\ba)$, then with high probability over the choice of a random point $\ba$,  we weill have that
\[ \dist\left( M^S[\ba,Q^{(<s^*)}(\ba)], Q^{(<s)} \right)  \leq 4\eps. \]
	\item The final local list-recovery algorithm $\MultimultMain$, given in Section~\ref{ssec:main_lld}.  
Recall that the goal is to output a list of randomized algorithms $A_1, \ldots, A_L$ so that for each codeword $c \in \MULT_{q,s}^{(m)}(d)$ with $\dist(c, S) \leq \alpha$, with probability at least $2/3$, there exists some $j$ so that $\Prn\left[ A_j(i) = c_i \right] \geq 2/3.$  We arrive at these algorithms $A_j$ as follows.

First, the algorithms runs $\RecCand$ on a random point $\ba \in \F_q^m$ to generate a short list $Z \subseteq \Sigma_{m,s^*}$  of possibilities for $Q^{(<s^*)}(\ba)$.  Then, for each $z \in Z$, it forms the oracle machine $M^S[\ba,z]$.  At this point it would be tempting to output the list of these oracle machines, but we are not quite done: even if $z = Q^{(<s^*)}(\ba)$ corresponds to the correct advice and the choice of $\ba$ is good, for some small fraction of points $\bx$, we may still have $M^S[\ba,z](\bx) \neq Q^{(<s)}(\bx)$ with decent probability.  Fortunately, for most $\by$ this will not be the case, and so we can implement the local correction algorithm of \cite{KSY14} for multiplicity codes on top of $M^S[\ba,z]$.  This will give us our final list of randomized algorithms $A_1,\ldots, A_L$ that the local list-recovery algorithm returns.
\end{enumerate}

We flesh out the details in the next three subsections.

\subsection{The algorithm $\RecCand$}\label{ssec:reccand}
As an important subroutine of the local list recovering algorithm, 
we will implement an algorithm which we call $\RecCand$ which will have the following features. 
It will have oracle access to a function $S: \F_q^m \to {\Sigma_{m,s} \choose \ell}$.
Think of this function as assigning to each element of $ \F_q^m$ a list of size $\ell$ of alphabet symbols of the multiplicity code. 
Now suppose that $Q$ is an $m$-variate polynomial of degree at most $d$ (think of $Q$ to represent a true codeword of the multiplicity code) that ``agrees" with at least $1-\alpha$ fraction of these lists. 
On being input $\bx$, a random element of $\F_q^m$, and for some parameter $\tilds$ (think of $\tilds$ to be much larger than $s$), the algorithm $\RecCand$ will make few queries to $S$ and output a small list $Z \subseteq \Sigma_{m, \tilds}$, such that with high probability (over the choice of $\bx$ and the randomness of the algorithm), the list $Z$ contains $Q^{(<\tilds)}(\bx)$.

The main feature of this algorithm is that given oracle access to small lists, that for most coordinates agree with the evaluations of order $s$ derivatives of $Q$, can output for most coordinates, a small list that agrees with evaluations of order $\tilds$ derivatives of $Q$.

\begin{lemma}\label{lem:recovercandidates}
Let $q$ be a prime, let $\delta > 0$ and let $s,d,m$ be nonnegative integers such that $d = (1-\delta) sq$. 
Let $\alpha, \alpha'$ be parameters such that $0< \alpha< \alpha'$. 

Let $L = L(d,q,s, \ell, \alpha') $ be the list size for list recovering univariate multiplicity codes $\MULT^{(1)}_{q,s}(d)$ of degree $d$ and multiplicity parameter $s$ with input list size $\ell$ and error tolerance $\alpha'$,  and $T$ be the corresponding running time for list recovering univariate multiplicity codes of the same parameters.

Let $S: \F_q^m \to {\Sigma_{m,s} \choose \ell}$. Let $\tilds > 0$ be a parameter. 
Suppose that $q > 100 \cdot \tilds \cdot L \cdot m^2$, and suppose that 
 $Q(X_1, \ldots, X_m) \in \F_q[X_1, \ldots, X_m]$ is a polynomial 
of degree at most $d$ such that:
$$ \Pr_{\bx \in \F_q^m} [ Q^{(<s)}(\bx) \in S(\bx) ] > 1-\alpha.$$

There is an algorithm $\RecCand$ which on input $x \in  \F_q^m$ and $\tilds$, and given oracle access to $S$, makes at most $\poly\left(q \cdot (L\tilds m)^m\right)$ queries to $S$, runs in time at most $\poly\left(T \cdot q \cdot (L\tilds m)^m\right)$, and outputs a list $Z \subseteq $ of size at most $O(L)$
such that if $\bx \in \F_q^m$ uniformly at random,
then:
$$ \Pr[ Q^{(<\tilds)}(\bx) \in Z ] \geq 1 - \frac{12}{(\alpha' - \alpha) q}, $$
where the probability is over the random choice of $\bx$ and the random choices of the algorithm $\RecCand$.

\end{lemma}

The high level idea of the algorithm is as follows. On input $\bx$, we take several random lines passing through $\bx$, and 
run the univariate multiplicity list-recovery algorithm on the restrictions of the received word to those lines.
This gives us, for each of these lines, a list of univariate polynomials. For a given line, this  list of univariate polynomials 
contains candidates for $Q$ restricted to that line. In particular, this gives us candidate values for $Q(\bx)$ and
the all higher order directional derivatives of $Q$ at $\bx$ in the directions of those lines. We combine
this information about the different directional derivatives to reconstruct $Q^{(<\tilds)}(\bx)$.

This combination turns out to be a certain kind of polynomial list-recovery problem: namely list-recovery for tuples of
polynomials. The following lemma, which is proved in Appendix~\ref{app:rm}, shows how this can be done algorithmically.
\begin{lemma}[Vector-valued Reed-Muller list recovery on a grid]\label{lem:RM-vectors}
Let $\ell, \tilds, K, m$ be given parameters. Let $\F$ be a finite field. Suppose that $U \subseteq \F$ and $|U| \geq 2 {\ell} \tilds K$.
Let $\alpha < 1- \frac{1}{\sqrt{K}}$ be a parameter.

Then for every $f: U^m \to {\F^t \choose \ell}$,
if 
\begin{align*}
\calL = & \{(Q_1, Q_2, \ldots, Q_t) \in (\F[Y_1, Y_2, \ldots, Y_t])^t \mid \forall i \in [t], \deg(Q_i) \leq \tilds \mbox{ and } \\
&\Pr_{\bu \in U^m}[(Q_1(\bu), Q_2(\bu), \ldots, Q_t(\bu)) \not\in f(\bu)] < \alpha \},
\end{align*}
the following hold:
\begin{enumerate}
\item  $|\calL| \leq 2K\ell.$

\item If $K \geq m^2$ and $\alpha < 1-  \frac{m}{\sqrt{K}}$, then  there is a $\poly(|U|^m, t, \log{|F|})$-time algorithm $\mathsf{VectorRMListRecover}$ which computes $\calL$. 
\end{enumerate}
\end{lemma}

We will use
$\mathsf{VectorRMListRecover}$ as a subroutine of $\RecCand$.  To see why this is relevant, we make the following observation about Hasse derivatives.

\begin{claim}
\label{claim:hasse}
Let $Q(X_1, \ldots, X_m) \in \F_q[X_1, \ldots, X_m]$.
Let $\bx, \bb \in \F_q^m$. Let $\tilds$ be an integer.
Let $\lambda(T) = \bx + T\bb$ be the line passing through $\bx$ in direction $\bb$. 

For each $j < \tilds$, define polynomial 
$$H_{\bx, j}(Y_1, \ldots, Y_m) = \sum_{\bj : \wt(\bj) = j} Q^{(\bj)}(\bx) \bY^{\bj},$$
and let $H_{\bx}(Y_1, \ldots, Y_m) \in (\F_q[Y_1, \ldots, Y_m])^\tilds$ be the tuple of polynomials:
$$H_{\bx} = (H_{\bx, 0}, H_{\bx,1}, \ldots, H_{\bx, \tilds-1}).$$
Then we have
 $$(Q\circ \lambda)^{(<\tilds)}(0) =  Q^{(<\tilds)}(\bx)|_{\bb} = H_{\bx}(\bb) . $$
\end{claim}
\begin{proof}
We have from the definition of the Hasse derivative that
\begin{align*}
\sum_i (Q \circ \lambda)^{(i)}(0) \cdot T^i &= (Q \circ \lambda)(T) \\
&= Q( \bx + T \bb ) \\
&= \sum_{\bi} Q^{(\bi)}(\bx) \cdot \bb^\bi \cdot T^{\wt(\bi)} \\
&= \sum_{i} \left(\sum_{\bi : \wt(\bi) = i} Q^{(\bi)}(\bx) \bb^\bi\right) T^{\wt(\bi)},
\end{align*}
and so by matching coefficients we have
\[ (Q \circ \lambda)^{(j)}(0) = \sum_{\bj : \wt(\bj) = j} Q^{(\bj)}(\bx) \bb^{\bj} = H_{\bx, j}(\bb). \]
This implies the claim.
\end{proof}
Thus, given information about $(Q \circ \lambda)^{(<\tilds)}(0)$ for various lines $\lambda$ and for some $\tilds$, we have information about the tuple of polynomials $H_{\bx}(\bY)$, evaluated at many different points $\bb$.  It is on these polynomials that we will use Lemma~\ref{lem:RM-vectors}.

Now we present our main subroutine $\RecCand$, and analyze it below.
\medskip
\putinbox{
\noindent{\bf Main Subroutine $\RecCand$.}
\begin{itemize}
\item Oracle access to $S: \F_q^m \to {\Sigma_{m,s} \choose \ell}$.
\item {\bf INPUT:} $\bx \in \F_q^m$, parameter $\tilds \in \mathbb N$.
\item The goal is to recover a small list of candidates for $Q^{(<\tilds)}(\bx)$.

\end{itemize}
\begin{enumerate}
\item Let $U \subseteq \F_q$ be a set of size $100 \tilds L m^2$.
\item Let $\bb \in \F_q^m$ be picked uniformly at random.
\item Let $B = \{ \bb_\bu = \bb + \bu \mid \bu \in U^m \}$.
\item For each $\bu \in U^m$:
\begin{enumerate}
\item Let $\lambda_\bu(T)$ be the line $\lambda_\bu(T) = \bx + T \bb_\bu$.
\item Consider the restriction $S_\bu: \F_q \to {\Sigma_{1,s} \choose \ell}$
of $S$ to $\lambda_\bu$.
Formally:
$$S_\bu = S \circ \lambda_\bu(t) = \{ z|_{\bb_{\bu}} \mid z \in S(\lambda_{\bu}(t)) \}.$$
\item Run the univariate list recovery algorithm on $S_\bu$ with error-tolerance $\alpha'$
for degree $d$ polynomials to obtain a list $\calL_{\lambda_\bu} \subseteq \F_q[T]$.
\end{enumerate}.

\item Define a function $f: U^m \to {\F_q^{\tilds} \choose L}$ as follows.
For each $\bu \in U^m$, define
$$f(\bu) =  \{ P^{(< \tilds)}(0)  \mid P(T) \in \calL_{\lambda_\bu} \}.$$

\item Let $\calL'$ be the set of all $\tilds$-tuples of polynomials

$$(Q'_j(Y_1, \ldots, Y_m))_{j=0}^{\tilds -1}$$
where $Q'_j$ is homogeneous of degree $j$, and such that
$$(Q'_j(\bu))_{j=0}^{\tilds -1} \in f(\bu)$$
for at least $2/3$ fraction of the $\bu \in U^m$.

Obtain this list $\calL'$ by running the $\mathsf{VectorRMListRecover}$ (from Lemma~\ref{lem:RM-vectors}) for $\tilds$-tuples of polynomials of 
degree $\leq \tilds$, where the evaluation points are $U^m$. Then prune the resulting list $\calL$ to ensure that
for each member $$(Q'_j(Y_1, \ldots, Y_m))_{j=0}^{\tilds -1}$$ of the list $\cal L$, and for each $j$ such that $0 \leq j \leq \tilds -1$, 
 $Q'_j$ is homogeneous of degree $j$. This pruned list is $\calL'$.

\item  For each $$(Q'_j(\bY))_{j=0}^{\tilds -1} \in \calL',$$
let $$(P_j(\bY))_{j=0}^{\tilds -1} =(Q'_j(\bY - \bb))_{j=0}^{\tilds -1},$$ and add this to a new list of tuples of  polynomials that we call $\calL''$.

\item Let
$$Z = \left\{ z \in \Sigma_{m,\tilds} \mid 
(\sum_{\wt(\bi) = j} z^{(\bi)} \bY^\bi )_{j=0}^{\tilds -1}
\in \calL'' \right\}.$$

\item Return $Z$.
\end{enumerate}
}

We now prove Lemma~\ref{lem:recovercandidates}. 
\begin{proof}[Proof of Lemma~\ref{lem:recovercandidates}]

Suppose $Q(X_1, \ldots, X_m) \in \F_q[X_1, \ldots, X_m]$ is a polynomial 
of degree at most $d$ such that:
$$ \Pr_{\bx \in \F_q^m} [ Q^{(<s)}(\bx) \in S(\bx) ] > 1-\alpha.$$

For each line $\lambda$ in $\F_q^m$, let $\calL_{\lambda}$ be the result of univariate list-recovering $S \circ \lambda$ with 
error-parameter $\alpha'$.
Let $\bb \in \F_q^m$ be the random choices of $\RecCand$. For $\bx, \bb \in \F_q^m$, and for $\bu \in U^m$, recall that $\lambda_\bu (T)$ denotes the line $\lambda_\bu(T) = \bx + T \bb_\bu$, where $\bb_\bu = \bb + \bu $. 
Then for each $\bu \in U^m$, let $B_\bu$ be the event that $Q\circ \lambda_\bu(T)$ is not in $\calL_{\lambda_\bu}$.

\begin{claim}\label{clm:chebyshev-randomline}
For each fixed $\bu \in U^m$,
$$ \Pr_{\bx, \bb \in \F_q^m}[B_{\bu}] \leq \frac{4}{(\alpha' - \alpha) q}.$$
\end{claim}
\begin{proof}
Note that when $\bx, \bb$ are uniformly random elements of $\F_q^m$, then
$\lambda_\bu(T)$ is a uniformly random line. The event that $Q\circ \lambda_\bu(T)$ is not in $\calL_{\lambda_\bu}$ is a subset of the event that $Q^{(<s)}(\bx) \not \in S(\bx)$ for more than $\alpha'$ fraction of points $\bx$ on the line $\lambda_\bu(T)$.  
The claim then follows from a standard application of Chebyshev's inequality, using the fact that the points on a uniformly random line are pairwise independent. 

More precisely,
\begin{align*}
\Pr_{\bx,\bb \in \F_q^m}[B_{\bu}] &\leq \Pr_{\bx,\bb \in \F_q^m} \left[ \sum_{t \in \F_q} \mathbf{1}\left\{ Q^{(<s)}(\lambda_\bu(t) \not\in S(\lambda_{\bu}(t)) \right\} > \alpha' q \right] \\
&=: \Pr_{Y} \left[ \sum_{t \in \F_q} Y_t > \alpha' q \right] \\
&= \Pr_{Y} \left[ \sum_{t \in \F_q} (Y_t - \mathbb{E} Y_t) > (\alpha' - \alpha) q \right],
\end{align*}
where the $Y_t$ are pairwise independent $\{0,1\}$-valued random variables with $\mathbb{E} Y_t \leq \alpha$.  Then by Chebyshev's inequality, this last quantity is at most
\[ \frac{ \sum_{t \in \F_q} \mathbb{E} ( Y_t - \mathbb{E} Y_t)^2 }{ (\alpha' - \alpha)^2 q^2 } \leq \frac{ 4 }{(\alpha' - \alpha)^2 q }. \]
\end{proof}

\begin{claim}
$$ \Pr_{\bx, \bb_1, \ldots, \bb_m}\left[ \sum_{\bu \in U^m} \mathbf{1}_{B_{\bu}} > \frac{|U|^m}{3} \right] < \frac{12}{(\alpha'-\alpha) q}.$$
\end{claim}
\begin{proof}
The proof is immediate from the previous claim and Markov's inequality.
\end{proof}

Thus we conclude that with probability at least $1- \frac{12}{(\alpha'-\alpha) q}$, when $\bx$ is a uniformly random element of $\F_q^m$, for at least $1/3$ of the $\bu \in U^m$, we have that $$Q\circ \lambda_\bu(T) \in \calL_{\lambda_\bu}.$$ We assume that this happens, and let $G \subseteq U^m$ be this set of $\bu$.

Recall that $\calL_{\lambda_\bu}$ is a list of size $L$. 
Consider the function  $$f: U^m \to {\F_q^{\tilds} \choose L},$$
where for each $\bu \in U^m$, 
$$f(\bu) =  \{ P^{(< \tilds)}(0)  \mid P(T) \in \calL_{\lambda_\bu} \}.$$

Fix any $\bu \in G$. Then since $$Q\circ \lambda_\bu(T) \in \calL_{\lambda_\bu},$$
it holds that $$(Q\circ \lambda_\bu)^{(< \tilds)}(0) \in f(\bu).$$

Now observe that by Claim~\ref{claim:hasse}, we have 
$$ H_{\bx}(\bb_\bu) =  (Q\circ \lambda_\bu)^{(< \tilds)}(0),$$
and thus
$$ H_{\bx}(\bb + \bu) = H_{\bx}(\bb_\bu) \in f(\bu).$$
Since this happens for each $\bu \in G$, we have that this
happens for at least $1/3$ fraction of $\bu \in U^m$.

Now by our assumption that $|U| \geq 100 m^2 L \tilds$, Lemma~\ref{lem:RM-vectors} implies\footnote{WARNING: we invoke the algorithm and statement of Lemma~\ref{lem:RM-vectors} with input list size equal to $L$; i\.e\., we take $\ell = L$ (and $K = 100m^2$) in the statement when we invoke it.} that the algorithm $\mathsf{VectorRMListRecover}$ on input $f$ and error parameter $2/3$
will include $H_{\bx}(\bb + \bY)$ in $\mathcal{L}'$ (here we also use the fact that $H_{\bx, j}$ is an $m$-variate polynomial of degree at most $j$).

In this event, $\calL''$ will contain $$H_{\bx}(\bY)$$
and then it follows that in Step 8 of $\RecCand$, the list $Z$ will contain $Q^{(< \tilds)}(\bx)$.

\end{proof}

\subsection{The Oracle Machine $M$}\label{ssec:oracle}
Our final local list recovery algorithm will output a short list of oracle machines, each of which is defined by a piece of advice.  In this case, the advice will be a point $\ba \in \F_q^m$, and $z \in \Sigma_{m,s^*}$, which is meant to be a guess for $Q^{(<s^*)}(\ba)$.  Given this advice, the oracle machine works as follows: on input $\bx$, with corresponding input list $Y = S(\bx)$, it will run the univariate list-recovery algorithm on the line $\lambda(T) = \bx + T(\ba - \bx)$ through $\bx$ and $\ba$ to obtain a list $\mathcal{L}$ of univariate polynomials $P(T)$.
We will show that with high probability (assuming the advice is good), there will be a unique polynomial $P(T)$ 
in $\mathcal{L}$ so that both $P^{(<s^*)}(1)$ is consistent with $z$, and $P^{(<s^*)}(0)$ is consistent with some element of $Y$. Then the oracle machine will output the symbol in $Y$ that $P^{(<s^*)}(0)$ agrees with. 

The key later will be that the advice $z$ will not vary over all possibilities in $\Sigma_{m,s^*}$; this would result in too long a list.  Rather, we will use $\RecCand$ in order to generate this advice.

Formally, we will prove the following lemma about our oracle machine, which we define below.
\begin{lemma}\label{lem:oracle}
Let $q$ be a prime, let $\delta > 0$ and let $s,d,m$ be nonnegative integers such that $d = (1-\delta) sq$. 
Let $\alpha, \alpha'$ be parameters such that $0< \alpha< \alpha'$. 

Let $L = L(d,q,s, \ell, \alpha') $ be the list size for list recovering univariate multiplicity codes $\MULT^{(1)}_{q,s}(d)$ of degree $d$ and multiplicity parameter $s$ with input list size $\ell$ and error tolerance $\alpha'$, and $T$ be the corresponding running time for list recovering univariate multiplicity codes of the same parameters.

Let $S: \F_q^m \to {\Sigma_{m,s} \choose \ell}$. Let $s^* > 0$ be a parameter. 
Suppose that $q > C \cdot s^* \cdot L \cdot m^2$, for some absolute constant $C$, and suppose that 
 $Q(X_1, \ldots, X_m) \in \F_q[X_1, \ldots, X_m]$ is a polynomial 
of degree at most $d$ such that:
$$ \Pr_{\bx \in \F_q^m} [ Q^{(<s)}(\bx) \in S(\bx) ] > 1-\alpha.$$
Let 
$$\epsilon = \alpha + \frac{\ell s}{q} + \frac{4}{(\alpha'-\alpha)q} + \frac{sL}{s^*}.$$

There is an algorithm $M^S[\ba, z](\bx)$ which on input $x \in  \F_q^m$, given as advice a point $\ba \in \F_q^m$, and 
$z \in \Sigma_{m,s^*}$, and given oracle access to $S$, makes at most $q$ queries to $S$, runs in time at most 
$\poly\left(T, q, (s \cdot s^*)^m\right)$, 
and outputs an element of $\Sigma_{m,s} \cup \{\bot\}$
such that if $\bx, \ba \in \F_q^m$ are chosen uniformly at random,
then:
$$\Pr_{\ba, \bx \in \F_q^m}\left[M^S[\ba,  Q^{(< s^*)}(\ba)](\bx) = Q^{(< s)}(\bx)\right] \geq 1-\epsilon. $$

In particular, if we view $M^S[\ba,  Q^{(< s^*)}(\ba)]$ as a function from $\F_q^m \to \Sigma_{m,s} \cup \{\bot\}$, then with probability at least $3/4$ over the choice of $\ba \in \F_q^m$, 
$$\dist\left( M^S[\ba,  Q^{(< s^*)}(\ba)], Q^{(<s)} \right) \leq 4 \epsilon.$$

\end{lemma}

We will first decribe the algorithm and then show that it satisfies the required properties. 

\putinbox{
\noindent{\bf Oracle machine $M$.}
\begin{itemize}
\item Oracle access to $S: \F_q^m \to {\Sigma_{m,s} \choose \ell}$.
\item {\bf INPUT:} $\bx \in \F_q^m$.
\item {\bf ADVICE:} Point $\ba \in \F_q^m$, and 
$z \in \Sigma_{m,s^*}$.
\end{itemize}
\begin{enumerate}
\item Let $Y = S(\bx)$.
\item Set $\bb_* = \ba - \bx$.
\item Let $\lambda_{\bb_*}$ be the line $\lambda_{\bb_*}(T) = \bx + T \bb_*$.
\item Consider the restriction $S_{\bb_*} : \F_q \to  {\Sigma_{1,s} \choose \ell}$ of $S$ to the line $\lambda_{\bb_*}$,
 and list recover this with error-tolerance $\alpha'$ for degree $d$ polynomials,
and obtain the list $\calL_{\lambda_{\bb_*}} \subseteq \F_q[T]$.
\item If there exists exactly one $P(T) \in \calL_{\lambda_{\bb_*}}$ such that
$P^{(<s^*)}(1) = z|_{\bb_*}$, then set $P_{\bb_*}(T)$ to equal that $P(T)$,
otherwise output $\bot$ and exit.
\item If there exists exactly one $y \in Y$ for which
$y|_{\bb_*} = P_{\bb_*}^{(<s)}(0)$, then output that $y$.
\item Otherwise output $\bot$.
\end{enumerate}
}

We will now analyze the above algorithm and show that is satisfies the required properties. 
\begin{proof}[Proof of Lemma~\ref{lem:oracle}]
By the description of the oracle machine, it is clear that it makes at most $q$ queries. 
Moreover its running time is at most $\poly\left(T, q, (s \cdot s^*)^m\right)$, since in addition to running the univariate list recovery algorithm, it needs to do some field calculations such as computing the restriction of $S$ to a line, as well as computing for each $P(T) \in \calL_*$, 
$P^{(<s^*)}(1)$ and $z|_{\bb_*}$ and comparing the two.  This takes time at most $\poly\left(q, (s \cdot s^*)^m\right)$. 

It remains to show that when $\ba$ and $\bx$ are chosen uniformly at random from $\F_q^m$,  then
$$\Pr_{\ba, \bx \in \F_q^m}\left[M^S[\ba,  Q^{(< s^*)}(\ba)](\bx) = Q^{(< s)}(\bx)\right] \geq 1-\epsilon. $$

\begin{claim}\label{clm:easy}
Let $y_0 = Q^{(<s)}(\bx)$. With probability at least $1-\alpha$ over the random choice of $x \in \F_q^m$, 
$y_0 \in  Y$. 
\end{claim}
\begin{proof}
Recalling that $Y = S(\bx)$, 
the proof is immediate since it is given to us that $$ \Pr_{\bx \in \F_q^m} [ Q^{(<s)}(\bx) \in S(\bx) ] > 1-\alpha.$$
\end{proof}

\begin{claim}\label{clm:distinct-restrictions}
Let $y_0 = Q^{(<s)}(\bx)$. For any $y \in Y$ with $y \neq y_0$,
with probability at least $1 - \frac{s}{q}$ over the random choice of $\ba \in \F_q^m$, we have that
$$y|_{\bb_*} \neq y_0 |_{\bb_*}.$$
\end{claim}
\begin{proof}
Recall that by definition, for an element $z \in \Sigma_{m,s}$, and a direction
$\bb \in \F_q^m$, $z|_{\bb}$ is to equal $h \in \Sigma_{1,s}$, where:
$$h^{(j)} = \sum_{\wt(\bj) = j} z^{(\bj)} \bb^{\bj}$$
for each $j$ such that $0 \leq j < s$.
Note that $h^{(j)}$ can be viewed as a polynomial of degree at most $j$ evaluated at $\bb$, where the coefficients of the polynomial depend only on $z$. 

Since $y \neq y_0$, the corresponding tuples of polynomials (each of degree at most $s$) will differ in at least one coordinate. Observe also that for any fixed choice of $\bx$, the randomness of $\ba$ implies that $\bb_*$ is a uniformly random element of $\F_q^m$. This in the coordinate where the tuples of polynomials differ, the evaluations at $\bb_*$ will be distinct with probability at least $1 - \frac{s}{q}$ by the Schwartz-Zippel Lemma. 

Thus with probability at least $1 - \frac{s}{q}$ over the random choice of $\ba \in \F_q^m$, we have that
$$y|_{\bb_*} \neq y_0 |_{\bb_*}.$$

\end{proof}

\begin{claim}\label{clm:unique}
Let $y_0 = Q^{(<s)}(\bx)$. For any $x \in \F_q^m$ such that $y_0 \in  Y$,  with probability at least $1 - \frac{\ell s}{q}$ over the random choice of $\ba \in \F_q^m$, 
$y_0|_{\bb_*}$ is unique element $y$ of $Y$ for which $y|_{\bb*} = Q \circ \lambda_{\bb_*}^{(<s)}(0)$.
\end{claim}

\begin{proof}
Clearly, by definition, $y_0|_{\bb*} = Q \circ \lambda_{\bb_*}^{(<s)}(0)$. 
Also, taking a union bound over all $\ell$ elements of $Y$, by Claim~\ref{clm:distinct-restrictions}, $y_0|_{\bb*} \neq y|_{\bb_*} $ for all other $y \in Y$ with probability at least $1 - \frac{\ell s}{q}$. 
\end{proof}

Claim~\ref{clm:easy} and Claim~\ref{clm:unique} together imply that with probability at least $1- \left( \alpha+ \frac{\ell s}{q}\right)$ over the random choice of $\ba$ and $\bx \in \F_q^m$, 
$Q^{(<s)}(\bx)|_{\bb_*}$ is the unique element $y$ of $Y$ for which $y|_{\bb*} = Q \circ \lambda_{\bb_*}^{(<s)}(0)$. 

We will now show that with probability at least $1- \left(  \frac{4}{(\alpha'-\alpha)q} + \frac{sL}{s^*} \right)$ over the random
choice of $\ba$ and $\bx \in \F_q^m$, $P_{\bb_*}(T) = Q\circ{\lambda_{\bb_*}}(T)$. 
Once we will have this, then it will immediately follow that the algorithm will output $Q^{(<s)}(\bx)$ with probability at least $1- \left(  \frac{4}{(\alpha'-\alpha)q} + \frac{sL}{s^*} + \frac{4}{(\alpha'-\alpha)q} + \frac{sL}{s^*}\right)$ over the random
choice of $\ba$ and $\bx \in \F_q^m$

For each line $\lambda$ in $\F_q^m$, let $\calL_{\lambda}$ be the result of list-recovering $S \circ \lambda$ with 
error-parameter $\alpha'$.
For points $\bx$ and $\ba$ picked uniformly at random from $\F_q^m$, let $\bb_* = \ba - \bx$, 
and let $\lambda_{\bb_*}$ be the line $\lambda_{\bb_*}(T) = \bx + T \bb_*$.

Let $B_{\lambda_{\bb_*}}$ denote the event that
$\calL_{\lambda_{\bb_*}}$ does not contain $Q\circ \lambda_{\bb_*}(T)$. Let $C_{\lambda_{\bb_*},\ba}$
denote the event that there exist $P(T) \in \calL_{\lambda_{\bb_*}}$ with
$P(T) \neq Q \circ \lambda_{\bb_*}(T)$, but  $P^{(<s^*)}(0) = (Q\circ \lambda_{\bb_*}) ^{(<s^*)}(0)$.
Thus $B_{\lambda_{\bb_*}}$ is the event  that there are too many errors on $\lambda_{\bb_*}$.
$C_{\lambda_{\bb_*},\ba}$ is the event that $\ba$ is not a disambiguating point.

\begin{claim}\label{clm:chebyshev}
$$ \Pr[B_{\lambda_{\bb_*}}] = \frac{4}{(\alpha' - \alpha) q}.$$
\end{claim}
\begin{proof}
The proof is identical to that of Claim~\ref{clm:chebyshev-randomline},  and it follows from a standard application of Chebyshev's inequality, using the fact that the points on a uniformly random line are pairwise independent. 
\end{proof}

\begin{claim}\label{clm:disambiguating}
$$ \Pr[C_{\lambda_{\bb_*},\ba}] < \frac{sL}{s^*}.$$
\end{claim}

\begin{proof}
Because of the way $\bx$, $\ba$ and the line $\lambda_{\bb_*}$ are sampled, equivalently one could let 
$\bx$ be picked uniformly at random from $\F_q^m$, $\lambda_{\bb_*}$ be a uniformly random line through $\bx$ and $\ba$ be a uniformly random point on $\lambda_{\bb_*}$. 

Now fix any polynomial $P(T) \in \calL_{\lambda_{\bb_*}}$ with
$P(T) \neq Q\circ \lambda_{\bb_*}(T)$.
We want to bound the probability that $P^{(<s^*)}(\alpha) = (Q\circ \lambda_{\bb_*})^{(<s^*)}(\alpha)$
where $\alpha$ is picked uniformly at random. But $P$ and $Q \circ \lambda_{\bb_*}$ are fixed distinct
polynomials of degree at most $sq$. Thus the probability that they agree with multiplicity $s^*$ 
on a random point of $\F_q$ is at most $\frac{sq}{s^* q} = \frac{s}{s^*}$.

The result follows from a union bound over all $P(T) \in \mathcal{L}_{\lambda_{\bb_*}}$.
\end{proof}

Claim~\ref{clm:chebyshev} and Claim~\ref{clm:disambiguating} together imply that with probability at least $1- \left(  \frac{4}{(\alpha'-\alpha)q} + \frac{sL}{s^*} \right)$ over the random
choice of $\ba$ and $\bx \in \F_q^m$, neither $B_{\lambda_{\bb_*}}$ nor $C_{\lambda_{\bb_*},\ba}$ occurs, and hence $P_{\bb_*}(T) = Q\circ{\lambda_{\bb_*}}(T)$.

Thus the result follows.

\end{proof}

\subsection{Main local list-recovery algorithm}\label{ssec:main_lld}
Together, Lemmas \ref{lem:recovercandidates} and \ref{lem:oracle} inspire a local-list-recovery algorithm for multivariate multiplicity codes.  
The idea is that \textsf{RecoverCandidates} will first obtain a list of possibilities, $Z$, for $Q^{(<s^*)}(\mathbf{a})$.  Then for each possibility $z \in Z$, we will create an oracle machine as in Lemma~\ref{lem:oracle} which guesses $Q^{(<s^*)}(\mathbf{a}) = z$.
Unfortunately, this will still have some amount of error; that is, there will be some small fraction of $\mathbf{x} \in \F_q^m$ so that the approach above will not be correct on $\mathbf{x}$.  To get around this, we will wrap the whole thing in the local (unique) correction algorithm for multiplicity codes from~\cite{KSY14}.

\begin{theorem}[\cite{KSY14}, Theorem 3.6]\label{thm:ksy}
Let $C$ be multiplicity code $\MULT_{q,s}^{(m)}(d)$.  Let $\delta = 1 -\frac{d}{sq}$.  Suppose that $q \geq \max\{10m, \frac{d + 6s}{s}, 12(s+1)\}$.  Then $C$ is locally correctable from $\frac{\delta}{10}$-fraction of errors with $( O(s)^m \cdot q)$ queries.
Moreover, the local corrector $\mathsf{SelfCorrect}^{c}(\mathbf{x})$, with query access to a codeword $c \in C$ running on a position $\mathbf{x} \in \F_q^m$, can be\footnote{This claim about the running time in~\cite{KSY14} was
only proved for fields of small characteristic.
There, in the discussion about ``Solving the Noisy System" in Section 4.3,
it was shown that the running time can be made $\poly(O(s)^m \cdot q)$ 
provided one could efficiently decode Reed-Muller codes over certain product sets in $\F_q$, and remarked that this was known over fields of small characteristic. Recently \cite{KK-rm-product-set} showed that this Reed-Muller decoding problem could be solved over all fields. This justifies the running time claim over all fields.} made to run in time $O(s)^m \cdot q^{O(1)}$.
\end{theorem}

With the self-correction algorithm for multiplicity codes in hand, we define our local-list-recovery algorithm as follows.

\medskip
\putinbox{
\noindent
\textbf{Algorithm} $\MultimultMain.$
\begin{itemize}
\item Oracle access to $S: \F_q^m \to { \Sigma_{m,s} \choose \ell }$.
\end{itemize}
\begin{enumerate}
	\item Pick $\mathbf{a} \in \F_q^m$ uniformly at random.
	\item Set $s^* = \frac{ 160 \cdot L \cdot s }{\delta}$.
	\item Let $Z$ be the output of $\RecCand^S(\mathbf{a}, s^*)$.
	\item for $z \in Z$, define $\mathcal{A}_z$ by:
	\begin{itemize}
		\item \textbf{INPUT:} $\mathbf{x} \in \F_q^m$
	\end{itemize}
	\begin{enumerate}
		\item Let $M$ denote the oracle machine $M^S[\mathbf{a},z]$
		\item Return $\mathsf{SelfCorrect}^{M}(\mathbf{x})$ 
	\end{enumerate}
	\item Return $\mathcal{L} = \{ \mathcal{A}_z \,:\, z \in Z \}$.
\end{enumerate}
}

The following lemma shows that this algorithm works, assuming a list-recovery algorithm for univariate multiplicity codes.  In the proof of Theorem~\ref{thm:main-multimult}, we will instantiate this with the list-recovery algorithm given in Section~\ref{sec:Unimult}.

\begin{lemma}\label{lem:multimult-params}
There is some constant $C > 0$ so that the following holds.
Let $q$ be a prime, let $\delta > 0$ and let $s,d,m$ be nonnegative integers such that $d = (1-\delta) sq$. 
Let $\alpha, \alpha'$ be parameters such that $0< \alpha< \alpha' < 1$. 

Let $L = L(d,q,s, \ell, \alpha') $ be the list size for list recovering univariate multiplicity codes $\MULT^{(1)}_{q,s}(d)$ of degree $d$ and multiplicity parameter $s$ with input list size $\ell$ and error tolerance $\alpha'$,  and $T$ be the corresponding running time for list recovering univariate multiplicity codes of the same parameters. 

Let $S: \F_q^m \to { \Sigma_{m,s} \choose \ell }$.  Suppose that 
\[ s^* \geq \frac{ 160 \cdot L \cdot s }{\delta} \]
and that
\[ q \geq \max\left\{ \frac{ 160 \ell s}{\delta}, \frac{ 640 }{ (\alpha' - \alpha) \cdot \delta }, C s^* L m^2, \frac{20 \cdot C m }{\alpha' - \alpha} , 10m, \frac{d+6}{s}, 12(s+1) \right\} \]
and that
\[ \alpha \leq \frac{\delta}{160}. \]
Then for all $Q(X_1,\ldots,X_m) \in \F_q^m[X_1,\ldots, X_m]$ with degree at most $d$ and so that
\[ \Prn_{\mathbf{x} \in \F_q^m} [ Q^{(<s)} (\mathbf{x}) \in S(\mathbf{x}) ] > 1 - \alpha, \]
with probability at least $2/3$ over the algorithm $\MultimultMain$, the following holds.  For all $\mathbf{x} \in \F_q^m$, there exists an oracle machine $\mathcal{A}_z \in \mathcal{L}$ so that
\[ \Prn\left[ \mathcal{A}_z(\mathbf{x}) = Q(\mathbf{x}) \right] \geq 2/3. \]

Moreover, the output list $\mathcal{L}$ has size $|\mathcal{L}| = O(L)$; and $\MultimultMain$ makes $\poly(q,(Ls^* m)^m)$ queries to $S$, and each $\mathcal{A}_z$ makes $O(s)^m \cdot q^2$
queries to $S$.  Finally, the algorithm $\MultimultMain$ runs in time $\poly( T, q,  (Ls^* m)^m)$ and each $\mathcal{A}_z$ runs in time 
$O(s)^m \cdot \poly(q, T, (s\cdot s^*)^m )$.
\end{lemma}
\begin{proof}
Fix a polynomial $Q \in \F_q[X_1,\ldots,X_m]$ of degree at most $d$, so that $\dist(Q,S) \leq \alpha$.   We first establish the correctness of the algorithm $\MultimultMain$ given above. 

By Lemma~\ref{lem:recovercandidates}, with probability at least $1 - \frac{ 20m }{(\alpha' - \alpha)q} \geq 1 - \frac{1}{C}$ over the randomness of both $\RecCand$ and $\mathbf{a}$, $\RecCand^S(\mathbf{a}, s^*)$ returns a list $Z$ of size at most $O(L)$ so that $Q^{(<s^*)}(\mathbf{a}) \in Z$. 
Let $G_1$ be the set of $\mathbf{a} \in \F_q^m$ so that $\Prn_{\RecCand}[Q^{(<s^*)}(\mathbf{x}) \in Z] \geq 1 - \frac{1}{\sqrt{C}}$.  By Markov's inequality along with the conclusion of Lemma~\ref{lem:recovercandidates} above, $G_1$ has density at least $1 - 1/\sqrt{C}$.  Now let $G_2$ be the set of $\mathbf{a} \in \F_q^m$ so that $\dist( M^S[\mathbf{a}, Q^{(<s^*)}(\mathbf{a}], Q^{(<s)} ) \leq 4\eps$.  By Lemma~\ref{lem:oracle}, $G_2$ has density at least $3/4$.  Thus by the union bound, with probability at least $3/4 - 1/\sqrt{C}$ over the choice of $\mathbf{a}$, both events hold, and so with probability at least $3/4 - 2/\sqrt{C}$ over the choice of $\mathbf{a}$ and the randomness of $\RecCand$, there is some $z \in Z$ so that
\begin{equation}\label{eq:favorable} \dist( M^S[ \mathbf{a}, z ], Q^{(<s)} ) \leq 4\eps \end{equation}
for any 
\[ \eps \leq \alpha + \frac{\ell s}{q} + \frac{ 4}{(\alpha' - \alpha)q} + \frac{sL}{s^*}. \]
By choosing $C \geq 24^2$ (as well as large enough so that Lemma~\ref{lem:recovercandidates} and \ref{lem:oracle} hold), we can ensure that \eqref{eq:favorable} occurs with probability at least 2/3.
Suppose that this happens, and \eqref{eq:favorable} does occur.  Observe that our parameter choices above are made precisely so that
\[ \frac{ \delta}{40} \geq \alpha + \frac{ \ell s }{q} + \frac{ 4}{ (\alpha' - \alpha)q } + \frac{ sL}{s^* }. \]
Thus, we may take $\eps = \delta/40$ in the above, and conclude that in the favorable case of \eqref{eq:favorable}, we have
\[ \dist( M^S[ \mathbf{a},z], Q^{(<s)} ) \leq \frac{\delta}{10}. \]
We may then apply Theorem~\ref{thm:ksy} to the oracle machine $M = M^S(\mathbf{a},z)$ in the algorithm above, and conclude that $\mathsf{SelfCorrect}^M(\mathbf{x})$ is a local-self-corrector for $\MULT_{q,s}^{(m)}(d)$.  In particular, for all $\mathbf{x} \in \F_q^m$, with probability at least $2/3$, $\mathcal{A}_z(\mathbf{x}) = Q^{(<s)}(\mathbf{x})$, as desired.

Now that we have established that the algorithm is correct, we quickly work out the list size, query complexity, and runtime.  The list size is clearly $O(L)$, because this is the list size returned by $\RecCand$.  For the query complexity, the algorithm $\MultimultMain$ has the same query complexity as $\RecCand$, while each $\mathcal{A}_z$ has query complexity which is the product of the query complexities of the oracle machines $M^S[\mathbf{a}, z]$ (which is $q$) and $\mathsf{SelfCorrect}$ (which is $O(s)^m \cdot q$), and together these give the reported values.  The runtime calculation is similar.  
\end{proof}

Finally, we may choose parameters and use Theorem~\ref{thm:full-unimult-list-rec} to prove Theorems~\ref{thm:main-multimult} and \ref{thm:main-multimult2}.
\begin{proof}[Proof of Theorem~\ref{thm:main-multimult}]
The proof proceeds by setting parameters in Lemma~\ref{lem:multimult-params}.  
We will let $C = \MULT_{q,s}^{(m)}(d)$, where $m$ is the parameter from the theorem statement.
We will choose $q,s,d$ below.
Let $\delta = \eps/(2m)$; we will verify below that $\delta$ is a bound on the distance of $C$.

Choose $s = \frac{ 16 \cdot \ell }{\delta^2}$, and $\alpha \leq \frac{ \delta^2}{160 \cdot \ell}$ as in the theorem statement.  We will choose $\alpha' = 2\alpha$, so $\alpha' - \alpha = \alpha$.
We note that these choices ensure that $\ell \leq \frac{ \delta^2 s }{16}$ and that $\alpha' < 1/2s$, both of which are required for Theorem~\ref{thm:full-unimult-list-rec} to hold (when called with $\alpha'$ as the error parameter), as well as $\alpha < \frac{\delta}{160}$, as required by Lemma~\ref{lem:multimult-params}.

Now, with these choices Theorem~\ref{thm:full-unimult-list-rec} says that $\MULT_{q,s}^{(1)}(d)$ is $(\alpha', \ell, L')$-list-recoverable with 
\[L' = \ell^{O(s\log(s))} \cdot s^{O(1)} = \ell^{O(\ell \log(\ell/\delta)/\delta^2)} \cdot \left( \frac{1}{\delta} \right)^{O(1)} = \ell^{ O( \ell \cdot m^2 \log(m\ell/\eps) / \eps^2) }\cdot\left( \frac{m}{\eps} \right)^{O(1)}, \]
using our choice of $s$ and $\delta$.  Since the list size in Lemma~\ref{lem:multimult-params} grows by at most a constant factor, this establishes our choice of $L$ in the theorem statement.

We will next choose $q$.  We need $q$ to be large enough so that Lemma~\ref{lem:multimult-params} holds.  
It can be checked that of all of the requirements on $q$ given in Lemma~\ref{lem:multimult-params}, the binding one is that $q = \Omega(s^* L m^2)$, where we chose $s^* = \Theta( Ls/\delta )$.  
We shall choose $q$ safely larger than this, choosing a prime $q$ so that
\[ q := \Theta\left( \left( \frac{Lms^*}{\delta} \right)^m \right) = \left( L m \ell / \eps \right)^{O(m)}. \]
The reason for this choice is that this is the largest we may take $q$ so that the query complexity expression 
\[ \poly( q \cdot (Ls^*m)^m ) \]
from Lemma~\ref{lem:multimult-params} does not substantially grow.

Now that we have chosen $s$ and $q$, we will finally choose 
\[ d = \left( 1 - \frac{\eps}{2m} \right) sq \]
so that the distance of $C$ is 
\[ \delta = 1 - \frac{d}{sq} = \frac{\eps}{2m} \]
as claimed.

With this choice the query complexity given in Lemma~\ref{lem:multimult-params} is
\[ t = O(s)^m \cdot q^2 + q^{O(1)} \cdot \left( \frac{ Ls^*m}{\delta} \right)^{O(m)}, \]
which with our choices of $s,\delta$ and $q$ is
\[ t = \left( \frac{ L m \ell }{\eps} \right)^{O(m)} \]
as claimed.

We now verify the rate.  As per Claim~\ref{claim:multparams}, the rate of $C$ is at least
\begin{align*}
R &\geq \left( 1 - \frac{m^2}{s} \right) \left( 1 - \delta \right)^m \\
&= \left( 1 - \frac{\eps^2}{16 \ell} \right) \left( 1 - \frac{ \eps }{2m } \right)^m \\
&\geq \left( 1 - \frac{ \eps^2 }{ 16 \ell } \right) \left( 1 - \frac{ 2\eps }{3 } \right) \\
&\geq 1 - \eps, 
\end{align*}
where the last two lines hold for sufficiently small $\eps$.  
We first note that our choice of $\alpha \leq \frac{ \delta^2}{160 \ell}$ satisfies $\alpha \leq \frac{\delta}{160}$, which was required in Lemma~\ref{lem:multimult-params}.  

Finally, we note that the length of the code $C$ is 
\begin{align*}
N &= q^m 
= \left( \frac{ Lm \ell }{\eps} \right)^{O(m^2)} 
\end{align*}
and that the alphabet size is similarly
\begin{align*}
|\Sigma| &= q^{s^m}  = \left( \frac{ Lm \ell }{\eps} \right)^{O(s^m)},
\end{align*}
which results in the expression given in the theorem statement.
Finally, the running time for the list-recovery algorithm guaranteed by Lemma~\ref{lem:multimult-params} is dominated by the $(L s^* m)^m$ term, which is $\poly(t)$.
\end{proof}
\begin{proof}[Proof of Theorem~\ref{thm:main-multimult2}]
Again, we set parameters in Lemma~\ref{lem:multimult-params}.  Let $\delta, \eps, \tau, \ell$ be as in the statement of Theorem~\ref{thm:main-multimult2}.  We will choose $C = \MULT_{q,s}^{(m)}(d)$, and we set parameters below.
First, we choose 
\[ s = \max \left\{ \frac{ 16 \ell }{\delta^2}, \frac{ c^2 }{\tau^2 \eps } \right\}, \]
where $c$ is some universal constant that will be chosen below.  
We will choose
\[ m = \frac{c}{\tau}, \]
and $d$ so that $d = (1 - \delta)sq$, ensuring that the relative distance of the code is at least $\delta$.
Now we choose 
\[ \alpha' \leq \min\left\{ \frac{ \delta^2}{160 \ell}, \frac{ \tau^2 \eps }{ 2c^2 }\right\}, \]
and $\alpha = \alpha'/2$.
which ensures that $\alpha \leq \delta/160$ (as is required for Lemma~\ref{lem:multimult-params}) and that $\alpha' \leq 1/2s$, which is required for Theorem~\ref{thm:full-unimult-list-rec}.  We also have $(\alpha' - \alpha) = \Omega_{\delta, \eps, \ell, \tau}(1)$.

Notice that all of the requirements on the size of $q$ in Lemma~\ref{lem:multimult-params} simply require $q = \Omega_{\ell, \eps, \delta, \tau}(1)$, so we choose any prime $q$ sufficiently large, and let $N = q^m$ is be the length of the multiplicity code.

By Claim~\ref{claim:multparams}, $C$ has rate at least
\[ \left( 1 - \frac{m^2}{s} \right) \left( \frac{ d }{sq}\right)^m \geq (1 - \eps) (1 - \delta)^m = (1 - \eps)(1 - \delta)^{c/\tau}, \]
which is what was claimed.

Then by Theorem~\ref{thm:full-unimult-list-rec}, the univariate multiplicity code $\MULT_{q,s}^{(1)}(d)$ is $(\alpha, \ell, L)$-list-recoverable with 
\[ L = \ell^{O(s \log(s))}s^{O(1)} = O_{\ell, \delta, \eps, \tau}(1) \]
in time $\poly_{\ell, \delta, \eps, \ell}(q)$.

Now we choose $s^* = 160 L s / \delta = O_{\ell, \delta, \eps, \tau}(1)$, and Lemma~\ref{lem:multimult-params} concludes that $C$ is $(t, \alpha, \ell, L')$-list-recoverable for $L' = O(L) = O_{\ell, \delta, \eps, \tau}(1)$ and for 
\[ t = q^c \cdot (Ls^* m)^O(m) = O_{\ell, \delta, \eps, \tau}( q^c ) \]
for some constant $c$.  (This defines the constant $c$).
Now, since $N = q^m = q^{c/\tau}$, we have $t = O_{\ell, \delta, \eps, \tau}(N^{\tau})$, as desired.  Finally, Lemma~\ref{lem:multimult-params} further implies that the running time of the local list-recovery algorithm is $\poly(t)$, where the exponent in the polynomial does not depend on $\ell, \delta, \eps,$ or $\tau$.

\end{proof}

\section{Capacity-achieving codes over constant-sized alphabets}\label{sec:Smallalpha}
Theorems~\ref{thm:main-multimult} and \ref{thm:main-multimult2} show that high-rate multivariate multiplicity codes are efficiently locally list-recoverable.  However, the alphabet sizes for both of these constructions are large, and they only tolerate a small amount of error.
Fortunately, via standard techniques, we can both boost the error tolerance and improve the alphabet size without substantially impacting the locality or list size.  We will prove the following theorems, based on Theorems~\ref{thm:main-multimult} and \ref{thm:main-multimult2} respectively.

First we give a statement with sub-polynomial query complexity and list size.
\begin{theorem}\label{thm:multimult-smallalpha}
Let $R > 0$.
Let $\eps > 0$ be sufficiently small, and let $m,\ell > 0$ be integers.  Suppose that $\eps,\ell$ are constants, independent of $m$, and that $R \in (\eps, 1 - 2\eps)$.  There is a code $C \subseteq \Sigma^N$ with rate $R$ that is $(t, 1 - R - \eps, \ell, L)$-list-recoverable for
\[ L = m^{O_{\ell, \eps}(m^2)} \]
\[ t = m^{O_{\ell, \eps}(m^3)} \]
\[ N = m^{O_{\ell, \eps}(m^4)} \]
\[ |\Sigma| = O_{\ell, \eps}(1) \]
which can be locally list-recovered in time $\poly(t)$.  Moreover, $C$ has a deterministic encoding algorithm which runs in time $\poly(N)$. 

In particular, solving for $m \approx \left( \frac{\log(N)}{\log\log(N)} \right)^{1/4}$, we have that
\[ L = \exp\left( \sqrt{ \log(N)\log\log(N) } \right) \qquad t = \exp\left( \log^{3/4}(N) \cdot ( \log\log(N))^{1/4} \right). \]
\end{theorem}

Next we give a statement with polynomial query complexity but constant list size.
\begin{theorem}\label{thm:multimult-smallalpha2}
Let $R > 0$.  Let $\eps, \tau > 0$ be sufficiently small, and let $\ell >0$ be an integer.  Suppose that $\eps, \tau, \ell$ are constants, and that $R \in (\eps, 1 - 2\eps)$.  Then for infinitely many $N$, there is a code $C \subseteq \Sigma^N$ of rate $R$ that is $(t, 1 - R - \eps, \ell, L)$-list-recoverable for
\[ L = O_{\ell, \eps, \tau}(1) \]
\[ t = O_{\ell, \eps, \tau}(N^\tau) \]
\[ |\Sigma| = O_{\ell, \eps, \tau}(1), \]
which can be locally list-recovered in time $\poly(t)$.  Moreover, $C$ has a deterministic encoding algorithm that runs in time $\poly(N)$.
\end{theorem}

The proof of Theorems~\ref{thm:multimult-smallalpha} and \ref{thm:multimult-smallalpha2} will follow from an expander-based construction~\cite{AEL95} which has been used in similar settings to reduce alphabet sizes and improve the rate/distance trade-offs (and in particular in~\cite{GR08_folded_RS,GKORS17,HRW17} in the context of list-recovery).  
We state a general transformation below.

\begin{theorem}\label{thm:alphabet-reduction}
Choose $\eps, \gamma, \zeta, R \in (0,1)$.  Let $C_1 \subseteq \Sigma_1^{n_1}$ be a code of rate $1 - \zeta$.  
Suppose that there exists a code $C_0 \subseteq \Sigma_0^{n_0}$ of rate $R$ which is $(1 - R - \eps, \ell, \ell_1)$-list-recoverable in time $T(C_0)$, which can be deterministically constructed in time $T_{construct}(C_0)$.

Then there exists a code $C \subseteq \Sigma^N$ of rate $(1 - \zeta)\cdot R$ over an alphabet of size 
\[ |\Sigma| = |\Sigma_0|^{O(1/(\eps^3\cdot \gamma))} \]
and block length
\[ N = O\inparen{ \frac{ n_1 \log|\Sigma_1| \eps^3 \gamma }{ R\log|\Sigma_0| }}. \]
so that $C$ can be deterministically constructed in time $T_{construct}(C_0) + \poly(|\Sigma_1|, n_1 ,1/\eps, 1/\gamma)$ and so that:

\begin{itemize}
\item If $C_1$ is $(\gamma, \ell_1, L)$-list-recoverable in time $T(C_1)$, then $C$ is $(1 - R - 4\eps, \ell, L)$-list-recoverable in time 
$O\left( n_1 T(C_0) + T(C_1) \right).$
\item If $C_1$ is $(t,\gamma, \ell_1, L)$-list-recoverable in time $T'(C_1)$, then $C$ is $(t', 1 - R - 4\eps, \ell, L)$-locally list-recoverable in time
\[ O\left( \frac{T(C_0) + T'(C_1)}{\eps^{3}\gamma} \right), \]
where
\[ t' = O\left(\frac{ t \log|\Sigma_1| }{ R \log|\Sigma_0|}\right).\]
\end{itemize}
\end{theorem}
The proof of Theorem~\ref{thm:alphabet-reduction} is by now standard, and we include it in Appendix~\ref{app:AEL} for completeness.  The basic idea is to concatenate $C_1$ with $C_0$, and then to scramble up and re-aggregate the symbols of the resulting concatenated code using a bipartite expander graph.

We will use Theorem~\ref{thm:alphabet-reduction} three times: once with $C_0$ as a random linear code and $C_1$ as the concatenation of two folded RS codes; and the next two times with $C_0$ as the code produced by the first application of Theorem~\ref{thm:alphabet-reduction} and with $C_1$ as a multivariate multiplicity code from Theorem~\ref{thm:main-multimult} and Theorem~\ref{thm:main-multimult2}, respectively.

For the list-recoverability of a random linear code, we use a result of~\cite{RW18}. 
\begin{theorem}[Follows from Theorem 6.1 in \cite{RW18}]\label{thm:rlc}
Choose $R \in (0,1)$ be constant, and let $\eps, \ell > 0$.  There is some $q_0 = (1 + \ell)^{O(1/\eps)}$ and 
\[ \ell_1 = \left( \frac{q\ell}{\eps} \right)^{ O(\log(\ell)/\eps^3 )} \]
so that the following holds.
Let $q$ be a prime power, and let $C_0$ be a random linear code over $\mathbb{F}_q$ of rate $R$ and length $n_0$.  Then with high probability, $C_0$ is $(1 - R - \eps, \ell, \ell_1)$-list-recoverable.
\end{theorem}

We will also use the following corollary of Theorem~\ref{thm:frs-list-rec-const-list}.
\begin{corollary}\label{cor:frs-concat}
Let $\eps > 0$ and $\ell \in \mathbb{N}$ be constants.  Then for infinitely many values of $n$, there is a code $C \subseteq \Sigma^n$ of rate $1 - \eps$, which is $(\eps^2/16, \ell, L)$-list-recoverable in time $\poly(n, L)$, for $L = O_{\ell, \eps}(1)$, and which has $|\Sigma| = \poly_{\ell, \eps}(\log(n))$.
\end{corollary}
\begin{proof}
The proof follows by concatenating two folded Reed-Solomon codes.  More precisely, let $C_0 \subseteq \Sigma_0^{n_0}$ be a folded RS code of rate $1 - 2\gamma$ which is $(\gamma, \ell, \ell_1)$-list-recoverable for $\ell_1 = O_{\ell,\gamma}(1)$, which has alphabet size $|\Sigma_0| = \poly_{\ell, \gamma}(n_0)$ and is list-recoverable in time $\poly(n_0,\ell_1)$; this exists by Theorem~\ref{thm:frs-list-rec-const-list}.  Then let $C_1 \subseteq \Sigma_1^{n_1}$ be another code of rate $1 - 2\gamma$ which is $(\gamma, \ell_1, L)$-list-recoverable for $L = O_{\ell_1, \gamma}(1) = O_{\ell, \gamma}(1)$ and has $|\Sigma_1| = \poly_{\ell,\gamma}(n_1)$ and is list-recoverable in time $\poly(n_1, L)$; again this exists by Theorem~\ref{thm:frs-list-rec-const-list}.  Since $|C_0| = n_0^{O_{\ell,\gamma}(n_0)}$ and $|\Sigma_1| = \poly_{\ell, \gamma}(n_1)$, there is a choice of $n_1$ so that $n_0 = O_{\ell, \gamma}(\log(n_1))$ so that  $|C_0| \geq |\Sigma_1|$.

Now consider the code $C$ which is the concatenation of $C_1$ with $C_0$.  
The length of the code is $n = n_1 \cdot n_0$.
The alphabet size of $C$ is $|\Sigma_0| = \poly_{\ell, \gamma}(n_0) = \poly_{\ell, \gamma}(\log(n))$, and the rate is $(1 - 2\gamma)^2 \geq 1 - 4\gamma$.  Finally, it is not hard to see that the composition of two list-recoverable codes is again list-recoverable (see, eg, \cite{HRW17}, Lemma 7.4), and we conclude that $C$ is $(\gamma^2, \ell, L)$-list-recoverable in time $\poly(N, L)$.  Setting $\eps =4\gamma$ completes the proof.
\end{proof}

Next, we instantiate Theorem~\ref{thm:alphabet-reduction} using the codes from Corollary~\ref{cor:frs-concat} as the outer code, and a random linear code as the inner code.
\begin{corollary}\label{cor:FRS_AEL}
Let $R \in (0,1), \ell, \eps > 0$ be constants so that $\eps < R < 1 - 2\eps$, and let $\alpha < 1 - R - \eps$.  Then 
there is a code $C \subset \Sigma^n$ of rate $R$, constructable in time $\poly_{\ell, \eps}(n)$, which is $(\alpha, \ell, L)$-list-recoverable in time $O_{\ell,\eps}(n^{O(1)}) + O(n \cdot \log(n)^{O_{\ell,\eps}(1)})$ with 
\[ |\Sigma| = (1 + \ell)^{O(1/\eps^6)} = O_{\ell, \eps}(1) \]
and
\[ |L| = O_{\ell, \eps}(1). \]
\end{corollary}
\begin{proof}
The proof follows by applying Theorem~\ref{thm:frs-list-rec-const-list} in Theorem~\ref{thm:alphabet-reduction}.
Let the outer code  $C_1 \subseteq \Sigma_1^{n_1}$ be a code of rate $1 - \eps/4$ which is $(\alpha, \ell_1, L)$-list-recoverable in time $\poly(n_1 \cdot O_{\ell,\eps}(1))$ for $\alpha = \Omega(\eps^2)$, and which has $|\Sigma_1| = \poly_{\ell, \eps}(\log(n_1))$, as guaranteed by Corollary~\ref{cor:frs-concat}.

For the inner code, we use a random linear code, choosing $|\Sigma_0| = (1 + \ell)^{O(1/\eps)}$.  By Theorem~\ref{thm:rlc}, for any $n_0$, 
there exists a linear code $C_0 \subseteq \Sigma_0^{n_0}$ of rate $R + \eps/2$ that is $(1 - R - \eps, \ell, \ell_1)$-list-recoverable, for
\[ \ell_1 = \left( \frac{\ell}{\eps} \right)^{ O(\log(\ell)/\eps^4 )} = O_{\ell,\eps}(1). \]
We have $|C_0| = (1 + \ell)^{O(n_0/\eps)}$, so there is a choice of
$n_0 = O_{\ell, \eps}(\log\log(n_1))$ so that $|C_0| \geq |\Sigma_1| = \poly_{\ell, \eps}(\log(n_1))$, and we make this choice.  
Thus, we may use $C_0$ and $C_1$ in the construction in Theorem~\ref{thm:alphabet-reduction} to construct a code $C \subseteq \Sigma^N$ of rate $(1 - \eps/4)(R + \eps/2) \geq R$ that is $(1 - R - 4\eps, \ell, L)$-list-recoverable, where $N = O_\eps( n_0 \cdot n_1 )$.  The final alphabet size is $|\Sigma| = |\Sigma_0|^{O(1/\eps^3 \alpha)} = (1 + \ell)^{1/\eps^6}$.

Finally, we consider how long it takes to construct and decode $C$.
To construct $C_0$ we iterate over all possible generator matrices and verify their list-recovery properties.   
There are at most $|\Sigma_0|^{R n_0^2} = |C_0|^{\eps \log_\ell(|C_0|)}$ linear codes, and checking the list-recoverability of any one of them takes time $O(|\Sigma_0|^{\ell n_0} \cdot |C_0|^{\ell_1+1} \cdot n_0 \ell)$, the time to search over all lists $S_1,\ldots,S_{n_0}$ of size $\ell$, and all subsets of $L+1$ codewords and compute their distance.  Since $n_0$ will be much larger than $\ell, \eps$, this is dominated by the $|\Sigma_0|^{R n_0^2}$ term, and the time it takes to find the generator matrix of such a code is
\[ T_{construct}(C_0) = O(|\Sigma_0|^{R \cdot n_0^2} ) = O( |C_0|^{n_0} ) = \log(n_1)^{O_{\ell, \eps}(\log\log(n_1))} = \poly_{\ell, \eps}(n_1) = \poly_{\ell, \eps}(N).  \]
Thus time to construct the whole code $C \subseteq \Sigma^N$ is also $\poly(N)$.

The time to perform list-recovery on the inner code $C_0$ by brute force is $\poly(|C_0|) = \poly_{\ell, \eps}(\log(n_1))$.
The time to perform list-recovery is then the time to run the list-recovery algorithm for $C_1$ (which is $\poly(N, O_{\ell,\eps}(1))$), plus the time to brute-force decode $C_0$ $n_1$ times, which is $O(N \cdot \log(N)^{O_{\ell,\eps}(1)})$. 

Instantiating Theorem~\ref{thm:alphabet-reduction} with these choices yields the corollary. 
\end{proof}

\begin{remark}
The reason to concatenate folded RS codes with themselves to obtain the outer code $C_1$ above is to make the alphabet size small enough that a brute-force search over all generator matrices for the inner code is still polynomial time.  If one omits this step, then the construction above still works with a quasipolynomial-time construction and a better list size.  It may be possible to create a version of Corollary~\ref{cor:FRS_AEL} which has a significantly smaller list size (close to the one guaranteed by Theorem~\ref{thm:frs-list-rec-const-list}) by using a folded RS code as $C_1$ and a derandomization of existing Monte-Carlo constructions of capacity-achieving list-recoverable codes as the inner code $C_0$.  However, since a list size of $O_{\ell, \eps}(1)$ is sufficient for our applications going forward, we stick with the simpler machinery.
\end{remark}

One might try to prove Theorem~\ref{thm:multimult-smallalpha} in the same way, with multivariate multiplicity codes as $C_1$ and a random linear code as $C_0$.  However, in this case the alphabet size is so large that doing exhaustive search to decode $C_0$ would yield a super-polynomial decoding time, and concatenating the outer code with smaller versions of itself until the alphabet size is smaller will yield too large a list size.  Therefore, we instead use for $C_0$ the code we have just created in Corollary~\ref{cor:FRS_AEL} instead. 

There is one more catch, which is that the codes from Theorem~\ref{thm:main-multimult} don't meet list-decoding capacity, since they have rate $1 - \eps$, but can only handle up to $O(\eps^2/\ell m^2)$ fraction of errors, which in our parameter regime is sub-constant.  If we applied Theorem~\ref{thm:alphabet-reduction} directly, we would need to take $\gamma$ in that theorem to be sub-constant, which would result in a super-constant alphabet size.  Thus, before we apply Theorem~\ref{thm:alphabet-reduction} to reduce the alphabet size, we amplify the distance to a constant, by applying a different version of the expander-based argument stated in Lemma~\ref{lem:GKORS} below. This will very slightly increase the alphabet size, but not so much that it will affect the asymptotics, and then we can apply Theorem~\ref{thm:alphabet-reduction}.
\begin{remark}
We believe it is possible to combine the two expander-based constructions into only one (with only one expander), which would give a slight improvement in the parameters.  However, our approach here (using both Lemma~\ref{lem:GKORS} and Theorem~\ref{thm:alphabet-reduction} in serial) is more modular and still yields the desired asymptotic result, so we stick with it for simplicity of exposition.
\end{remark}
We use the following lemma from \cite{GKORS17}.
\begin{lemma}[\cite{GKORS17}, Distance amplification for local list-recovery]\label{lem:GKORS}
For any constants $\delta_1, \alpha_1, \gamma > 0$, there exists an integer $d \leq \poly( 1/\delta_1, 1/\alpha_1, 1/\gamma)$ so that the following holds.  
\begin{itemize}
	\item Let $C_1 \subseteq \Sigma_1^{n_1}$ have rate $R_1$ and and distance $\delta_1$ and be $(t, \alpha_1, \ell_1, L)$-locally-list-recoverable in time $T(C_1)$.
	\item Let $C_0 \subseteq \Sigma_0^{n_0}$ have rate $R_0$ and be $(\alpha_0, \ell, \ell_1)$-globally-list-recoverable in time $T(C_0)$.
	\item Further suppose that $n_0 \geq d, |\Sigma_1| = |C_0|$.
\end{itemize}
Then there exists a code $C \subseteq ( \Sigma_0^{n_0} )^{n_1}$ of block length $n_1$ over $\Sigma = \Sigma_0^{n_0}$ with rate $R_0 \cdot R_1$ that is $(t', \alpha_0 - \gamma, \ell, L)$-locally-list-recoverable for
\[ t' = t \cdot n_0^2 \cdot \log(n_0). \]
Moreover there is a local list-recovery algorithm for $C$ which runs in time $O(T(C_1)) + O( t \cdot T(C_0)) + \poly( t, n_0, \ell).$
Further, if both codes can be constructed in time $T_{construct}(C_0)$ and $T_{construct}(C_1)$ (in the sense that this is the time it take to generate a short description which suffices for polynomial-time encoding), then the final code $C$ can be constructed in time $O(T_{construct}(C_0) + T_{construct}(C_1))$. 
\end{lemma}
\begin{remark}
The statement of this lemma in \cite{GKORS17} is slightly different in that both the hypotheses and the conclusion are slightly stronger.  They both specify linear codes, and the both have an additional ``soundness" parameter which we will not need.  However, an inspection of the proof shows that it goes through if these additional requirements and conclusions are dropped.
\end{remark}

\begin{corollary}\label{cor:multimult-AEL}
Let $\eps > 0$ be sufficiently small, and let $m, \ell > 0$ be any integers.
Suppose that $\eps, \ell$ are constant, and $m$ is growing.
Then there is a code $C \subseteq \Sigma^N$ with rate at least $ 1- 3\eps$ so that $C$ is $(t, \eps, \ell, L)$-locally-list-recoverable for 
\[ L = m^{O_{\ell, \eps}(m^2)} \]
\[ N = m^{O_{\ell, \eps}(m^4)} \]
\[ t = m^{O_{\ell, \eps}(m^3)} \]
\[ |\Sigma| = m^{O_{\ell, \eps}(m^{2m + 2})}\]
Further, $C$ can be locally list recovered in time $\poly(t)$, and can be deterministically encoded in time $\poly(N)$.
\end{corollary}
\begin{proof}
We instantiate Lemma~\ref{lem:GKORS} using the code from Corollary~\ref{cor:FRS_AEL} with rate $1- 2\eps$ as $C_0$ and a multivariate multiplicity code, from Theorem~\ref{thm:main-multimult}, as $C_1$.
Thus, $C_0 \subseteq \Sigma_0^{n_0}$ has rate $R_0 = 1 - 2\eps$ and is $(1 - \eps, \ell, \ell_1)$-globally list-recoverable in time $\poly_{\ell, \eps}(n_0)$, with $\ell_1 = O_{\ell, \eps}(1)$.
Meanwhile, $C_1 \subseteq \Sigma_1^{n_1}$ is $(\alpha, \ell_1, L)$-list-recoverable with
\[ \alpha = O_{\ell, \eps}(1/m^2) , \]
\[ L = m^{O_{\ell_1, \eps}(m^2)} = m^{O_{\ell,\eps}(m^2)},\]
\[ N = m^{O_{\ell_1, \eps}(m^4)} = m^{O_{\ell, \eps}(m^4)},\]
\[ t = m^{O_{\ell_1, \eps}(m^3)} = m^{O_{\ell, \eps}(m^3)}, \]
\[ |\Sigma| = m^{O_{\ell_1, \eps}(m^{2m + 2})},  = m^{O_{\ell, \eps}(m^{2m + 2})},\]
where we have used the fact that $\ell_1 = O_{\ell, \eps}(1)$ to turn $O_{\ell_1,\eps}(\cdot)$ into $O_{\ell, \eps}(\cdot)$.
The fact that $|C_0| = |\Sigma_1|$ thus implies that
\[ n_0 = O(\log|\Sigma_1|) = O_{\ell, \eps}( m^{2m + 2}\log(m) ). \]
Observe that $n_0$ is much larger than the $\poly( 1/\alpha, 1/\eps, 1/\delta_1 ) = \poly_{\ell, \eps}(m)$ that is require by Lemma~\ref{lem:GKORS}, for sufficiently large $m$.
Now we check the conclusions.  We immediately have the desired expressions for $N = n_1$ and $L$.  The final alphabet size is $|\Sigma| = |\Sigma_1|^{1/R_0} = \poly(|\Sigma_1|) = m^{O_{\ell, \eps}(m^3)}$, as before.  The query complexity is
\[ t' = t \cdot n_0^2 \log(n_0) = m^{O_{\ell, \eps}(m^3)} \cdot m^{O_{\ell, \eps}(m)} = m^{O_{\ell, \eps}(m^3)}. \]
The dominating term in the local list-recovery time is $\poly(t)$, so the running time is still $\poly(t)$.  And finally the time to construct a generator matrix for the inner code is $\poly(n_0) = \poly(N)$.
Thus we may treat $C$ as a deterministic code whose encoding map performs the search for $C_0$ in polynomial time and then encodes the message in polynomial time.
\end{proof}

Finally we are ready to prove Theorem~\ref{thm:multimult-smallalpha}.
\begin{proof}[Proof of Theorem~\ref{thm:multimult-smallalpha}]
Again we will use Theorem~\ref{thm:alphabet-reduction}.

Let $C_0 \subseteq \Sigma_0^{n_0}$ be the code from Corollary~\ref{cor:FRS_AEL}, with rate $R + 4\eps$.  Thus, we have
\[ |\Sigma_0| = (1 + \ell)^{O(1/\eps^6)},\]
and $C_0$ is $(1 - R - 5\eps, \ell, \ell_1)$-list-recoverable in time $\poly_{\ell, \eps}(n_0)$, for $\ell_1 = O_{\ell, \eps}(1)$.  Moreover, $C_0$ can be constructed in time $n_0^{O_{\ell,\eps}(\log(n_0))}$.

We choose $C_1 \subseteq \Sigma_1^{n_1}$ to be the code from Corollary~\ref{cor:multimult-AEL}, so that $C_1$ is a code of rate $(1 - 3\eps)$, which is $(t, \eps, \ell_1, L)$-list-recoverable, for
\[ L = m^{O_{\ell, \eps}(m^2)} \]
\[ n_1 = m^{O_{\ell, \eps}(m^4)} \]
\[ t = m^{O_{\ell, \eps}(m^3)} \]
\[ |\Sigma_1| = m^{O_{\ell, \eps}(m^{2m + 2})}\]
where as in the proof of Corollary~\ref{cor:multimult-AEL}, above we have used the fact that
$ \ell_1 = O_{\ell, \eps}(1) $ 
to hide dependence on $\ell_1$ in the notation $O_{\ell, \eps}(\cdot)$.  

Now we apply Theorem~\ref{thm:alphabet-reduction}, which concludes that there exists a code $C \subset \Sigma^N$ of rate $(1 - 3 \eps)\cdot (R + 4\eps) \geq R$ which is $(t', 1 - R - 9\eps, \ell, L)$-locally list-recoverable in time
\[ O_{\ell,\eps}( T(C_0) + T(C_1) ) = \poly_{\ell,\eps}(n_0) + m^{O_{\ell,\eps}(m^3)} =  m^{O_{\ell,\eps}(m)} + m^{O_{\ell, \eps}(m^3)} = m^{O_{\ell,\eps}(m^3)} = \poly(t), \]
where
\[ t'= O_{\ell, \eps}( t \log|\Sigma_1| ) = m^{O_{\ell, \eps}(m^3)}. \]
Moreover, we have 
\[ N = O_{\ell, \eps}(n_1 \cdot \log|\Sigma_1|) = m^{O_{\ell,\eps}(m^3)} \]
and
\[ |\Sigma| = |\Sigma_0|^{O(1/\eps^4)} = (1 + \ell)^{O(1/\eps^{10})} = O_{\ell, \eps}(1). \]

Finally, to deterministically encode a message in $C$, we run the construction algorithm for $C_0$ (in time $n_0^{O(\log(n_0))} = m^{O(m^2 \log(m))} = \poly(N)$), and then use the polynomial-time encoding algorithm for $C_1$ which exists from Corollary~\ref{cor:multimult-AEL}.

Now applying the proof above with $\eps/9$ instead of $\eps$ gives the theorem statement.
\end{proof}

Finally, we prove Theorem~\ref{thm:multimult-smallalpha2}.
\begin{proof}[Proof of Theorem~\ref{thm:multimult-smallalpha2}]
Choose $\delta, \tau, \eps > 0$ constant and sufficiently small, and $\ell > 0$ constant.
Let $C_0 \subseteq \Sigma_0^{n_0}$ be a code from Corollary~\ref{cor:FRS_AEL}, so that the rate of $C_0$ is $R + 2\eps$,  $|\Sigma_0| = (1 + \ell)^{O(1/\eps^6)} = O_{\eps, \ell}(1)$, and so that $C_0$ is $(1 - R - 3\eps, \ell, \ell_1)$-list-recoverable in time $\poly_{\ell, \eps}(n_0)$.

Now choose $C_1 \subseteq \Sigma_1^{n_1}$ which is $(t, \alpha, \ell_1, L)$-locally list-recoverable in time $\poly(t)$ for $\alpha = \Omega_{\eps, \ell, \tau, \ell}(1)$, so that $C_1$ has rate $1 - \eps$, query complexity $t = O_{\ell, \delta, \tau, \eps}(n_1^{\tau/2})$, $L = O_{\ell, \delta, \tau, \eps}(1)$, and $|\Sigma_1| = \poly_{\ell, \delta, \tau, \eps}(N)$.  Such a code exists by Theorem~\ref{thm:main-multimult2} (where we have used the fact that $\ell_1 = O_{\eps, \ell}(1)$).

Now we apply Theorem~\ref{thm:alphabet-reduction}, and conclude that there exists a code $C \subseteq \Sigma^N$, so that
the rate of $C$ is at least $(1 - \eps)(R + 2\eps) \geq R$ and so that 
\[ |\Sigma| = |\Sigma_0|^{O(1/(\eps^3 \alpha))} = O_{\ell, \delta, \tau,\eps}(1), \]
so that $C$ is
$(t', 1 -R - 7\eps, \ell, L)$-list-recoverable for
\[ t' = O\left( \frac{ t \log|\Sigma_1| }{ (R + 2\eps) \log|\Sigma_0| } \right) = O_{\ell, \delta, \tau, \eps}( N^{\tau/2} \log(N) ) = O_{\ell, \delta, \tau, \eps}(N^\tau), \]
in time 
\[ O_{\ell, \delta, \tau, \eps}(T(C_0) + T'(C_1)) = O_{\ell, \delta, \tau, \eps}( N^{O(\tau)} ). \]
Moreover, $C$ has rate $(1 - \eps)(R + 2\eps) \geq R$.  The stated result follows by replacing $\eps$ with $\eps/7$ in the above analysis.
\end{proof}

\section{Conclusion}
We have shown that folded Reed-Solomon codes and multiplicity codes perform better than previously known in the context of (local) list-recovery.  In addition to improving our knowledge about these codes, our results also lead to new and improved constructions of locally-list-recoverable codes.
However, there is still much left to do, and we conclude with some open questions.

\begin{enumerate}
\item Theorem~\ref{thm:frs-list-rec-const-list} shows that the list size for folded Reed-Solomon codes is $(\ell/\eps)^{O\left( \frac{1}{\eps} \log(\ell/\eps) \right)}$.  However, it is known that it is possible for codes to achieve a list size of $O(\ell/\eps)$.  It would be very interesting to strengthen our result to this bound, or even to reduce the list size to $\poly(\ell, 1/\eps)$.
\item It would be very interesting to improve the list size in Theorem~\ref{thm:full-unimult-list-rec} on univariate multiplicity codes with large $d$ to be $s^{O_{\ell,1/\eps}(1)}$, rather than the current bound of $\ell^{O(s\log(s))} \cdot s^{O(1)}$.  Beyond intrinsic interest, such an improvement would lead to an improvement in the query complexity of local list-recovery of multivariate multiplicity codes.
\item The algorithm given in Theorem~\ref{thm:frs-list-rec-const-list} is a randomized algorithm.  It is a very interesting open problem to design a deterministic list-decoding algorithm for folded RS codes with fixed polynomial running time that works up to list-decoding capacity.
\item We give a construction of a high-rate locally list-recoverable code with sub-polynomial query complexity.  But we do not know if this is the best we could do; for example, could one get away with polylogarithmic query complexity in the same setting?  Any lower bounds would be extremely interesting.
\end{enumerate}

\section*{Acknowledgements}
We would like to thank Atri Rudra and Venkatesan Guruswami for helpful discussions.

\bibliographystyle{alpha}
\bibliography{refs}

\appendix

\section{A Wronskian lemma}\label{app:wronskian}
In this section, we prove a lemma that shows that certain Wronskian determinants needed in Section~\ref{sec:Unimult} are nonzero.

\begin{lemma}
\label{lem:wronskmodq}
Suppose $q$ is prime.
Suppose $f_1, \ldots, f_t \in \F_q[X]$ are such that
$\deg(f_i) \not\equiv \deg(f_j) \mod q$ for $i \neq j$.
Then $W(f_1, \ldots, f_t)$ is nonsingular.
\end{lemma}
\begin{proof}
Let $d_j = \deg(f_j)$.

For $0 \leq i \leq t-1$ and $j \in [t]$,
the $(i,j)$ entry of the Wronskian matrix
is $f_j^{(i)}(X)$, whose leading term is ${d_j \choose i} X^{d_j - i}$.

Thus the determinant of the Wronskian matrix has degree
at most $D = \sum_{j=1}^{t} d_j - \sum_{i=0}^{t-1} i$.
Furthermore, the coefficient of $X^{D}$ in the Wronskian determinant
equals the determinant of the matrix whose $(i,j)$ entry (for $0 \leq i < t$, $1 \leq j \leq t$)
equals ${d_j \choose i}$. This latter determinant is essentially
a Vandermonde determinant, and by our hypotheses on the $d_j$, is nonzero in $\F_q$.
Thus the determinant of the Wronskian matrix is a polynomial of degree
exactly $D$, and in particular is nonzero.
\end{proof}

\section{Proof of Theorem~\ref{thm:linlistdec-mult}}\label{app:pflinlistdec-mult}

Here we prove Theorem~\ref{thm:linlistdec-mult}, which adapts a theorem of~\cite{GW13} to our setting.
\begin{proof}[Proof of Theorem~\ref{thm:linlistdec-mult}]
We begin by giving the algorithm.

\medskip
\putinbox{
\textbf{Algorithm} \textsf{FindPolys}.
\begin{itemize}
\item \textbf{INPUT:} A parameter $r$, and access to $S:\F_q \to { \F_q^s \choose \ell }$
\item \textbf{OUTPUT:} An affine subspace $v_0 + V$ that contains all polynomials that are close to $S$. 
\end{itemize}
\begin{enumerate}
\item Set $D = (s-r+1)(1-\alpha)q - 1$.
\item By solving a linear system of equations over $\F_q$,
find a nonzero $(A(X), B_0(X), \ldots, B_{r}(X)) \in (\F_q[X])^{r+1}$
such that:
\begin{enumerate}
\item $\deg(A) \leq D$, and for all $i$, $\deg(B_i) \leq D-d$.
\item For each $\lambda$ with $0 \leq \lambda \leq s-r$, for each $x \in \F_q$, for each
$y \in S(x)$:
$$ A(x) + \sum_{i=0}^{r-1} \sum_{j=0}^{\lambda} { i + j \choose i} y^{(i+j)} B^{(\lambda - j)}(x) = 0.$$
\end{enumerate}
\item Let $ v_0 + V $ be the affine space
$$ v_0 + V = \{ f(X) \mid A(X) + \sum_{i=0}^{r-1} f^{(i)}(X) B_i(X) = 0\}.$$
\item Output $v_0 + V$.
\end{enumerate}
}
We need to show:
\begin{enumerate}
\item The linear system has a nonzero solution,
\item $\mathcal L \subseteq v_0 + V$,
\end{enumerate}

To see that the linear system has a nonzero solution, we
show that the homogeneous system of linear equations
in Step 2 of the algorithm has more variables than constraints. 
The total number of free coefficients in 
$A(X), B_0(X), \ldots, B_r(X)$ equals:
\begin{align*}
(D+1) + r(D-d+1) &= (D+1)(r+1) - d\cdot r\\
&= (s-r+1)(1-\alpha)q (r+1) - d r\\
&>  (s-r+1)\left(\frac{\ell}{r+1} + \frac{r}{r+1} \frac{d}{(s-r+1)q}\right)q (r+1) - d r\\
&= (s-r+1) \ell q + dr - dr\\
&= (s-r+1) \ell q.
\end{align*}
The total number of constraints equals:
$$q \cdot (s-r+1) \cdot \ell.$$
By choice of $D$, the number of free coefficients is larger
than the number of constraints. This proves that the algorithm
can find a nonzero solution in Step 2.

Now take any $g(X) \in \mathcal L$. We will show that $g(X)$
is an element of the affine space $v_0 + V$ that is output by the algorithm.

Define $Q(X) = A(X) + \sum_{i=0}^{r-1} g^{(i)}(X) B_i(X)$.
Observe that $\deg(Q) \leq D$. 

Now take any $x \in \F_q$ and $y \in S(x)$ for which
\begin{align}
\label{eq:gagreer}
g^{(<s)}(x) =  y.
\end{align}
Let $\lambda$ be an integer with $0 \leq \lambda \leq s-r$. 
Then by the chain rule for Hasse derivatives:
\begin{align*}
Q^{(\lambda)}(x) &= A^{(\lambda)}(x) + \sum_{i=0}^{r-1} \left(g^{(i)} \cdot B_i\right)^{(\lambda)}(x) \\
&= A^{(\lambda)}(x) + \sum_{i=0}^{r-1} \sum_{j=0}^{\lambda} \left(g^{(i)}\right)^{(j)}(x) B^{(\lambda - j)}(x) \\
&= A^{(\lambda)}(x) + \sum_{i=0}^{r-1} \sum_{j=0}^{\lambda} { i + j \choose i} g^{(i+j)}(x) B^{(\lambda - j)}(x) \\
&= A^{(\lambda)}(x) + \sum_{i=0}^{r-1} \sum_{j=0}^{\lambda} { i + j \choose i} y^{(i+j)} B^{(\lambda - j)}(x) \\
&= 0
\end{align*}
Since this holds for every $\lambda$ with $0 \leq \lambda \leq s-r$, we get that:
$$\mult(Q, x) \geq s-r+1.$$

By assumption on $g$, there are at least $(1-\alpha) q$ values of $x\in \F_q$ such that
there exists some $y \in S(x)$ for which Equation~\eqref{eq:gagreer} holds.
Thus there are at least $(1-\alpha) q$ points where $Q$ vanishes with multiplicity at least
$s-r+1$. Since $\deg(Q) \leq D < (s-r+1)(1-\alpha)q$, we conclude that $Q(X) = 0$.

By definition of $Q(X)$ and $v_0 + V$, this
implies that $g(X) \in v_0 + V$, as desired.
\end{proof}

\section{Coordinate restrictions of subspaces}

In this section, we discuss the relationship between the list-recovery results
for  Folded Reed-Solomon codes and univariate multiplicity codes, as well as
the ideas that go into their proofs.

First we point out that we could use Lemma~\ref{lem:general-subspace-design} to
analyze Algorithm $\ALGPRUNE$ and Algorithm $\ALGPRUNEMULT$.
 This can be used in place of Theorem~\ref{thm:subspace-design} and Theorem~\ref{thm:cool-subspace-design}, and would give a proof of the
list recoverability of Folded Reed-Solomon and univariate multiplicity codes {\em from very small error}.
However, this approach is not able to reproduce the capacity achieving list-decodability in Theorem~\ref{thm:frs-list-rec-const-list}
and Theorem~\ref{thm:unimult-list-rec-const-list-smalld}, and only gives a quantiatively weaker version of 
Theorem~\ref{thm:full-unimult-list-rec}\footnote{The version proved using
Lemma~\ref{lem:general-subspace-design} requires $\alpha < \frac{1}{s^2}$,
and gives an output list-size of $\ell^{O(s^2 \log s)}$. This
in turn would be sufficient to give a
local list-recovery algorithm
for length $N$ multivariate multiplicity codes with query complexity
$\exp(\log^{5/6}(N))$.}.

The statement of Theorem~\ref{thm:full-unimult-list-rec} on list-recovery of whole field univariate multiplicity codes is
noticeably weaker than the statement of Theorem~\ref{thm:frs-list-rec-const-list} on list-recovery of
Folded Reed-Solomon codes: the former only gives list-recoverability in the presence of very few errors. The
proof of the former is also noticeably more involved.
Inspecting the components of the proofs, we see that this difference
arises from the significantly different quantitative natures
of the analyses of algorithm $\ALGPRUNE$ and $\ALGPRUNEMULT$.

The following example shows that this difference is not
just an artifact of the analysis: there are instances where
the algorithm $\ALGPRUNEMULT$ (which is exactly analogous
to the algorithm $\ALGPRUNE$) really requires the error fraction
to be very small, and produces an output list size
which is exponentially large in $s$.
\begin{example}
Let $V \subseteq \F_q[X]$
be given by:
$$ V = \{ \sum_{i < d/q} a_i X^{iq} \mid a_i \in \F_q \}.$$

Let $\tau$ be any integer $< d/q - 1$,
and let $b_1, \ldots, b_\tau \in \F_q$ be distinct.
Then:
$$ V \cap \bigcap_{j = 1}^{\tau} H_{b_j} = \left\{ \prod_{j=1}^{\tau} (X - b_j)^{q} \cdot \left(\sum_{i< d/q-\tau} c_i X^{iq}\right) \mid c_i \in \F_q \right\}.$$
In particular:
$$ \dim (V \cap H_{b_1} \cap H_{b_2} \cap \ldots \cap H_{b_\tau}) = \dim(V) - \tau \geq 1.$$
This means that when we run the algorithm $\ALGPRUNEMULT$ on $V$ as input,
the step where we search for $P(X) \in v_0 + V$ will NEVER find a unique solution.

Thus for the algorithm $\ALGPRUNEMULT$ to succeed with positive probability
we must have $\tau \geq d/q$. For the constant rate setting, this means that $\tau = \Omega(s)$.
With $\tau = \Omega(s)$, the success probability of $\ALGPRUNEMULT$ is at most $(1- \alpha)^{\Omega(s)}$,
and the output list-size is at least $\ell^{\Omega(s)}$.
Thus the analysis of the algorithm $\ALGPRUNEMULT$ in Lemma~\ref{lem:algprune-mult} cannot be improved.
\end{example}

\section{List recovering Reed-Muller codes on product sets}
\label{app:rm}
In this section, we prove Lemma~\ref{lem:RM-vectors} about list-recovery of Reed-Muller codes that we need for $\RecCand$ and Lemma~\ref{lem:recovercandidates}.

%
%

We will prove Lemma~\ref{lem:RM-vectors} by reducing the case of tuples of polynomials to the case of a single polynomial over a large field.  We go through the details in the next two subsections.

\subsection{Replacing vector values with big field values}
\label{subsec:vector-to-field}

Let $\F$ be a finite field, and let $\K$ be the degree $t$ field
extension of $\F$.
Let $\phi : \F^t \to \K$ be an arbitrary $\F$-linear bijection.
Let $U \subseteq \F$ (and thus $U \subseteq \K$).

Then to every function $f: U^m \to \F^t$, we can associate a
function $\tilde{f}: U^m \to \K$, where $\tilde{f} = \phi \circ f$. 
This identifies the underlying Hamming metric spaces.
The key observation is that under this identification, 
$f$ is the evaluation table of a tuple of $t$ polynomials in $\F[X_1, \ldots, X_m]$ of degree $\leq d$ if and only if $\tilde{f}: U^m \to \K$ is the evaluation table
of a degree $\leq d$ polynomial in $\K[X_1, \ldots, X_m]$.
Thus questions about decoding (list-decoding, list-recovering) vector valued polynomial codes
reduce to questions about decoding (list-decoding, list-recovering) scalar valued polynomial codes
over larger fields.

Through this connection, Lemma~\ref{lem:RM-vectors} is a consequence of
the following lemma (and the fact that $\K$ can be constructed in randomized 
$\poly(\log |\K|)$ time).

\begin{lemma}[Reed-Muller list recovery on a grid]
\label{lem:RMlist}
Let $\ell, \tilds, K, m$ be given parameters. Let $\K$ be a finite field. Suppose that $U \subseteq \K$ and $|U| \geq 2 {\ell}\tilds K$.
Let $\alpha < 1 - \frac{1}{\sqrt{K}}$ be a parameter.

Then for every $f: U^m \to {\K \choose \ell}$,
if 
$$ \calL = \{ Q(Y_1, \ldots, Y_m) \in \K[Y_1, \ldots, Y_m] \mid \deg(Q) \leq \tilds \mbox{ and } \Pr_{\bu \in U^m}[Q(\bu) \not\in f(\bu)] < \alpha \},$$
we have:
\begin{enumerate}
\item  $|\calL| \leq 2K\ell.$

\item Suppose further that $K > m^2$ and  $\alpha \leq 1- \frac{m}{\sqrt{K}}$, then  there is a $\poly(|U|^m, \log{|K|})$ time algorithm to compute $\calL$. 
\end{enumerate}
\end{lemma}

\subsection{Proof of Lemma~\ref{lem:RMlist}}

\begin{proof}[Proof of Lemma~\ref{lem:RMlist}]
Item 1 will follow from the Johnson bound for list-recovery, which is a general statement implying
good list recoverability for codes with large distance. Specifically, Lemma V.2 in \cite{GKORS17} (see also Corollary 3.7 in \cite{Guruswami-Thesis})
states that a code of distance $\delta$ is $(\alpha, \ell, L)$ list-recoverable for $\alpha > 1 - \sqrt{\ell(1-\delta)}$
and $L = \frac{\ell}{(1-\alpha)^2 - \ell (1-\delta)}$. In our setting, the code of polynomials of degree at most $\tilds$ on 
$U^m$ has distance $\delta$ at least $1 - \frac{\tilds}{|U|} \geq 1 - \frac{1}{2K\ell}$. 
Thus for $\alpha < 1 - \frac{1}{\sqrt{K}}$, the Johnson bound for list-recovery implies that
we can take $L = 2K\ell$, as desired.

We prove Item 2 by induction on $m$.
The $m=1$ case is simply the Sudan list-recovery algorithm~\cite{Sudan} for Reed-Solomon codes, which works with the claimed parameters (since the total number of points $n = \ell \cdot |U|$, the number of agreement points
$A$ is at least $(1-\alpha) n \geq \frac{n}{\sqrt{K}}$, and so
$A \geq 2\sqrt{n \tilds}$, which is the requirement for the Sudan algorithm to work).

For general $m$, we first do list-recovery on $m-1$ dimensional grids, and then combine
the results using list-recovery for vector-valued univariate polynomials. We crucially use the previous combinatorial bound
on the list size to ensure that the intermediate list size is under control (as the recursion unfolds).

More concretely, the algorithm proceeds as follows:
\begin{enumerate}
\item First, for each setting of $u \in U$, we consider
the received word $f_u : U^{m-1} \to {\K \choose \ell}$,
given by $f_u(\by) = f(\by, u)$.

Now list-recover $f_u$ to radius $\beta = \left(1- \frac{m-1}{\sqrt{K}}\right)$ to find the
set of nearby $m-1$-variate polynomials $\calL_u$. By the previous combinatorial bound,
we can assume (after confirming that all elements of $\calL_u$ are indeed close to $f_u$) 
that:
$$|\calL_u| \leq 2 K \ell.$$

\item Next we combine all $\calL_u$. Let $M_1(Y_1, \ldots, Y_{m-1}), \ldots, M_t(Y_1, \ldots, Y_{m-1})$ be all the $(m-1)$-variate monomials of total
degree at most $\tilds$. Define a function $$g: U \to {\K^{t} \choose 2K\ell }$$
as follows: for each $u \in U$ and each element $P$ of $\calL_u$, include the vector of coefficients of
$P$ into $g(u)$.

Then, using a vector-valued Sudan list-recovery algorithm for univariate polynomials (obtained from
the standard scalar-valued Sudan list-recovery algorithm via the connection in Section~\ref{subsec:vector-to-field}),
we find all tuples of univariate polynomials $\mathbf P(Z) = (P_1(Z), \ldots, P_t(Z)) \in (\K[Z])^t$ such that:
$$ \Pr_{u \in U} [ \mathbf P(u) \not\in g(u) ] < \gamma,$$
where $\gamma = 1 - \frac{1}{\sqrt{K}}$.
\item For each $\mathbf P(Z) \in (\K[Z])^t$ found in the previous step,
we construct the polynomial:
$$R(Y_1, \ldots, Y_{m-1}, Y_m) =  \sum_{i=1}^t M_i(Y_1, \ldots, Y_{m-1}) P_i(Y_m).$$
If this polynomial has total degree at most $\tilds$ and is
$\alpha$-close to $f$, then we include it in the output list.
\end{enumerate}

To prove correctness of this algorithm, consider any $Q(Y_1, \ldots, Y_m) \in \calL$.
Let $Q_u(Y_1, \ldots, Y_{m-1}) = Q(Y_1, \ldots, Y_{m-1}, u)$.
Then we have:
$$ \mathbb E_{u \in U} [ \dist( Q_u, f_u) ] = \dist(Q, f) < \alpha.$$
Thus 
$$ \Pr_{u \in U} [ \dist(Q_u, f_u) \geq \beta ] \leq \frac{\alpha}{\beta} < \frac{ 1 - m/\sqrt{K}}{1 - (m-1)/\sqrt{K}} \leq 1 - \frac{1}{\sqrt{K}} = \gamma.$$
This implies that for at most $\gamma$-fraction of $u \in U$, we have that $Q_u \not\in \calL_u$.

Write $Q(Y_1, \ldots, Y_m)$ as $\sum_{i=1}^t M_i(Y_1, \ldots, Y_{m-1}) G_i(Y_m)$.
Then the above discussion means that for at most $\gamma$ fraction of $u \in U$,
we have that $(G_1(u), \ldots, G_t(u)) \not\in g(u)$.
This implies that $(G_1(u), \ldots, G_t(u))$ will be included in the list returned by the 
univariate list-recovery algorithm in Step 2, and thus that $Q$ will be included in the output of the algorithm 
in Step 3.

This completes the proof of correctness.
The bound on the running time follows immediately from the description of the algorithm (using the fact that $t \leq (\tilds+m)^m \leq |U|^m$).
\end{proof}


\section{Proof of Theorem~\ref{thm:alphabet-reduction}}\label{app:AEL}
In this section, we prove Theorem~\ref{thm:alphabet-reduction}.  Our proof is based on a construction first attributed to~\cite{AEL95}, which has since been used in many works to improve the parameters of list-recoverable and locally list-recoverable codes.  We include the proof here for completeness.

\begin{proof}[Proof of Theorem~\ref{thm:alphabet-reduction}]
The construction uses three ingredients: a bipartite expander graph $G$; the code $C_1$ guaranteed in the problem statement; and an inner code $\cC_0$ as in the theorem statement.  
Notice that the size and rate of $C_0$ implies that $n_0 = \left\lceil\frac{ \log|\Sigma_1| }{ R\cdot \log|\Sigma_0|}\right\rceil$.  Choosing $|\Sigma_0| = (\max\{2,\ell\}^{O(1/\eps)}$, this reads
\[ n_0 = O \left( \frac{ \log|\Sigma_1| \cdot \eps }{R \cdot (1 + \log(\ell))} \right). \]

It is known that the double-cover of a Ramanujan graph has the properties we want; we state these properties formally in the following claim.
\begin{claim}\label{claim:existsG}[See~\cite{KMRS}, Lemma 2.7]
Let $\xi,\eps, R \in [0,1]$, so that $\xi$ and $\eps$ are sufficiently small. 
For infinitely many integers $N > 0$, there exists a $D = O(1/\xi \eps^2)$ so that the following holds.
There exists a bipartite expander graph 
 $G = (V_L, V_R, E)$ be a bipartite expander graph with $N$ vertices on each side, with degree $D$, and with the following property: for any set $Y \subseteq V_R$ of right-hand-vertices with $|Y| \geq (R + 4\eps)N$, we have
\[ |\{ v \in V_L \,:\, |\Gamma(v) \cap Y| < (R + 3\eps)D \} | \leq \xi N, \]
where $\Gamma(v) \subseteq V_R$ is the set of neighbors of $v$ in $G$.
\end{claim}

We will instantiate Claim~\ref{claim:existsG} with the $\eps$ from the guarantee in $C_0$, and with $\xi := \gamma \cdot \eps$.   Thus, we have $D = O(1/(\eps^3 \gamma))$.
With these ingredients $C_0$ and $G$ in hand, let $C_1$ be as in the theorem statement, and let $\bar{C} = C_0 \circ C_1$ be the concatenation of $C_0$ and $C_1$.  Thus, a codeword in $\bar{C}$ has the form
\[ \bar{c} = ( C_0(x_1), C_0(x_2), \ldots, C_0(x_{n_1}) ) \in (\Sigma_0^{n_0})^{n_1} \]
for $(x_1, \ldots, x_{n_1}) \in C_1$.  
Suppose without loss of generality that $D$ divides $n_0$. (Otherwise, we may pad the codewords of $\cC_0$ with zeros to make this be the case).
Then break up the codewords $\bar{c}$ into $N = n_0n_1/D$ blocks of length $D$:
\[ \bar{c} = ( y^{(1)}, y^{(2)}, \ldots, y^{(N)} ) \in (\Sigma_0^D)^N. \]
We will form our final code $C \subseteq (\Sigma_0^D)^N$ as follows:  
for each codeword $(y^{(1)}, \ldots, y^{(N)}) \in \bar{C}$
(thought of as an element of $(\Sigma_0^D)^N$), define a codeword $c = (c^{(1)}, \ldots, c^{(N)}) \in C \subseteq (\Sigma_0^D)^N$ by
\[ c_\ell^{(j)} = y_r^{(i)}, \]
where $j = \Gamma_r(V_L[i])$ and $i = \Gamma_\ell(V_R[j])$ and where the notation $\Gamma_r(v)$ denotes the $r$'th neighbor of vertex $v$ (according to some arbitrary order) and $V_L[i]$ denotes the $i$'th vertex in $V_L$ (again according to an arbitrary order).

The code $C \subseteq (\Sigma_0^D)^N$ will be the set of all codewords obtained this way.  Notice that the rate of $C$ is the same as that of $\bar{C}$, since the operation above just permutes the symbols of a codeword.  Thus, the rate of $C$ is
\[ (1 - \zeta) \cdot R, \]
as claimed.

\paragraph{Global list-recovery.}
We first argue that if $C_1$ is efficiently $(\eps, \ell_1, L)$-list-recoverable, then $C$ is efficiently $(1 - R - 4\eps, \ell, L)$-list-recoverable. 
Suppose that $S_1,\ldots, S_N \subset \Sigma_0^D$ have $|S_i| \leq \ell$, and suppose that $c \in C$ has
$c^{(i)} \in S_i$ for all $i \in Y$, for some set $Y \subseteq [N]$ of size at least $(R + 2\gamma)N$.
Suppose that $c$ is obtained as above from $(y^{(1)}, \ldots, y^{(N)}) \in \bar{C}$, which is obtained by concatenation from $(x_1,\ldots,x_{n_1}) \in C_1$.  Suppose that $z$ is the original message so that $C_1(z) = (x_1,\ldots, x_{n_1})$.  Thus, our goal is to recover a short list $S$ of size at most $L$, so that $z \in S$.

For $i \in [N]$ and $r \in [D]$, let 
\[ T_{i,r} = \{ \alpha \in \Sigma_0 \,: \, \exists \beta \in S_j \subseteq \Sigma_0^D, \beta_\ell = \alpha, j = \Gamma_r(V_L[i]), i = \Gamma_\ell(V_R[j]) \}. \]
That is, $T_{i,r}$ is the list of symbols in $\Sigma_0$ that $y^{(i)}_r$ could be that is consistent with the lists $S_1,\ldots,S_N$.

The decoding algorithm for $C$ is then straightforward: given $S_1,\ldots,S_N$, compute the lists $T_{i,r}$, and then run the list-recovery algorithm for $C_0$ on each block $C_0(x_t)$ to obtain a list $S'_t \subset \Sigma_1$ of possible values of $x_t$.  Since $\cC_0$ is obtained via a random coding argument, there is not an efficient algorithm for this; however, $C_0$ is small enough that the brute-force decoding algorithm will do.  
Next, we run the list-recovery algorithm for $C_1$ on the lists $S'_t$, to obtain our final list $S$ of size at most $L$.

Before we show that this is correct, consider the run-time of this algorithm.  
The dominating term in the running time is the time to run the list-recovery algorithms of $C_0$ and $C_1$.
The time to list-recover $C_1$ is given by $T(C_1)$, and the time to list-recover each of the $n_1$ copies of $C_0$ is bounded above by $O(|C_0|) = O(|\Sigma_1|)$. 
Together, these expressions give the runtime bound claimed in the theorem.

Next, we argue that this algorithm is correct.  Let $c \in C$ be as above, so that 
$c^{(i)} \in S_i$ for all $i \in Y$, for some set $Y \subseteq [N]$ of size at least $(R + 4\eps)N$.
By the expansion property of $G$, the set $S \subseteq [N]$ of indices $i$ so $V_L[i]$ has at most $(R + 3\eps)D$ neighbors in $Y$ has size $|W| \leq \xi N = \gamma \eps N$.  Thus, for each $i \not\in W$, $y^{(i)}_r \in T_{i,r}$ for at least $(R + 3\eps)D$ values of $r \in [D]$.

Now consider the blocks $\bar{c} = (C_0(x_1), \ldots, C_0(x_{n_1}))$.  Each $C_0(x_t)$ is made up of $n_0/D$ blocks $y^{(i)} \in \Sigma_0^D$.  By an averaging argument, since $|W| \leq \gamma \eps N$, at most an $\gamma$-fraction of the blocks $C_0(x_t)$ have more than an $\eps$-fraction of its constituent length-$D$ blocks in $W$.

Suppose that the block $C_0(x_t)$ is one of the $(1 - \gamma)$-fraction of the blocks for which this does not hold; that is, at most an $\eps$-fraction of the blocks $y^{(i)}$ in $C_0(x_t)$ have $i \in W$.   Then the number of symbols $y^{(i)}_r \in \Sigma_0$ that make up $C_0(x_t)$ so that $y^{(i)}_r \in T_{i,r}$ is at least
\[ | \{ (i,r) \,:\, 
r \in [D],
y^{(i)} \text{ is contained in } C_0(x_t),
\text{ and } y_r^{(i)} \in T_{i,r}
\}| \geq (R + 3\eps) \cdot D \cdot (1 - \eps) \frac{n_0}{D} \geq (R + \eps)n_0. \]
Since $C_0$ is $(1 - R - \eps, \ell, \ell_1)$-list-recoverable, this implies that for all such $t$, the list-recovery algorithm for $C_0$ returns a list $S_t'$ of length at most $\ell_1$ so that $x_t \in S_t'$.  

Now since there are at least $(1 - \gamma)n_1$ such blocks, the list-recovery algorithm for $C_1$ will return a list $S$ of size at most $L$, so that the original message $z$ is guaranteed to be contained in $S$.  Thus, the algorithm is correct.

\paragraph{Local list-recovery.}
Suppose that $C_1$ is $(t,\gamma,\ell_1,L)$-locally list-recoverable via an algorithm $A$, which expects advice $\xi \in [L]$.
Now we may use exactly the same construction as above to obtain a locally list-recoverable code.  

More precisely, for any $x \in C_1$ resulting in a codeword $c \in C$ that agrees with a $1 - R - 4\eps$ fraction of the lists $S_i$, the argument above shows that, for a $1 - \gamma$ fraction of the indices $i \in [n_1]$, we may obtain a list $S'_i \subseteq \Sigma_1$ of size at most $L$, so that $x_i \in S'_i$, using $n_0 = O( \log_{|\Sigma_0|}|\Sigma_1|/R )$ queries to the input lists $S_i$.  This immediately implies the following local list-recovery algorithm:

\medskip
\putinbox{
\textbf{Algorithm} \textsf{LocalListRecoveryAEL}
\begin{itemize}
\item \textbf{INPUT:} Query access to the lists $S_1, \ldots, S_N \subseteq \Sigma_0^D$, and an index $i \in [N]$
\item \textbf{ADVICE:} $\xi \in [L]$
\end{itemize}
\begin{enumerate}
\item For each $j \in \Gamma( V_R[i] )$:
	\begin{itemize}
	\item Let $r \in [n_1]$ be the index so that the block $y^{(j)}$ is contained in $C_0(x_r)$.
	\item Run $C_1$'s local list-recovery algorithm $A$ with advice $\xi$, simulating query access to the lists $S_s'$ by using $n_0$ queries to the lists $S_1,\ldots, S_N$, as described above.  This returns $x_r$.
	\item Use $C_0$'s encoder (which in this case we may treat as a look-up table) to obtain $C_0(x_r)$ and hence $y^{(j)}$. 
	\end{itemize}
\item Given $y^{(j)}$ for each $j \in \Gamma( V_R[i] )$, assemble $c^{(i)}$ and return it. 
\end{enumerate}
}
The correctness of this algorithm follows from the correctness of the query-simulation procedure, which was shown above. 

The query complexity is $D \cdot t \cdot n_0$, because for each of $D$ values of $j$, we need to simulate $t$ queries to the lists $S_r'$, each of which requires $n_0$ queries to the lists $S_r$.  Plugging in our settings of $D$ and $n_0$ gives the query complexity claimed in the theorem.

The running time for each simulated query to a list $S_r'$ is dominated by the time to correct $C_0$, which is $O(|C_0|) = O(|\Sigma_1|)$ by using a brute-force algorithm, plus the time used by $A$, which is $T(C_1)$.  There are $D= O(1/(\eps^3\gamma))$ such queries which gives the running time claimed in the theorem statement.
\end{proof}

\end{document}